\documentclass{article}

\usepackage[dvipsnames]{xcolor}
\usepackage[colorlinks=true,
            linkcolor=MidnightBlue,
            urlcolor=MidnightBlue,
            citecolor=MidnightBlue,
            filecolor=black]{hyperref}
            
\usepackage{amsmath, amsthm, amsfonts, bm, bbm}
\usepackage[authoryear,round]{natbib}
\usepackage{enumitem}
\usepackage{float}
\usepackage{xspace}
\usepackage{algorithm}
\usepackage{algpseudocode}
\usepackage[T1]{fontenc}
\usepackage{lmodern}

\usepackage{url}
\usepackage[english]{babel}
\usepackage{csquotes}
\usepackage{graphicx}
\usepackage{subcaption}
\usepackage{booktabs, multirow, siunitx, tablefootnote}
\usepackage[nameinlink]{cleveref}
\usepackage[letterpaper,top=2cm,bottom=2cm,left=3cm,right=3cm,marginparwidth=1.75cm]{geometry}

\usepackage{authblk}

\newcommand{\R}{\mathbb{R}}
\newcommand{\N}{\mathbb{N}}
\newcommand{\Z}{\mathbb{Z}}

\newcommand{\eps}{\varepsilon}
\newcommand{\parteps}{\varepsilon_{\!P}}
\newcommand{\cceps}{\varepsilon_{\!C}}
\newcommand{\bcceps}{\varepsilon_{\!B}}
\renewcommand{\epsilon}{\varepsilon}
\newcommand{\cM}{\mathcal{M}}

\newcommand{\calX}{\mathcal{X}}
\newcommand{\calC}{\mathcal{C}}
\newcommand{\calE}{\mathcal{E}}
\newcommand{\calY}{\mathcal{Y}}
\newcommand{\calZ}{\mathcal{Z}}

\newcommand{\calM}{\mathcal{M}}

\newcommand{\svt}{\textsf{SVT}\xspace}
\newcommand{\wrp}{\textsf{W}}
\newcommand{\init}{\textsc{Initialize}\xspace}

\newcommand{\chk}{\textsc{Check}\xspace}
\newcommand{\upd}{\textsc{Update}\xspace}

\newcommand{\cc}{\textsf{CC}\xspace}

\newcommand{\bcc}{\textsf{BCC}\xspace}

\newcommand{\ecc}{\textsf{ECC}\xspace}
\newcommand{\partt}{\textsf{Part}\xspace}
\newcommand{\simpart}{\textsf{SimPart}\xspace}
\newcommand{\simecc}{\textsf{SimECC}\xspace}
\newcommand{\basecc}{\textsf{BaseCC}\xspace}
\newcommand{\gausscc}{\textsf{GaussianCC}\xspace}
\newcommand{\partitioncc}{\textsf{PartitionCC}\xspace}
\newcommand{\Lap}{\mathrm{Lap}}
\newcommand{\DLap}{\mathrm{DLap}}
\newcommand{\SPACESUM}{\!\!\!\!\!\!\!\!\!\!\!\!}
\newcommand{\sched}{\mathsf{sched}}
\newcommand{\cdistinct}{\mathsf{CountDistinct}}

\newcommand{\ccpbits}[1]{\mathsf{Count}(\sim_p, \{0,1\}^{#1})}
\newcommand{\ccptrits}[1]{\mathsf{Count}(\sim_p, \{-1,0,1\}^{#1})}
\newcommand{\ccebits}[1]{\mathsf{Count}(\sim_e, \{0,1,\bot\}^{#1})}
\newcommand{\ccgeneric}[1]{\mathsf{Count}(\sim, \mathcal{S}^{#1})}
\newcommand{\cccustom}[2]{\mathsf{Count}(#1, #2)}
\newcommand{\indic}{\mathbbm{1}}

\newcommand{\slvar}{\mathrm{schedList}}
\newcommand{\intsum}{\mathsf{intsum}}
\newcommand{\postfunc}{\mathrm{Post}}

\newcommand{\MS}{\text{MS}}

\newcommand{\pmonetobin}{\mathrm{TritsToBits}}

\newcommand{\err}[2]{$#1_{\pm #2}$}

\usepackage{aliascnt}

\newtheorem{theorem}{Theorem}
\numberwithin{theorem}{section}

\newcommand{\newaliastheorem}[3]{%
  \newaliascnt{#1}{theorem}%
  \newtheorem{#1}[#1]{#2}%
  \aliascntresetthe{#1}%
  \crefname{#1}{#2}{#3}%
  \Crefname{#1}{#2}{#3}%
}

\newaliastheorem{lemma}{Lemma}{Lemmas}
\newaliastheorem{corollary}{Corollary}{Corollaries}
\newaliastheorem{definition}{Definition}{Definitions}
\newaliastheorem{assumption}{Assumption}{Assumptions}
\newaliastheorem{fact}{Fact}{Facts}
\newaliastheorem{claim}{Claim}{Claims}
\newaliastheorem{remark}{Remark}{Remarks}
\newaliastheorem{notation}{Notation}{Notations}
\newaliastheorem{observation}{Observation}{Observations}
\newaliastheorem{proposition}{Proposition}{Propositions}

\newtheorem*{theorem*}{Theorem}
\newtheorem*{lemma*}{Lemma}
\newtheorem*{corollary*}{Corollary}
\newtheorem*{observation*}{Observation}
\newtheorem*{question*}{Question}

\usepackage[nameinlink]{cleveref}

\title{Edit-Neighboring Data Streams and Privacy under Continual Observation}
\author{Joel Daniel Andersson}
\author{Anamay Chaturvedi}
\author{Monika Henzinger}
\author{Roodabeh Safavi}
\affil{Institute of Science and Technology Austria}

\date{}

\begin{document}

\maketitle

\begin{abstract}
        Differential privacy under Continual Observation (CO) quantifies the loss in privacy that occurs when outputs generated using a stream of sensitive input data are published in the online setting.
    In prior work, this formulation requires that any private mechanism when given as input two neighboring streams that differ in the value of at most one stream element must generate output streams that are almost indistinguishable.

    In this paper, we consider a more general notion of privacy wherein an individual's decision to participate in the data collection process may potentially \emph{shift} the entire stream by a time-step.
    We define a new notion of \emph{edit-neighboring} streams that captures this scenario.
    Our findings are as follows.

    First, we prove that on a stream of length $T$, no \emph{additive-noise mechanism} achieves additive error less than $\tilde{\Omega}(\min\{T^{1/3}/\eps^{2/3}, T\})$ when required to be $\eps$-DP under CO for edit-neighboring streams.
    In particular, this includes state-of-the-art continual counters constructed via the factorization mechanism that in the standard neighboring setting incur only polylogarithmic additive error.

    Second, we construct the first mechanisms with polylogarithmic additive error for our more stringent notion of privacy.
    We show that we can recover the same additive error as in the standard notion of privacy albeit with worse constant coefficients for both arbitrary input streams and sparse streams.

    Third, we show that the notion of edit-neighboring streams inhabits a `sweet-spot' in terms of generality and additive error incurred.
    More precisely, we show that the even more general notion of \emph{prefix-sum neighboring} streams---which arises naturally in reductions for problems under CO---must incur additive error scaling as $\tilde{\Omega}(\min\{T^{1/3}/\eps^{2/3}, T\})$ for any mechanism that is $\eps$-DP under continual observation.

    Finally, we show empirically on synthetic data that when compared with prior work, our mechanism achieves a superior trade-off between the success probability of a simple distinguishing attack, and the additive error incurred by the respective mechanisms.
\end{abstract}

\begingroup
\makeatletter
\let\thefootnote\relax
\long\def\@makefntext#1{\noindent#1}%
\makeatother
\footnotetext{Contact: \texttt{\{\href{mailto:joel.andersson@ist.ac.at}{joel.andersson},\href{mailto:anamay.chaturvedi@ist.ac.at}{anamay.chaturvedi},\href{mailto:monika.henzinger@ist.ac.at}{monika.henzinger}\}@ist.ac.at, \href{mailto:roodabehsafavi@gmail.com}{roodabehsafavi@gmail.com}.}}
\endgroup

\newpage

\section{Introduction}
\label{sec:introduction}

The seminal works \cite{dwork2010differentially} and \cite{chan2011private} introduced the notion of differential privacy (DP) under continual observation (CO), or DP-CO in short, to reason about the loss in privacy that occurs when an analyst collects sensitive data over time, and outputs the values of a sensitive function of the accumulated data iteratively. This definition has proven to be widely influential, as it captures important instances of sensitive data monitoring, such as federated learning of on-device language models \citep{xu2023federated}, and continual release of aggregate web-service statistics, where the absence of any privatization has been shown to leak individual behavior \citep{calandrino2011you}. At a high level, we say that a mechanism $\calM: \calX^* \to \calY^*$ is $(\eps,\delta)$ DP-CO if for any pair of \emph{neighboring} data streams $x, x'\in\calX^*$, and any set of output streams $Y \in \calY^*$,
    \[ \Pr[\calM(x) \in Y] \leq e^{\eps} \Pr[\calM(x')\in Y] + \delta. \]
Here, $x$ and $x'$ are considered neighboring if they have the same length, and if there is at most one index $t$ such that $x_t \not= x'_t$ - we will refer to this notion as swap-neighboring. Intuitively, this definition captures the idea that the value of an individual's contribution at a particular time-step does not significantly perturb the output distribution of a DP mechanism.

\cite{dwork2010differentially} motivate this notion of (event-level) adjacency by modeling the iterative process of processing a stream of inputs and generating outputs as occurring over an atomic sequence of discrete time intervals. Each stream element is the value contributed by an individual, and in each time interval a mechanism receives a value, processes it, and generates a new output. They point out that since they are studying real systems where time must be accounted for, they identify time-steps where ``nothing happened'', i.e., no stream element was given to the mechanism, with receiving values of $0$. Further work by Dong, Luo, and Yi~\citep{dong23continual} introduced the symbol $\bot$ to depict the absence of a value received at certain time-steps; this avoids conflating the absence of an input with a $0$-valued input for problem settings where the two symbols are semantically distinct.

There is a large quantity of prior work \citep{kifer2011no, kifer2014pufferfish, kasiviswanathan2014semantics, tschantz2020sok} in the DP literature reasoning about the semantics of privacy and apt definitions of neighboring data sets, given the threat model and the adversary's side information. At a high level, the analyst should be able to promise an individual that the latter's decision to participate in the data collection process does not significantly change the output of the mechanism. More concretely, the posterior belief of an adversary after observing the mechanism's output if the individual participates in the process should be close to the counterfactual posterior if they choose not to participate in the process.

The choice of neighboring relation determines which counterfactual the privacy guarantee binds. The standard decision to model neighboring streams in the CO setting using swap neighbors ensures that the mechanism is insensitive to their value. Adding a $\bot$ symbol to the data universe $\calX$ to indicate non-participation is more general than just the swap-neighboring definition by itself, but we argue below that this by itself does not suffice to capture all counterfactual neighboring streams when the goal is to hide participation.

For example, consider an individual who has the option to participate in the data collection, but with the constraint that they can only contribute their value $x_*$ at some real-world time $r_*$. Following \cite{dwork2010differentially}, we consider a DP-CO mechanism that discretizes time into atomic steps, and accepts at most one individual's value per time-step. If the individual chooses to participate in the data collection, then the real-world time $r_*$ at which they may contribute their value must be mapped to some time-step $t_*$. Let $x = (x_1,\dots)$ be the `default' data stream that would be received by a mechanism in the CO setting if this individual does not participate. We consider two situations. If in the absence of the user's contribution the time-step $t_*$ was such that $x_{t_*} = \bot$, i.e., no other user participated at a real-world time which was mapped to $t_*$, then the counterfactual stream for when the user does participate is defined simply by defining $x'_t = x_t$ for $t\not=t_*$, and $x'_{t_*} = x_*$. On the other hand, if the default stream were such that $x_{t_*}\not=\bot$, then the real-world time $r'_*$ of some other individual was also mapped to $t_*$.

Since the mechanism can process only one individual per time step, a choice must now be made - the mechanism sets either $x'_{t_*}=x_{t_*}$ or $x'_{t_*}=x_*$. Regardless of the tie-breaking rule, we now have an extraneous value that must again be allocated a position in the stream. The natural choice would be to then set $x'_{t+1}$ equal to the extraneous value, but this leads to the same problem when $x_{t+1}\not=\bot$. We see that the counterfactual stream necessarily involves a sequence of shifts until the next $\bot$ value is encountered - in the absence of a $\bot$ value, potentially the whole stream, starting from the time step $t$, is shifted by $1$ unit.

This example illustrates a gap between existing notions of adjacency studied in the CO setting, and what counterfactual streams look like when the goal is to privatize an individual's decision to participate in the data collection of a CO mechanism. In this paper, we study the significance of this shifting artifact by focusing on the \emph{continual counting} problem in DP-CO.

The continual counting problem was introduced by \cite{dwork2010differentially,chan2011private} in the very first works in DP-CO. The analyst is given a stream of real-valued inputs, and must output at each time step an estimate of the sum of all inputs seen so far. If the exact counts are released, then the counter reveals all stream values. When the privacy of the stream must be preserved, the analyst must compute some noisy outputs that balance utility and privacy over time.
As a motivating example, consider a testing facility that publishes a running count of the number of positive COVID test results observed so far. Here, each input $x_t \in \{0,1\}$ indicates whether the $t$-th test is positive or negative, and a value $x_t = \bot$ indicates that no-one tested at time-step $t$. Continual counting is one of the foundational problems in the study of DP-CO, and serves as an excellent starting point for our investigation of counterfactual data streams and neighboring relations.

\subsection{Our contributions}

\noindent\textbf{I. Edit-neighboring streams.} To capture the sort of counterfactual stream described above via the neighborhood relation that privatizes participation, we introduce a notion of adjacency that allows for streams with an \emph{arbitrary number of contiguous shifts} to be neighboring. This captures the collision scenario and queuing phenomenon that we discuss; a single insertion at any point can lead to an arbitrary number of downstream shifts, and any $\bot$
values encountered absorb and terminate this sequence of shifts.

\begin{definition}[edit neighboring, informal]
    Suppose a mechanism accepts input for some $T$ time-steps. Given $x, x'\in (\calX\cup\{\bot\})^*$ and a time step $i\in[T]$, let $j$ be the first time step after $i$ with $x'_j=\bot$; if no such time step exists, let $j = T+1$.
    We define $x$ and $x'$ to be \emph{edit-neighboring} if $x'_k = x_k$ for $k<i$, $x_i\in \calX$, $x_{k-1}'=x_k$ for $i+1\leq k\leq \min\{j, T\}$ and $x'_k=x_k$ for $j<k\leq T$, i.e.~the time steps ``between'' $i$ and $j$ are shifted. 
\end{definition}
The above definition captures settings where the system is run for a certain number of steps $T$, and thus if a shifting occurs and the last input ends up in spot $T+1$, it will not be processed and is not considered part of an input. For example, if the last person of a queue reaches a COVID testing center when no more slots are available, they will not get tested. One can also consider a natural variant of edit neighboring which allows different lengths for neighbor streams and covers the scenarios where the system continues processing until the data generation process indicates that the stream has terminated. We address this as well by a simple reduction to the fixed time-horizon setting defined above, so for exposition we focus primarily on the simpler definition described here.

We note that this definition can lead to the challenging setting where many bits in the input differ. For example, if $(x_i, x_{i+1}, \dots , x_j) =  (0, 1, 0, 1, \dots, 0, 1, \bot)$, then $x$ and $x'$ differ in \emph{all} values starting from  $i+1$ to $j$. Thus, this definition requires us to ``hide'' the difference of $j-i$ many bits from the adversary and not just one bit as required by the classic definition of neighboring streams. 
The edit neighboring definition gives us the expressiveness to represent arbitrarily long time stretches where no data is recorded, and also protect settings where an insertion collision at one time step causes all entries to shift. Further, in contrast with some prior works on continual counting, in this paper we will consider real-valued inputs $x_t \in [0,1]$, as opposed to just the binary setting $x_t\in \{0,1\}$.

\medskip\noindent\textbf{II. Lower bounds for private continual counting for edit-neighboring streams.}

At the outset it is not clear whether one can achieve a good trade-off between privacy and accuracy for continual counting when measuring privacy loss by defining adjacency via edit-neighboring streams. The privacy analyses of the state-of-the-art \citep{henzinger25normalized} data-independent factorization mechanisms require their inputs to have low $\ell_1$-sensitivity, i.e., $\| x - x' \|_1 \leq 1$. However, when shifts are introduced, the $\ell_1$-sensitivity between neighboring inputs can be as high as $T$, leading to polynomially scaling additive error. We find that this is not just a shortcoming of the analyses, but that in fact \emph{all} data-independent additive-noise mechanisms (i.e., mechanisms which add noise to the prefix sums where the additive noise is independent of the sums) \emph{must} incur polynomial error in this setting.

\begin{theorem}[Lower bound, informal]\label{thm:lb-edit-informal}
      Let $\eps\in(0, 1/2], \delta\in[0, 1)$, and $T\in\N$ be sufficiently large.
      Let $\cM$ be an $(\eps, \delta)$-DP algorithm for continual counting on $\{0,1,\bot\}^T$ with edit-neighboring streams.
      Suppose that $\cM$ is a data-independent additive noise mechanism.
      Then, with probability 0.99, it incurs additive $\ell_\infty$ error $\alpha$ where
      \begin{enumerate}
          \item If $\delta = o(\eps / T)$, then $\alpha = \tilde{\Omega}\left(\min\left\{\frac{T^{1/3}}{\eps^{2/3}}, T\right\}\right)$.
          \item If $\delta = 0$, then $\alpha = \Omega\left(\min\left\{\sqrt{\frac{T}{\eps}}, T\right\}\right)$.
      \end{enumerate}
\end{theorem}

\medskip
\noindent\textbf{III. Continual Counting Mechanisms with polylog additive error for edit-neighboring streams.} The lower bounds above suggest a significant gap between the privacy-accuracy trade-offs for swap neighboring and edit-neighboring streams for data-independent additive-noise mechanisms. We address this gap by developing a meta-algorithm $\simecc$ (Simple edit-neighboring Continual Counter) that takes as input any continual counting algorithm that preserves privacy in the standard setting as a black box, and constructs a continual counter for the edit-neighboring setting whilst incurring additive error close to that of the standard setting.

\begin{theorem}[Upper bound, informal version of \Cref{thm:simecc-privacy} and \Cref{thm:ecc-simpart-accuracy}]\label{thm:informal-ub}
    For privacy parameters $\varepsilon > 0$, $\delta\in(0,1)$, there exists an algorithm $\ecc$ that is $(\varepsilon,\delta)$ DP-CO for edit-neighboring streams, and with probability $1-\delta$ for all time-steps $t \geq 1$, incurs additive error $O(\frac{1}{\eps}\ln t \sqrt{\ln(t/\delta) \ln(1/\delta)}) + O(\tfrac{1}{\eps}\log(t/\delta))$.
\end{theorem}

In comparison, state-of-the-art $(\eps,\delta)$-DP continual counters in the swap-neighboring setting have the guarantee that for any failure probability $\beta>0$, with probability $1-\beta$, for all $t\geq1$, they incur additive error of the form $O(\tfrac{1}{\eps}\ln(t)\sqrt{\ln(t/\beta)\ln(1/\delta)})$. Given any fixed constant $c >1$, if $\delta = \beta^c$, then the additive error $\simecc$ incurs matches that in the standard setting up to constant factors. 

In real-world data collection the data generation might be bursty and sparse, leading to input streams consisting mostly of $\bot$ values, which increases the additive error polylogarithmically even though most of the time no data is being collected. We show how to handle sparse streams by giving a more general mechanism $\ecc$ (edit-neighboring Continual Counter) that adapts to the \emph{sparsity} $s(x,t)$ of the data generation process, defined by the expression
\[ s(x,t) \,:=\, \sum_{i=1}^t 1(x_i\not=\bot), \]
i.e., $s(x,t)$ is the number of non-$\bot$ values observed. We see that in the real-valued setting wherein $x_t\in [0,1]$ for all $t\geq1$, the prefix sum 
$$s_t = \sum_{i=1}^t x_i$$
is bounded from above by the sparsity $s(x,t)$ for all $t$.
 
\ecc incurs data-dependent additive error that scales primarily with the prefix sum of the input stream. Our mechanism incurs additive error that scales polylogarithmically only with $s_t$, instead of with $t$. The dependence of the additive error on the stream length $t$ is relegated to a lower order $\log t$ additive term.
 This extends prior work by \cite{dwork2015pure} for swap-neighboring streams to the edit-neighboring setting. 

\begin{theorem}[Upper bound, informal version of \Cref{thm:ecc-privacy}, \Cref{thm:ecc-accuracy} and \Cref{obv:sparsity}]\label{thm:informal-ub.2}
    In the same setting as \Cref{thm:informal-ub}, there exists a mechanism that is $(\eps,\delta)$ DP-CO for edit-neighboring streams that, with probability $1-\delta$ for all time-steps $t \geq 1$, incurs additive error 
    $$O\left(\frac{1}{\eps}\ln s_t \sqrt{\ln(s_t/\delta) \ln(1/\delta)}) + O(\tfrac{1}{\eps}\log(t/\delta)\right).$$ 
\end{theorem}

\medskip
\noindent\textbf{IV. Lower bounds for prefix-neighboring streams.} Given that similar additive error is achievable whilst fulfilling a stronger notion of privacy, a natural question is whether it is possible to further generalize the notion of neighboring streams.
In the privacy literature one such notion that is studied is that of \emph{prefix-sum neighboring} streams.
Two streams are prefix-sum neighboring if the prefix sums of one stream always trails the other by at most 1.
This property holds for certain dynamic problems when reduced to continual counting on difference sequences, e.g., as in the case of counting distinct elements~\citep{jain23distinct,andersson26improved}.
However, we show that \emph{no} DP mechanism can achieve sub-polynomial error in this even more general setting.

\begin{theorem}[Lower bound, informal]\label{thm:lb-prefix-informal}
      Let $\eps\in(0, 1/2], \delta\in[0, 1)$, and $T\in\N$ be sufficiently large.
      Let $\cM$ be an $(\eps, \delta)$-DP algorithm for continual counting on $\{0,1\}^T$ with prefix-sum-neighboring streams.
      Then there exist inputs for which, with probability 0.99, it incurs additive $\ell_\infty$ error $\alpha$ where
      \begin{enumerate}
          \item If $\delta = o(\eps / T)$, then $\alpha = \tilde{\Omega}\left(\min\left\{\frac{T^{1/3}}{\eps^{2/3}}, T\right\}\right)$.
          \item If $\delta = 0$, then $\alpha = \Omega\left(\min\left\{\sqrt{\frac{T}{\eps}}, T\right\}\right)$.
      \end{enumerate}
\end{theorem}

\medskip
\noindent\textbf{V. Experiments on Synthetic Data. }
Our first lower bound (\Cref{thm:lb-edit-informal}) and upper bound (\Cref{thm:informal-ub}) together imply that the vast majority of existing swap-private mechanisms, when applied to edit-neighboring streams, must incur exponentially larger error to provide a privacy guarantee comparable to our mechanism.
This separation is however asymptotic - it says nothing about whether it manifests for streams of moderate length seen in practice, or whether the separation requires carefully engineered, and very long, worst-case inputs.
We empirically demonstrate that a separation appears already at $T=10^4$, on inputs drawn from a structured (and arguably realistic) distribution, under a simple attack.

Recall that edit-neighboring captures a scenario where an individual's participation inserts an event shifting subsequent events along time.
In particular, if the stream $x\in\{0,1\}^T$ is sampled from a distribution with temporal patterns, then the offset in the pattern can potentially be detected in the output of a continual counter.
We instantiate this concretely: a defender draws a Bernoulli stream $x^{(0)}$ of length $T=10^4$ with rate alternating between 0.1 and 0.9 in blocks of 10, samples a challenge bit $b\in\{0,1\}$, and releases $y=\mathcal{M}(x^{(b)})$ where $x^{(1)}$ is an edit neighbor of $x^{(0)}$ derived by inserting a 1 at position 1.
The attacker, given knowledge of the Bernoulli rate profile and $y$, computes a linear statistic on $y$ to infer the value of $b$.

\begin{figure*}[t]
  \centering
  \begin{subfigure}[b]{0.49\linewidth}
    \centering
    \includegraphics[width=\linewidth]{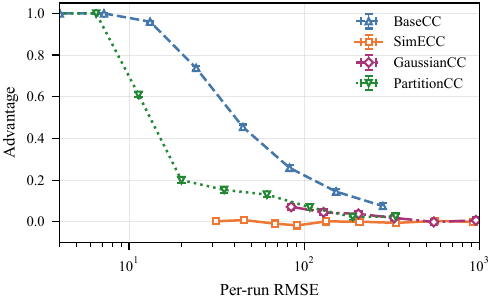}
    \caption{Attacker advantage vs.\ per-run RMSE.}
    \label{fig:adv-vs-rmse-main}
  \end{subfigure}
  \hfill
  \begin{subfigure}[b]{0.49\linewidth}
    \centering
    \includegraphics[width=\linewidth]{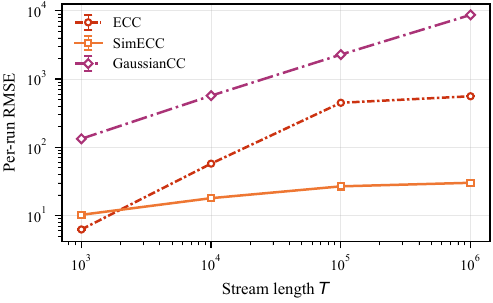}
    \caption{Per-run RMSE vs\ $T$.}
    \label{fig:rmse-vs-t-main}
  \end{subfigure}
  \caption{
    Plots evaluating mechanism performance on Bernoulli streams.
    \Cref{fig:adv-vs-rmse-main} uses $T=10^4$, $\delta=10^{-5}$, and alternating Bernoulli rates $p\in\{0.1, 0.9\}^T$ on blocks of width $10$, with each point shown being the average over $N=20000$ runs for $\varepsilon\in\{0.10, 0.19, 0.38, 0.73, 1.4, 2.8, 5.4, 10, 20, 40, 77, 150\}$.
    \Cref{fig:rmse-vs-t-main} uses $\eps=10$, $\delta = 10^{-5}$ a constant Bernoulli rate of $0.01$, with each point shown being the average over $N=100$ runs for one value of $T\in[10^3, 10^6]$.
    The error bars on the per-run RMSE are showing the standard error (SE), and the error bars on the advantage are 95\% Clopper--Pearson confidence intervals; $N$ has been set such that the error bars are small enough to be barely visible.
    Advantage is defined as $2\cdot \Pr[\text{Adversary guesses $b$ correctly}] - 1$.
    }
  \label{fig:adv-vs-rmse}
\end{figure*}

We compare two continual counters (\basecc, \partitioncc) for swap-neighboring streams against our edit-neighboring private mechanism \simecc, and a baseline for edit-neighboring streams based on releasing the prefix sum at each step with the Gaussian mechanism (\gausscc).
On fixing $\delta=10^{-5}$ and sweeping $\eps$, we plot the trade-off between the \emph{advantage} (defined as $2\Pr[b\ \text{guessed correctly}]-1$) of the attacker, and the per-run root-mean-squared error (RMSE) of the mechanism.
Intuitively, a smaller $\eps$ should promote lower advantage but comes at a cost of higher error, and the shape of this trade-off becomes our basis for comparison.
We are in particular interested in comparing them at a given bound on the attack advantage - the analyst is interested in knowing what error they must pay to bound the attack success rate of the attacker.

\Cref{fig:adv-vs-rmse-main} shows the result.
To bound the attacker's advantage at 0.1, \basecc and \partitioncc must incur error $9\times$ and $3\times$ larger respectively than \simecc; \gausscc also has to pay roughly $3\times$.
\Cref{fig:rmse-vs-t-main} sweeps $T\in\{10^3, 10^4, 10^5, 10^6\}$ for the edit-neighboring continual counters, demonstrating that the error gap to the Gaussian baseline grows with $T$, in line with the predicted polylog vs. polynomial error growth.

\medskip
\noindent\textbf{Summary.}
Our main theoretical findings are that the notion of edit-neighboring streams is strictly more general in terms of characterizing privacy loss than the standard setting, and still permits additive error on the same order. Further, this notion of edit neighboring seems to achieve a `sweet-spot' in terms of the additive error incurred, as the only more general setting that we are aware of must provably suffer additive error that is exponentially larger.

Empirically, we find that in an idealized toy model, even for randomly generated data following a schedule known to the adversary, to restrict the attack advantage of the adversary to the same value requires incurring significantly more error for private continual counters designed for swap-neighboring streams when compared with our mechanisms for edit-neighboring streams.

\subsection{Related work}

The issue of insertion collisions has been implicitly addressed in different ways in prior work. Most works, starting from \cite{dwork2010differentially} and \cite{chan2011private}, assume as a fact that at every atomic time-step, at most one value is received, and that therefore the swap model suffices. In principle this could be achieved by a sufficiently fine discretization of time, depending on the rate of data generation. This approach has the caveat that a finer discretization of time invariably leads to more additive error, even for fundamental tasks like counting, and a good choice of discretization requires some prior knowledge of how high the data generation rate could be. The model assumption fails if this prior is incorrect. 

\cite{DBLP:conf/aistats/Cardoso022} allow arbitrarily many values to be received per time-step, so the input to the mechanism is a sequence of sets, as opposed to singletons.  This avoids the discretization issue but introduces a different issue. Real-world systems have capacity limits and bottlenecks; when the data generation rate exceeds the system's per-step capacity, excess inputs are queued even prior to time-stamping, in which case the model's per-step set no longer matches what the system processes. The natural fix is \emph{truncation}, wherein if the number of inputs exceed per-step capacity, then elements that arrive once capacity is reached are simply dropped. However, truncating sets introduces data-dependent error when server capacity is below the data-generation rate, since discarded data is permanently lost. In contrast, our model buffers excess data; bursts exceeding capacity incur only temporary delay, and any induced error disappears once the buffer clears. Moreover, dropped points may be correlated. If bursts align with server time steps and contain mostly $0$'s initially and $1$'s later, truncation disproportionately removes $1$'s, producing error that need not remain sublinear in the number of time steps.

Beyond the question of per-step capacity, the choice of adjacency relation itself has been refined in prior work. The closest prior precedent for our edit-neighboring relation is the insert-delete adjacency introduced by \cite{casacuberta22widespread} for the batch DP setting. Their work identifies a class of vulnerabilities in deployed DP libraries: finite-precision arithmetic (floating-point non-associativity and integer overflow) makes the implemented sensitivity of basic statistics like sums substantially larger than the idealized sensitivity used to calibrate noise, breaking the DP guarantees of essentially every major library at the time.

Their definition of insert-delete adjacency is structurally similar to the edit-neighboring relation we introduce here; both define neighbors via the insertion or deletion of a single record, with downstream records shifting accordingly. However, other than the difference in setting and techniques, there is also a significant difference in motivation. In the batch setting, ordering is irrelevant to sensitivity except as a finite-precision artifact, so an ordered adjacency relation is needed only to address implementation vulnerabilities. In continual observation, ordering is intrinsic to the setting; participation shifts the timing and order of downstream records even for the theoretical model.

\subsection*{Independent concurrent work} In independent work \cite{chan2026continual} study \emph{single-edit neighboring user streams}. Their goal is to construct DP-CO mechanisms with adaptive adversaries, with a notion of neighboring streams which is a special case of ours: given a stream of \emph{user events} (e.g. value contributions), a neighboring stream is constructed by inserting an event at some time-step $t$, and incrementing the index of every element from this point onward by $1$, but they do not allow $\bot$ values, i.e., they do not formalize \textit{non-participation}.
Algorithmically, their algorithm is a generalization of our \simecc mechanism. As they do not allow $\bot$ values, they do not consider limited cascades of shifts, or sparsity-adaptive methods, such as our \ecc mechanism. This is one of our main contributions. They also do not study any lower bounds for edit-neighboring or prefix-sum-neighboring streams.

\cite{chan2026continual} give a general approach towards adapting standard DP-CO mechanisms for the edit-neighboring setting, with a randomized binning wrapper that converts any standard DP-CO mechanism into an edit-neighboring DP-CO mechanism. 
They incur slightly more additive error (in their setting called  \emph{backlog}) of $O(\tfrac{1}{\eps}\ln (t) \ln(t/\delta))$ compared to the $O(\frac{1}{\eps}\ln(t) \sqrt{\ln(t/\delta) \ln(1/\delta)} + \tfrac{1}{\eps}\log(t/\delta))$ error that
our continual counting method \simecc achieves.
\section{Technical Overview}

\noindent\textbf{I. Edit-neighboring streams and the \simecc mechanism}

Let us consider a pair of edit-neighboring streams $x=(x_1, \dots, x_T)$ and $x'=(x_1', \dots, x_T')$, where for simplicity, in this overview, we assume neither of $x$ and $x'$ contains a $\bot$, and where
$x$ is an insertion neighbor of $x'$ at step $i\in [T]$, i.e.~$x$ is obtained from $x'$ by inserting $x_i$ after $x_{i-1}'$, shifting all subsequent elements $x_i', \dots, x_T'$ one position to the right, and discarding the last element $x_T'$. Note that insertion neighbors are also edit neighbors if there are no $\bot$s.
Our goal is to design a continual counter such that if the respective output streams for $x$ and $x'$ are denoted $y$ and $y'$, then the distributions of $y$ and $y'$ are $(\eps,\delta)$-indistinguishable. Further, for any fixed stream $x$, the output $y$ should be ``close'' to the stream $s = (s_t)_{t\geq1}$ with high probability, i.e., the maximum absolute difference over all time steps ~$\max_{t\geq1} \| y_t - s_t \|$ is small,  where $s$ is the prefix sum stream defined element-wise by the expression $s_t := \sum_{i=1}^t x_i$ for all $t\geq1$.

The starting point for our upper bounds is to try and reduce the edit-neighboring setting to the standard swap neighboring setting, and then apply standard continual counters. 
It will be useful to keep in mind that for a standard continual counter $\cc$ based on the factorization mechanism, if we have the promise that neighboring streams have $\ell_1$-sensitivity $\mathsf{sens}$, then with an appropriate choice of parameters we can achieve $(\eps,\delta)$-DP\footnote{There are also continual counters for $(\eps, 0)$-DP that leverage Laplace noise. Our focus is on $\delta > 0$ where Gaussian noise achieves smaller asymptotic error as measured in $t$.} and the guarantee that with probability $1-\beta$, $\max_{t\geq1} \| y_t - s_t \| \leq E_\cc(t)$, where
    \[E_\cc(t,\beta) \,:=\, O\left(\frac{\mathsf{sens}}{\eps} \ln t \sqrt{\log (t/\beta) \log (1/\delta)} \right). \]
Let $s'$ denote the prefix-sum stream for $x'$. We see that $s$ and $s'$ are identical up to point $i-1$, and starting from point $i$ there is a potential mismatch. 
Standard private continual counters achieve privacy by adding a vector of noise values $N = (N_t)_{t\geq1}$ directly to the prefix sum stream element-wise, i.e., they set $y_t = s_t + N_t$ and $y'_t = s'_t + N_t$. The magnitude of this noise vector scales with the $\ell_1$ sensitivity of the input stream under the neighboring relation. For edit-neighboring streams, starting from the point of difference $t\geq i$, if the terms of $x'$ alternate in value (for instance an alternating sequence $(0,1,0,1\dots)$), then it is easy to see that the $\ell_1$-sensitivity $\mathsf{sens}$ scales as $\Theta(t-i)$, and the standard analysis implies large error.

One idea to reduce the $\ell_1$-sensitivity of edit-neighboring streams is \emph{bucketing}. Instead of releasing a fresh update at every time step, we only change the output value every $B$-many time steps, i.e., for $j \in \N$ we define $y_{j\cdot B} = s_{j\cdot B} + N_{j}$, and for $t = j\cdot B + k$ for some $k\in \{1,\dots,B-1\}$ we define $y_t = y_{j\cdot B}$. This introduces some error since the running sum is only updated intermittently, but the magnitude of this error is bounded by the bucket size. We see that bucketing gets rid of some pathological examples, for instance setting the bucket size $B = 2$ ensures that the $\ell_1$ sensitivity of a pair of edit-neighboring streams with the alternating $0,1$ sequence described above no longer scales with $t$. However, bucketing clearly does not suffice by itself in every case. For example, suppose $x'$ is defined by $x'_i = 1$ for $i = 2kB$ for integers $k$ and $x'_i=0$ otherwise, i.e. that $x'$ has the value $1$ at step $i$ if and only if $i$ is an even multiple of $B$. Then one can check that when $x$ is an insertion neighbor of $x'$ at step $2B$ with $x_{2B} = 0$, then the sequences of sums of the values in each bucket for $x$ and $x'$ are $(0,1,0,1,\dots)$ and $(0,0,1,0,1,\dots)$ respectively. These are streams of length $\approx T/B$, and have $\ell_1$ sensitivity $\sqrt{T/B}$. One could set $B$ to be so large as to offset the $\sqrt{T}$ scaling, but the choice of bucket size itself introduces a  separate additive error of $B-1$, because there is now a delay of $B$ time-steps before the count is updated. In summary, the simple bucketing scheme described here still gives a polynomially scaling additive error.

To address these pathological instances with large $\ell_1$-sensitivity, we turn to \emph{randomizing the bucket size}. For each bucket index $\ell$ we draw the bucket size $B_\ell$ from a discrete Laplace distribution with mean parameter (simplified in this overview) of $(1/\varepsilon) \log (\ell/\delta)$, and scale parameter $1/\varepsilon$; we denote this distribution $\DLap((1/\eps)\log(\ell/\delta),1/\eps)$. With this choice, we have that with probability $1-\delta$, $B_\ell$ is positive for all $\ell \in [1,T]$, and that  $\Pr[B_\ell = b] = e^{ \eps} \Pr[B_\ell = b+1]$ or  $\Pr[B_\ell = b] = e^{- \eps} \Pr[B_\ell = b+1]$ for $b\geq 1$. Further, with high probability, the maximum bucket size scales only logarithmically with the number of buckets generated, and the accuracy error introduced by this choice of distribution will grow roughly as $E_1(\ell)$, where for all $\ell\geq1$,
\[ E_1(\ell) = O(\tfrac{1}{\eps}\log \ell/\delta). \] 

For ease of expression, we define a sequence of \emph{checkpoints} $(t_\ell)_{\ell\geq1} \in \N^*$. Let  $t_0 = 0$ and $t_\ell = \sum_{k=1}^\ell B_k$ for $\ell\geq1$; we see that in our intermittent update scheme, $t_\ell$ are precisely the time-steps at which we change the output, and demarcate the ends of buckets. Let us define the bucket sum values $\intsum_\ell = \sum_{t = t_{\ell-1}+1}^{t_\ell} x_t$ to denote the sum of all stream values $x_t$ such that $t$ falls in the $\ell$-th bucket, i.e., $t\in (t_{\ell-1},t_\ell]$. Since we have fixed a finite time horizon $T$, for any choice of checkpoints sequence $(t_\ell)_{\ell\geq1}$, we implicitly set $t_\ell \gets \min\{T,t_\ell\}$, and adjust bucket lengths accordingly.

Recall that $x$ is an insertion neighbor of $x'$ at $i$, and let $p$ denote the index of the bucket which contains $i$, i.e., $i \in (t_{p-1},t_{p})$. If $p = L$ is the last bucket index, then the $\ell_1$-sensitivity of the  $\intsum$ vector is $1$. If $p\not=L$, then we make the following high-level observation: if we couple the randomized bucket sizes to be equal, with the exception of $B_{p+1} = B'_{p+1} + 1$, then $\intsum_\ell = \intsum'_\ell$, except for $\ell \in \{p,p+1\}$, where they differ by at most 1; $\intsum_p$ and $\intsum'_p$ can vary because the insertion of $x_i$ changes the values of $x_k$ for $k\in (t_{p-1},t_p]$, and $\intsum_{p+1}$ and $\intsum'_{p+1}$ can vary for the same reason as well as their unequal bucket lengths. In other words, under this coupling of bucket sizes, the shift introduced at $i$ is largely absorbed in the buckets indexed $p$, $p+1$ - there is some additional subtlety regarding the end of the stream, but at a high level, we get the guarantee that only a constant number of bucket sums are modified.

This construction and randomized coupling reduces the problem of privatizing continual sums for edit neighbors $x$ and $x'$, to streams that are different at three entries by at most one unit, i.e. $(\intsum_\ell)_{\ell\geq1}$ and $(\intsum'_\ell)_{\ell\geq1}$, and consequently also $O(1)$ apart in $\ell_1$-distance. We can now use a standard private continual counter with $\mathsf{sens} = O(1)$ to generate privatized prefix sums $(\widehat{v}_\ell)_{\ell\geq1}$ for $(\intsum_\ell)_{\ell\geq1}$, and consequently for $x$ with minor error due to the lag introduced by bucketing. Then we have that for all $l\in [L]$, with probability $1-\beta$,
$|\widehat{v}_\ell - \intsum_\ell| \leq E_2 (\ell)$, where
\[ E_2 (\ell) := E_\cc (\ell,\beta).\]
However, there is one more remaining subtlety. The coupling-based privacy analysis above uses as privatizing randomness the obfuscation of the bucket sizes $B_\ell$. If we update the counter at time-steps $t_\ell$, then an adversary can reconstruct $B_\ell = t_\ell - t_{\ell-1}$, which violates the requirement that the privatizing randomness be hidden from the adversary. 

To deal with this, we use a \emph{biased continual counter} (\bcc) to generate privatized proxies $\widehat{t}_j$ for $t_j$, called \emph{noisy checkpoints}; we refer to the original checkpoints as \emph{true checkpoints} to avoid any ambiguity. Using \cc, we construct this biased continual counter that always generates overestimates of the prefix sums of its input streams, and suffers additive error roughly twice the additive error compared to \cc.

Concretely, let $(y'_t)_{t\geq1}$ be the outputs of $\cc$, when given inputs $(x_t)_{t\geq1}$. We define $y_t = \max(y'_t + E_\cc (\delta/2), \sum_{k=1}^t x_k)$. Since $|y'_t - \sum_{k=1}^t x_k| \leq E_\cc (t,\delta/2)$ with probability $1-\delta/2$ for all $t\geq1$, it follows that with probability $1-\delta/2$, the sequence $y_t$ equals the values $y'_t + E_\cc (\delta/2)$. Since $y'_t$ are generated in an $(\eps,\delta/2)$-DP manner, it follows that with probability $1-\delta/2$, the values $y_t=y'_t + E_\cc (\delta/2)$ are also $(\eps,\delta/2)$-DP, i.e. unconditionally $(\eps,\delta)$-DP. Further, by definition, $y_t \geq \sum_{k=1}^t x_k$.

One small complication is that in principle, even though checkpoint $t_2$ is defined after checkpoint $t_1$, the noisy checkpoint $\widehat{t}_2$ could occur before $\widehat{t}_1$, or even collide with $\widehat{t}_1$. 
We construct a schedule data structure $\slvar$ to keep track of when what updates are to be made, and we use any arbitrary tie-breaking rule to deal with noisy checkpoint collisions; since the tuples $((\widehat{t}_i,y_{t_i}))_{i\geq1}$ are already privatized, as long as tie-breaking depends only on these values, no additional privacy is lost.

We give $\bcc$ as input the difference sequence $(t_\ell - t_{\ell-1})_{\ell\geq1}$, and define its output sequence to be the noisy checkpoints $\widehat{t}_\ell$. Under the coupling of the randomized bucket sizes described above, this difference sequence also has constant $\ell_1$-sensitivity, and we can show that with probability $1-\beta$, for all $\ell\geq1$, $|\widehat{t}_\ell - t_\ell| \leq E_3(\ell)$ where $E_3(\ell)$ is the error introduced by the biased continual counter, i.e.
\[ E_3(\ell) = E_\cc(\ell,\delta/2) + E_\cc(\ell,\beta).\]
\sloppy We can now completely describe \simecc, a mechanism for continual counting that achieves $(\eps,\delta)$-DP for edit-neighboring streams. \simecc generates randomized bucket sizes $B_{\ell} \sim \DLap(1/\eps \log (\ell/\delta) , 1/\eps)$, and defines the \emph{true} checkpoints $t_0 = 0$ and $t_{\ell} = \sum_{i=1}^{\ell} B_i$ for $\ell\geq1$. As it processes the input stream $x$, it computes internally the bucket sums $\intsum_\ell = \sum_{i=t_{\ell-1}+1}^{t_\ell} x_t$, and, on reaching the end of each bucket $B_\ell$, it feeds $\intsum_\ell$ to a continual counter to generate the privatized value $\widehat{v}_\ell$, which is an estimate for $\sum_{k=1}^\ell \intsum_k = \sum_{t=1}^{t_\ell}x_t$. In parallel, it feeds the difference sequence of true checkpoints to $\bcc$ which in turn generates the noisy privatized checkpoints $\widehat{t}_\ell$. Starting from the first output $0$ at the first checkpoint $t_0 = 1$, for every $\ell\geq 1$, the mechanism $\simecc$ returns the same output as it did at $t_{\ell-1}$ until reaching a noisy checkpoint $\widehat{t}_\ell$ at which it updates it output value to $\widehat{v}_\ell$. The output of $\simecc$ can be described entirely in terms of the tuples $(\widehat{v}_{\ell},\widehat{t}_\ell)$; for any time-step $t$, the output of the mechanism is understood to be $y_j$ where $j$ is the index such that $\widehat{t}_j$ is the most recent noisy checkpoint that precedes $t$. 

The privacy loss of this routine can be bounded in terms of the coupling of the randomized bucket sizes needed to absorb insertions when comparing output distributions for edit neighbors, the privacy loss of the internal continual counter run on bucket sums, and the privacy loss of the biased continual counter that computes noisy checkpoints and schedules updates. The additive error can be bounded by the maximum bucket size $E_1$, the error of the internal continual counter $E_2(\ell,\beta)$,  and the additional delay introduced by the randomized scheduling of the biased continual counter, i.e. $E_3(\ell,\delta)$. This error expression is dominated by the error of the last term, so if we set $\beta = \delta$, then with probability $1- 3\delta$, for all $t\geq1$, 
\[ \|y - s \|_1 \,= \, O\left(\tfrac{1}{\eps} \ln t \sqrt{\log (t/\delta) \log (1/\delta)} \right). \]

\medskip
\noindent\textbf{II. Sparsity-Adaptive Continual Counting via \ecc}

The approach described above works well in the general case. We recall from the introduction that in our model, for time-steps where no data is received, we must record a $\bot$ value. For the purposes of defining the target continual count value, these placeholders are identified with $0$ and do not change the count, but they do increase the length of the stream and consequently the additive error. This is unnecessarily wasteful, and to address this we turn to the sparsity adaptive methods of~\cite{dwork2015pure}. 

We recall that the \emph{sparsity} of a stream $x$ at a time step $t$ is defined to be the number of non-$\bot$ values observed by $t$, i.e. $s(x,t) = \sum_{i=1}^t 1(x\not=\bot).$
We see that at any given time-step $t$, the prefix sum $s_t$ is always at most the sparsity $s(x,t)$. We introduce a second mechanism, called the \ecc mechanism, which incurs additive error that in the leading polylogarithmic term is only a function of $s_t$.

At a high level, \cite{dwork2015pure} introduce a \emph{stream partitioning mechanism} that for sparse streams incurs lower additive error than the general case. The partitioning mechanism uses a sequence of \emph{sparse vector technique} (SVT) instances to update the online continual count only once sufficiently many nonzero values have been seen since the previous most recent update to the count, i.e. to generate the $\ell$-th checkpoint, it compares the interval sum $s_t - s_{t_{\ell-1}}$ with some threshold $\tau_t$, and generates a new partition of the stream $(t_{\ell-1},t_\ell]$ when indicated to do so by the SVT. The upshot of this approach is that with similar privacy guarantees, one can achieve additive error that scales polylogarithmically only with the sparsity, and reduces the dependence of the additive error on the stream length to an $O((1/\eps)\log t)$ summand.

We introduce a subroutine \partt that operates analogously to that of \cite{dwork2015pure}, although our privacy and accuracy analyses will now be significantly different. As in prior work, it operates by using a sequence of SVT instances to partition the stream. We first briefly recall the guarantee of the SVT. The SVT is a mechanism that compares a stream of values $v = (v_t)_{t\geq 1}$ with a stream of user-defined thresholds $\tau_t$, and generates a stream of boolean outputs $a_t\in \{\bot,\top\}$. For each value, it either rejects the value $v_t$, indicated by the output $a_t = \bot$, or accepts it and halts, outputting $a_t = \top$ and not accepting any more values from the stream $v$. The outputs generated by the SVT have the promise that with probability $1-\beta$, for all $t\geq1$, $a_t = \top \Rightarrow v_t \geq \tau_t - E_\svt(t)$, and $a_t = \bot \Rightarrow v_t \leq \tau_t + E_\svt(t)$, where $E_\svt (t) = \Theta(\tfrac{1}{\eps}\log t/\beta)$. A remarkable property of the SVT is that it is $\varepsilon$-DP under the promise that any two neighboring streams $v$ and $v'$ are $1$-Lipschitz \emph{at every time step}. This is in contrast to the standard notion of neighboring mentioned before, where two neighboring streams could only differ at one time-step; here we have that for all $t$, $|v_t - v'_t|\leq 1$, but no other constraint.

At a high level, we would like $\partt$ to define new checkpoints at time-steps $t$ for which the SVT indicates that we have accumulated a large enough prefix sum value. When the $\ell$-th instance of the SVT is running, it compares the running interval sum $s_t - s_{t_{\ell-1}}$ with some sequence of threshold $\tau_t$. To avoid generating buckets for empty streams, we require that seeing an arbitrarily long sequence of $\bot$ values should not lead to a checkpoint declaration and bucket creation. This can be achieved by a simple modification of the SVT, by setting $\tau_t = \tau + E_\svt(t)$ for some $\tau >1$. The accuracy guarantee of the SVT now gives us that with probability $1-\beta$, for all $t\geq t_{\ell-1}$, if $a_t = \top$, then $s_t - s_{t_{\ell-1}} \geq \tau$. In other words, with high probability, regardless of the number of values seen, a bucket is only generated when the increase in the prefix sum exceeds $\tau$.

Analyzing the privacy loss incurred by \partt for edit-neighboring streams turns out to be a challenging task, primarily due to the fact that our bucketing is now data-dependent, and the coupling used before for \simecc no longer applies.

To address this challenge, we first establish a simple but useful novel observation about the SVT in this problem context. If, in addition to the promise of $1$-sensitivity along every time-step, one has the promise that the values being tested are $1$-sensitive when compared across \emph{consecutive} time-steps (i.e. for all $i\geq 1$, $|v_{i+1} - v'_{i}| \leq 1$), then the event that the SVT accepts a value and halts at time step $t$ is approximate-DP indistinguishable to the event that it accepts a value at time-step $t+1$ and halts. This fact allows us to construct a new data-dependent coupling between the bucket endpoints for the runs of the partitioning mechanism on two neighboring input streams. Essentially, we are able to show that under this data-dependent coupling, one can again ensure that the bucket interval sum streams $\intsum$ and $\intsum'$ defined by the bucket sums differ only in a constant number of elements, and are consequently close in $\ell_1$ distance as well. The details of the data-dependent coupling and its analysis are quite involved, and we direct the reader to \Cref{sec:ecc-privacy} for a more complete description. To privatize the true checkpoints defined by this routine, we use an instance of the \bcc as before. 

To reason about the accuracy of this method, we first note that the lengths of the streams passed to \cc and \bcc scale only with the number of buckets, which are now generated by the partitioning mechanism \partt. We would like to show that the number of buckets declared by the partitioning mechanism now grows only with the prefix sum $s_t$. The SVT guarantee gives us that when the $\ell$-th instance of the SVT returns $a_t = \top$, then $s_t - s_{t_{\ell-1}} \geq \tau$. Summing over all checkpoints up to time-step $t$, it follows that $s_t \geq \ell\cdot \tau$. It will follow that for a choice of $\tau \geq 1$, the number of checkpoints $\ell$ generated by time-step $t$ is at most $s_t$.

Since the number of buckets generated by time-step $t$ is at most $s_t$, it follows that the error incurred by the internal continual counter \cc, and the biased continual counter \bcc, now scale with $s_t$ instead of $t$. However, our accounting for the delay in defining new buckets must now be modified, since the maximum bucket size can be arbitrarily large and does not give us a useful error bound, unlike the analysis for \simecc.

We see that if the $\ell$-th instance of the SVT returns $a_t = \bot$, then we have the promise that $s_t - s_{t_{\ell-1}}\leq 2E_\svt(t)$. Since $E_\svt(t) = O(\tfrac{1}{\eps}\log\tfrac{t}{\beta})$, it now follows that we have the promise that $s_t \leq s_{t_{\ell-1}} + O(\tfrac{1}{\eps}\log\tfrac{t}{\beta})$. Putting everything together, we get the error bound
    \[ O\left(\tfrac{1}{\eps} \ln s_t \sqrt{\log (s_t/\delta) \log (1/\delta)} \right) + O(\tfrac{1}{\eps}\log\tfrac{t}{\beta}).\]

\medskip
\noindent\textbf{III. Extended Edit-Neighboring Streams}

As mentioned in the introduction, there is another plausible definition of edit-neighboring streams that permits neighboring streams to have different lengths. Concretely, here the time-horizon is determined by the data generation process, and not by the analyst. We assume that the input stream is marked by an end of stream symbol $\$$ which indicates to the analyst that the stream has terminated, and the mechanism may halt. Extending our discussion in the introduction on the definition of edit-neighboring streams, we see that if a value is inserted at some time-step $i$ which is not followed by any $\bot$ values, then the length of the stream would increase by $1$ and the end of stream symbol occurs one position later.

In \Cref{sec:modified-insertion-neighboring}, we formally define \emph{extended edit neighbors}, which captures this notion of edit neighboring. In addition to the shift phenomenon that the standard definition of edit-neighboring streams tries to capture, we must now also account for variable stream lengths. We appeal to a simple padding technique that adds a randomized number of $\bot$ values drawn from a Laplace distribution with scale parameter $(2/\eps)$ (if the value is negative, we set it to $0$). After this padding, we appeal to any $(\eps/2,\delta/2)$ edit-neighboring continual counter for the standard setting, and achieve the guarantee that the outputs of the continual counter on the padded stream are $(\eps,\delta)$-DP with respect to this notion of extended edit neighbors.

We characterize the error of this counter using two values - the excess number of values generated to obfuscate the length of the stream, and the difference between the outputs of the continual counter and the true sum value. For outputs generated by the continual counter after the end of the true stream that do not correspond to any real inputs from the stream, we define the additive error in terms of the true count at the end of the stream. Our main finding is that the additive error is identical to that incurred for the fixed time-horizon notion of edit-neighboring continual counters standard setting up to constant factors.

\medskip
\noindent\textbf{IV. Lower bounds}
To discuss our lower bounds, we need to first formally define \emph{prefix-sum neighboring} streams, which we do next.
\begin{definition}[Prefix-Sum Neighbors]
    For $T\in\N$, two sequences $\sigma=(x_1,\dots,x_T)$ and $\sigma'=(x_1',\dots,x_T')$ in $\mathbb{Z}^T$ (or a subset thereof) are said to be \emph{prefix-sum neighboring}, written $\sigma\sim_p\sigma'$, if one of the following conditions hold:
    \begin{itemize}
        \item For every $t\in[T]$, $\sum_{i=1}^t x_i - \sum_{i=1}^t x_i'\in \{0,1\}$.
        \item For every $t\in[T]$, $\sum_{i=1}^t x_i - \sum_{i=1}^t x_i'\in \{0,-1\}$.
    \end{itemize}
\end{definition}
We will use the notation $\ccgeneric{T}$ to denote the problem of continual counting with neighboring relation $\sim$ and input streams from $\mathcal{S}^T$.
Our lower bounds follow from the following chain of reductions:
\begin{align*}
   \cdistinct &\Rightarrow \ccptrits{T}\\
   \ccptrits{T} &\Rightarrow \ccpbits{T}\\
   \ccpbits{T} &\Rightarrow^* \ccebits{T}
\end{align*}
where $\cdistinct$ denotes the problem of counting distinct elements in a turnstile stream of length $T$ under item-level differential privacy (see \Cref{def:cdistinct}).
In particular, the second reduction gives \Cref{thm:lb-prefix-informal} and the third gives \Cref{thm:lb-edit-informal}.
The last reduction is starred to indicate that it only applies to a subclass of mechanisms.
We proceed to sketch the ideas necessary to complete the chain of reductions.

\paragraph{From \texorpdfstring{$\cdistinct$}{CountDistinct} to \texorpdfstring{$\ccptrits{T}$}{Count}.}
\cite{andersson26improved} gave a reduction from $\cdistinct$ to (after translating their notation to ours) $\cccustom{\sim_p}{\mathbb{Z}^T}$.
They considered a version of $\cdistinct$ in which multiple updates (item insertions/deletion) are allowed per step, whereas we restrict the problem to one update per time step.
With this restriction, we can refine their reduction down to $\ccptrits{T}$.

To then prove our lower bound on $\ccptrits{T}$, we use a lower bound for $\cdistinct$ from \cite[Theorem 1.7]{jain23distinct}.
Their version of $\cdistinct$ has a subtly different neighboring relation, which implies that 1-neighboring outputs for them can be up to 2-neighboring for us.
Nevertheless, by group privacy, any $(\eps, \delta)$-DP algorithm $\cM$ for us is also a $(2\eps, (1+e^{\eps})\delta)$-DP algorithm for them.
Hence, their lower bound can be used in a black-box manner at only a constant cost in the privacy parameters.
This gives us a lower bound on $\ccptrits{T}$.

\paragraph{From \texorpdfstring{$\ccptrits{T}$}{Count} to \texorpdfstring{$\ccpbits{T}$}{Count}.}
To extend our lower bound from ternary to binary inputs, we use a natural idea:
encode tritstreams $x\in \{-1,0,1\}^T$ into bitstreams $s\in\{0,1\}^n$.

There are a few constraints on $s$ in order for it to be a useful encoding for our lower bound: (1) prefix sums on $s$ need to encode prefix sums on $x$, (2) $s$ has to be constructible in an \emph{online} manner, (3) $n = O(T)$, and (4) given $x\sim_p x'$, we need $s \sim_p s'$.
Fortunately, a natural encoding works for our purpose.
Choosing $n=2T$ and defining $s$ via
\begin{equation*}
    (s_{2i-1}, s_{2i}) =
    \begin{cases}
        (0,1), & \text{if } x_i = 0, \\[4pt]
        (1,1), & \text{if } x_i = 1, \\[4pt]
        (0,0), & \text{if } x_i = -1,
    \end{cases}
    \qquad \forall i \in [T]\,,
\end{equation*}
suffices.
To decode the prefix sums, note that $t+\sum_{i=1}^t x_i = \sum_{i=1}^{2t} s_i$.
The only property that does not obviously hold is that prefix-sum neighbors are preserved (4), and we prove this via case analysis.

Equipped with this encoding, the reduction is straightforward.
Namely, given an input to $x$ to $\ccptrits{T}$, produce $s$ and feed it into a private algorithm $\cM$ for $\ccpbits{2T}$.
The even-index outputs of $\cM(s)$ (minus an offset) are an estimate of the prefix sums on $x$.
Hence, the lower bound extends (up to constants).

\paragraph{From \texorpdfstring{$\ccpbits{T}$}{Count} to \texorpdfstring{$\ccebits{T}$}{Count}.}
For this reduction, our goal is different.
Given that we have already designed an algorithm for $\ccebits{T}$ with polylogarithmic utility (\Cref{thm:informal-ub}), we cannot extend the lower bound without further qualification.
Hence, the lower bound we show is for a restricted class of mechanisms, namely those which are \emph{data-independent}.
Informally, a private mechanism $\cM$ for releasing a real-valued query $f$ on (any) input $x$, is data-independent if its output $\cM(x)$ is distributionally equivalent to $f(x) + Z$, where $Z$ is sampled from a distribution $\mu$, independent of $x$.

To prove our lower bound, we reason about \emph{sensitivity sets}.
Simplified, for a private estimation problem $P=(f, \sim)$ (interpretation: privately release $f(x)$ for inputs with neighboring relation $\sim$), the sensitivity set of $P$ is $S_P := \{f(x) - f(x')\,\vert\, x\sim x' \}$.
Given two such problems, $P=(f,\sim)$ and $P'=(f,\sim')$ where $S_{P'} \subseteq S_{P}$, we show that, for every data-independent mechanism $\cM$ for $P$ with additive error $\alpha$, there exists a mechanism $\cM'$ for $P'$ that attains the same error $\alpha$.
Hence, an \emph{unconditional} lower bound for $P'$ can be translated into a lower bound for any \emph{data-independent} mechanism for $P$.

Our lower bound for the class of data-independent mechanisms solving $P=\ccebits{T}$ proceeds by showing the corresponding sensitivity-set containment argument for $P' = \ccpbits{T}$.

\section{Preliminaries}\label{sec:prelims}

We make the following notational definition.

\begin{definition}[Notation]
    \begin{enumerate}
        \item We use the symbols $\N$ and $\R$ to denote the natural numbers and real numbers respectively.
        \item For $n\in \N$, we write $[n]$ to denote the set $\{1, \dots, n\}$.
    \end{enumerate}
\end{definition}
 
We recall the definition of differential privacy in both the batch and continual-observation model and introduce two well-studied mechanisms used by our mechanism. Then, we formally define the problem studied in this paper.

\medskip\noindent\textbf{Differential Privacy in Batch Model.} {
Let $\calX$ be a family of data records. A dataset is a multiset over $\calX$. We denote the family of all multisets over $\calX$ by $\MS(\calX)$. A \emph{mechanism} $\cM:\MS(\calX)\to\calY$ is a randomized function mapping a dataset $D\in\MS(\calX)$ to an output in $\calY$. The notion of differential privacy is built on the concept of \emph{indistinguishability} between probability distributions:

\begin{definition}[\texorpdfstring{$(\eps,\delta)$}{(epsilon,delta)}-indistinguishability~\cite{dwork2006calibrating}]
\label{def:indistinguishability}
    Let $\calY$ be a set. For $\eps\geq 0$ and $0\leq \delta\leq 1$, two random variables $Y,Y'$ over $\calY$ are said to be \emph{$(\eps,\delta)$-indistinguishable} if for every measurable set $S\subseteq\calY$,
    \[
        \Pr[Y\in S] \;\leq\; e^{\eps}\cdot \Pr[Y'\in S] + \delta,
        \qquad
        \Pr[Y'\in S] \;\leq\; e^{\eps}\cdot \Pr[Y\in S] + \delta.
    \]
\end{definition}

Differential privacy is defined with respect to a \emph{neighbor relation} $\sim$, which is a binary relation on the dataset space $\MS(\calX)$. In a common definition, two datasets are considered to be neighboring if they differ in the presence or absence of a single record.

\begin{definition}[Differential Privacy~\cite{dwork2006calibrating}]
\label{def:dp}
    Let $\calX$ be a data universe and $\sim$ a neighbor relation on $\MS(\calX)$. For $\eps\geq 0$ and $0\leq \delta\leq 1$, a mechanism $\cM:\MS(\calX)\to\calY$ is said to be \emph{$(\eps,\delta)$-DP with respect to $\sim$} if for every pair of datasets $D, D'\in\MS(\calX)$ satisfying $D\sim D'$,
    the random variables $\cM(D)$ and $\cM(D')$ are $(\eps,\delta)$-indistinguishable.  
    In the special case $\delta=0$, we say $\cM$ is $\eps$-DP.
\end{definition}

The Gaussian mechanism is a key primitive for $(\eps, \delta)$-differential privacy.

\begin{lemma}[(Analytic) Gaussian Mechanism~\citep{dwork2014algorithmic,balle18improving}]
\label{lem:gaussian-mech}
    Let $\calX$ be a data universe, $\sim$ be a neighbor relation on $\MS(\calX)$, and $f : \MS(\calX) \to \mathbb{R}^d$ a function.
    Define the $\ell_2$ sensitivity of $f$ as $\mathsf{sens}_f := \max_{D\sim D'}\| f(D)-f(D') \|_2$.
    Then, for every $\eps > 0$, $\delta\in(0,1)$, there exists a positive constant $C_{\eps, \delta}$ such that, for every dataset $D\in\MS(\calX)$, the mechanism that outputs $f(D) + \mathcal{N}(0, \sigma^2 I_{d\times d})$ for $\sigma = C_{\eps, \delta}\cdot\mathsf{sens}_f$ satisfies $(\eps, \delta)$-DP.
    Additionally, if $\eps \in (0, 1)$, then $C_{\eps, \delta} \leq \frac{1}{\eps}\sqrt{2\ln(1.25/\delta)}$.
\end{lemma}
}

\medskip\noindent\textbf{Properties of Differential Privacy.} {
We next state some fundamental properties of differential privacy.

\begin{lemma}[Post-processing~\cite{dwork2014algorithmic}]\label{lem:post}
    For $\eps\geq 0$ and $0\leq \delta\leq 1$, let $Y$ and $Y'$ be two $(\eps,\delta)$-indistinguishable random variables over $\calY$. Let $f:\calY\to \calZ$ be a (post-processing) function. Then the random variables $f(Y)$ and $f(Y')$ are also $(\eps, \delta)$-indistinguishable.
\end{lemma}

\begin{lemma}[Basic Composition~\cite{dwork2014algorithmic}] \label{lem:basic-composition}
    Let $\eps_1,\eps_2\geq 0$ and $0\leq \delta_1,\delta_2\leq 1$. Let $Y$ and $Y'$ be two $(\eps_1,\delta_1)$-indistinguishable random variables over $\calY$, and let $Z$ and $Z'$ be two $(\eps_2,\delta_2)$-indistinguishable random variables over $\calZ$. Then the joint random variables $(Y,Z)$ and $(Y',Z')$ are $(\eps_1+\eps_2,\;\delta_1+\delta_2)$-indistinguishable.
\end{lemma}

\begin{lemma}[Group Privacy]\label{lem:group-privacy}
    Let $\eps\geq 0$ and $0\leq \delta\leq 1$. For $k\in\N$, let $(Y_1, Y_1'), \dots, (Y_k, Y_k')$ be pairs of random variables over the same domain $\calY$. Suppose, for each $i\in [k]$, the random variables $Y_i$ and $Y_i'$ are $(\eps, \delta)$-indistinguishable. Then the tuples $Y=(Y_1, \dots, Y_k)$ and $Y'=(Y_1', \dots, Y_k')$ are $(k \cdot \eps, k\cdot e^{k\eps}\cdot\delta)$-indistinguishable.
\end{lemma}
}

\medskip\noindent\textbf{Continual Observation.} {
The mechanisms discussed so far operate in the \emph{batch model}: they map a static dataset $D$ to a single output, estimating the value of a function $f$ on $D$. In the \emph{continual observation model}, data arrives sequentially and the mechanism estimates the value of $f$ on the current dataset, updating the dataset at each time step.

Let $\calX$ denote the data universe and $\calY$ the output space. A \emph{continual mechanism} $\cM$ maps a data stream $\sigma=(x_1, x_2, \dots)$ to an output sequence $\cM(\sigma)=(y_1, y_2, \dots)$ over time. In the batch model, neighbor relation was defined on the family of datasets $\MS(\calX)$. In contrast, in the continual observation model, the neighbor relation is defined on the family of finite sequences of data, i.e., $\calX^*$. For instance, two sequences might be neighboring if they differ at a single step.

\begin{definition}[Differential Privacy under Continual Observation]
\label{def:dp-continual}
    Let $\calX$ be a data universe and $\sim$ a neighbor relation on $\calX^*$. For $\eps\geq 0$ and $0\leq \delta\leq 1$, a continual mechanism $\cM$ is said to satisfy \emph{$(\eps,\delta)$-DP with respect to $\sim$} if for every $T\in \N$ and every $\sigma,\sigma'\in\calX^T$ satisfying $\sigma\sim \sigma'$, the output sequences $\cM(\sigma)$ and $\cM(\sigma')$ are $(\eps,\delta)$-indistinguishable.
\end{definition}
}

\medskip\noindent\textbf{Neighbor Relations.} {
In this paper, we work with the following neighbor relations.

\begin{definition}[\texorpdfstring{$k$}{k}-Step \texorpdfstring{$\Delta$}{Delta}-Neighboring]\label{def:k-step-delta-neighboring}
    Let $k \in \N$ and $\Delta > 0$. For every $T\in\N$, two sequences $\sigma = (x_1,\dots,x_T)$ and $\sigma' = (x_1',\dots,x_T')$ in $\R^T$ are said to be \emph{$k$-step $\Delta$-neighbors} if there exist indices $i_1,\dots,i_k \in [T]$ such that:
    \begin{itemize}
        \item $x_j = x_j'$ for all $j \notin \{i_1,\dots,i_k\}$, and
        \item $|x_j - x_j'| \le \Delta$ for all $j \in \{i_1,\dots,i_k\}$.
    \end{itemize}
    The sequences $\sigma$ and $\sigma'$ are said to be \emph{all-step $\Delta$-neighbors} if $|x_j - x_j'| \le \Delta$ holds for every $j \in [T]$.
\end{definition}

\begin{definition}[\texorpdfstring{$k$}{k}-Shift \texorpdfstring{$\Delta$}{Delta}-Neighboring]\label{def:one-shift-delta-neighbor}
    Let $k\in\N$, $\Delta\in\R^{+}$, and $T\in\N\setminus \{1, \dots, k-1\}$. For two real-valued sequences $\sigma=(x_1,\dots,x_T)$ and $\sigma'=(x_1',\dots,x_{T-k}')$, we say that $\sigma$ is a \emph{$k$-shift $\Delta$-neighbor} of $\sigma'$ if the following conditions hold:
    \begin{enumerate}[label=(\roman*), ref=(\roman*)]
        \item\label{def:one-shift-delta-neighbor:1} $|x_{i+k}-x_i'|\le \Delta$ for all $i\in[T-k]$;
        \item\label{def:one-shift-delta-neighbor:2} $|x_i|\leq \Delta$ for all $i\in [k]$.
    \end{enumerate}
    We note that the $k$-shift $\Delta$-neighbor relation is not symmetric.
\end{definition}

\begin{definition}[Edit Neighboring]\label{def:queue-neighboring}
    Let $T\in\N$, and let $\sigma=(x_1,\dots,x_T)$ and $\sigma'=(x_1',\dots,x_T')$ be two sequences in $([0,1]\cup\{\bot\})^T$. We say that $\sigma$ is an \emph{insertion neighbor of $\sigma'$ at step} $i\in\{1, \dots, T\}$ if the following conditions hold:
    \begin{itemize}
        \item For all $t< i$, we have $x_t=x_t'$.
        \item Let $j$ be the smallest index in $\{i,\dots,T\}$ such that $x_{j}'=\bot$. If no such index exists, set $j=T+1$. Then, for every $t\in\{i+1,\dots,\min\{j, T\}\}$, we have $x_t=x_{t-1}'$.
        \item For all $t\in\{j+1,\dots,T\}$, we have $x_t=x_t'$.
    \end{itemize}
    We say that $\sigma$ and $\sigma'$ are \emph{edit neighbors}, written $\sigma\sim_e\sigma'$, if there exists $i\in[T]$ such that one of $\sigma$ and $\sigma'$ is an insertion neighbor of the other one at step $i$.
\end{definition}
}
\medskip\noindent\textbf{Private Continual Counter.} {
One of the fundamental problems in differential privacy under continual observation is \emph{private continual counting}. In this problem, a mechanism $\cc$ receives an input $x_t \in \R$ at each time step $t \in \N$ and outputs an estimate $y_t$ of the prefix sum $\sum_{i=1}^t x_i$. Later, when designing a continual counter that is private with respect to the edit-neighbor relation, we restrict the input domain to the bounded range $[0,1]$, which is necessary because the prefix sums of two edit-neighboring sequences could differ by the value of a single data record, and thus must be bounded in order to ensure privacy.

For $\alpha : \N \to \mathbb{R}_{\geq 0}$, and $\beta \in [0,1]$, the mechanism $\cc$ is said to be \emph{$(\alpha,\beta)$-accurate} if for every $T\in\N$ and input sequence $(x_1,\ldots,x_T)$, its output sequence $(y_1,\ldots,y_T) \in \R^{T}$ satisfies
\[
    \Pr\!\left[\forall t\in[T]\,:\,\left|y_t - \sum_{i=1}^t x_i\right| \le \alpha(t) \right] \ge 1 - \beta.
\]  
When important, we write $\alpha_\beta$ to emphasize that the function $\alpha$ has a dependency on $\beta$.

Differential privacy for the continual counting problem is well-studied under the neighbor relation in Definition~\ref{def:k-step-delta-neighboring}. We use the following result as our baseline continual counter.

\begin{lemma}\label{lem:cc-standard}
    For every $\eps > 0$ and $\delta \in (0,1)$, there exists a continual counting mechanism $\cc$ that is $(\eps,\delta)$-DP with respect to the 1-step $1$-neighbor relation.
    For any $\beta \in (0, 1)$, it is $(\alpha, \beta)$-accurate for
    \begin{equation*}
       \alpha_\cc(\eps,\delta,\beta,t) = O(C_{\eps, \delta} \ln(t)\sqrt{\ln(t/\beta)})
    \end{equation*}
    where $C_{\eps,\delta}$ is the Gaussian mechanism constant from \Cref{lem:gaussian-mech}.
\end{lemma}
\Cref{lem:cc-standard} follows from taking any (good) continual counter for bounded streams, based on adding correlated Gaussian noise to the true counts, and using the \enquote{doubling trick} from \cite{chan2011private}.
We give a proof in \Cref{app:cc-proofs}.

We also require private continual counters that in addition to being accurate, (i) never underestimate the true prefix sums and (ii) produce integer outputs. The following lemma allows us to construct such a \emph{biased} continual counter from a generic continual counter. 
\begin{lemma}\label{lem:cc-biased}
    Let $\cc$ be a continual counter 
    such that when given any fixed $\eps'>0$ and $\delta' \in (0,1)$, it is $(\eps', \delta')$-DP with respect to the \emph{1-step 1-neighbor} relation and such that for every $\beta'\in(0,1)$, it is $(\alpha(\eps',\delta',\beta',\cdot), \beta')$-accurate for some arbitrary given function $\alpha(\eps,\delta',\beta',\cdot): \N \to \R_{\geq 0}$.
    Then, for any given $\eps>0$ and $\delta\in(0,1)$, there exists a \emph{biased} continual counter $\bcc$ that is $(\eps, \delta)$-DP with respect to the \emph{1-step 1-neighbor}, and which satisfies the following properties:
    \begin{enumerate}
        \item On receiving $x_t$, its output $y_t$ is an integer and $y_t \geq \sum_{i=1}^t x_i$
        \item For every $\beta \in (0,1)$, it is $(\alpha_\bcc, \beta)$-accurate for
        \begin{equation*}
           \alpha_\bcc(\eps,\delta,\beta,t) = \alpha_\cc(\eps,\delta/2,\eta,t) + \alpha_\cc(\eps,\delta/2,\beta,t) + 1 
        \end{equation*}
        where $\eta = 0.5 \delta / (1+e^{\eps})\in(0, 1)$.
    \end{enumerate}
\end{lemma}
Both \Cref{lem:cc-standard,lem:cc-biased} are proved in Appendix~\ref{app:cc-proofs}, where we also show their respective construction.
}

\medskip\noindent\textbf{Sparse Vector Technique (SVT).}
The \emph{sparse vector technique (SVT)}, introduced by~\cite{DBLP:conf/stoc/DworkNRRV09}, is a mechanism for continually checking whether elements of a stream exceed a given threshold. In this work, we use a variant of SVT, denoted by $\svt$ and described in Algorithm~\ref{alg:svt}. The mechanism $\svt$ is initialized with a privacy parameter $\eps > 0$ and a threshold $\tau \in \R$, and processes a stream of real-valued inputs $x_1, x_2, \dots$. At each time step $t \in \N$, it privately compares the input $x_t$ with $\tau + 8\ln(t)/\eps$, and outputs either $\top$ or $\bot$ to indicate whether the input exceeds this term. The mechanism halts once it outputs $\top$ for the first time.

\begin{algorithm}
    \begin{algorithmic}
        \Function{Initialize}{Privacy parameter $\eps$, threshold $\tau$}
            \State $Z_0 \leftarrow \Lap(2/\eps)$
        \EndFunction
        \Function{Update}{$x_t \in\mathbb{R}$}\Comment{at time step $t\in \N$}
            \State $Z_t \leftarrow \Lap(4/\eps)$
            \If{$x_t + Z_t > \tau + 8\ln\left(t+\lceil \frac{12}{\eps}\rceil\right)/\eps + Z_0$}
                \State Output $\top$ and \textbf{halt}
            \Else
                \State Output $\bot$
            \EndIf
        \EndFunction
    \end{algorithmic}
    \caption{\texorpdfstring{$\svt$}{SVT}}\label{alg:svt}
\end{algorithm}

\begin{notation}\label{not:svt}
    For an input stream $\sigma$ of length $T \in \N \cup \{0\}$, we denote by $\svt(\sigma) \in [T+1]$ the time step at which $\svt$ outputs $\top$. If $\svt$ never outputs $\top$, we write $\svt(\sigma) = T+1$. In particular, for the empty sequence $()$ of length $0$, the mechanism never outputs any value including $\top$. Thus, $\svt(()) = 1$ with probability $1$.
\end{notation}

\begin{lemma}[\cite{lyu17understanding}]\label{lem:svt-privacy}
    Let $\tau\in\R$, $\eps\ge 0$, and $\Delta>0$. The mechanism $\svt$ described in Algorithm~\ref{alg:svt} is $\Delta\cdot \eps$-DP with respect to the all-step $\Delta$-neighbor relation in Definition~\ref{def:k-step-delta-neighboring}.
\end{lemma}

\begin{lemma}\label{lem:svt_accuracy}
    Let $\tau\in\R$, $\eps\ge 0$, $\Delta>0$, and $0< \beta\leq 1$. There is an event of probability $1-\beta$ conditioned on which, for every $T\in \N$ and every input sequence $(x_1, \dots, x_T)$, the mechanism $\svt$ described in Algorithm~\ref{alg:svt} satisfies the following accuracy guarantees for every $t\in [T]$:
    \begin{enumerate}
        \item \textbf{(Above threshold).} If $x_t\geq \tau + 16\ln \left(t+\lceil \frac{12}{\eps}\rceil\right)/\eps + 6\ln (2/\beta)/\eps + 4\ln(\pi^2/6)/\eps$ the mechanism $\svt$ returns $\top$ (and halts) at a step $t'\leq t$.
        \item \textbf{(Below threshold).} If $x_{t}\leq \tau - 6\ln(2/\beta)/\eps - 4\ln (\pi^2/6)/\eps$ the mechanism $\svt$ returns $\bot$.
    \end{enumerate}
\end{lemma}
\begin{proof}
    We see that the random variable $Z_0 \sim \Lap(2/\eps)$ and the random variables $Z_t \sim \Lap(4/\eps)$ i.i.d. From the tail bound of the Laplace distribution, we have that for $z_i \geq 0$, $\Pr[|Z_0| \geq z_0] = \exp(-\eps z_0/2)$ and $\Pr[|Z_t|\geq z_t] = \exp(-\eps z_t/4)$. Setting $z_0 = (2/\eps)\ln (2/\beta)$ and $z_t = (4/\eps) \ln (2/\beta_t)$ for $\beta_t = \tfrac{6}{\pi^2} \cdot \tfrac{\beta}{t^2}$, we have that via a union bound, with probability $1-\beta$, $|Z_0| \leq z_0$ and $|Z_t| \leq z_t$ for all $t\geq1$ simultaneously. Conditioning on this $(1-\beta)$-probability event holding, it follows that for all $t\in\N$,
    \[ x_t + Z_t = x_t \pm \frac{4}{\eps} \ln \frac{\pi^2 t^2}{3 \beta} \qquad\text{and}\qquad \tau + Z_0 = \tau \pm \frac{2}{\eps} \ln \frac{2}{\beta}.\]
    We can write 
    \begin{align*} x_t - \frac{4}{\eps} \ln \frac{\pi^2 t^2}{3 \beta} &\geq \tau + \frac{8}{\eps}\ln \left( t + \left\lceil \frac{12}{\eps}\right\rceil \right) + \frac{2}{\eps} \ln \frac{2}{\beta}  \\
        \Leftrightarrow x_t &\geq \tau + \frac{16}{\eps} \ln \left(t +\left\lceil \frac{12}{\eps}\right\rceil\right) + \frac{6}{\eps} \ln \frac{2}{\beta} + \frac{4}{\eps} \ln \frac{\pi^2}{6}. \end{align*}
    Similarly,
        \[ x_t + \frac{4}{\eps} \ln \frac{\pi^2 t^2}{3 \beta} \leq \tau + \frac{8}{\eps}\ln \left(t + \left\lceil \frac{12}{\eps}\right\rceil\right) - \frac{2}{\eps} \ln \frac{2}{\beta} \Leftrightarrow x_t \leq \tau - \frac{6}{\eps} \ln \frac{2}{\beta} - \frac{4}{\eps} \ln \frac{\pi^2}{6}. \]
    We see that in the pseudocode of \Cref{alg:svt}, the values $x_t + Z_t$ are compared with the threshold $\tau + 8\ln\left(t+\lceil \frac{12}{\eps}\rceil\right)/\eps + Z_0$. It follows that conditioned on the $(1-\beta)$-probability event described above, the stated accuracy bounds hold.
\end{proof}

\medskip\noindent\textbf{Probability theory.}
In this paper, we use the following probability facts: 
\begin{fact}\label{fact:sum-conjunctions}
    For every two events $\calE_1$ and $\calE_2$, we have
    \begin{align*}
        \Pr[\calE_1]=\Pr[\calE_1 \wedge \calE_2] + \Pr[\calE_1\wedge \bar{\calE}_2],
    \end{align*}
    where $\bar{\calE}_2$ denotes the complement of the event $\calE_2$.
\end{fact}
\begin{fact}\label{fact:sum-conditions}
    Let $Z$ be a random variable with a measurable domain $\calZ$, and let $Y$ be a random variable with a finite domain $\calY$. Define $\text{supp}(\calY)=\{y\in\calY\mid \Pr[Y=y]>0\}$. For every measurable set $S\subseteq \calZ$, we have
    \[
    \Pr[Z\in S]=\sum_{y\in\text{supp}(\calY)}\Pr[Y=y]\Pr[Z\in S\mid Y=y].
    \]
\end{fact}

An important probability distribution that we will use is the Laplace distribution:
\begin{definition}[Laplace Distribution]\label{def:laplace}
    The Laplace distribution, centered at $0$ with scale parameter $b$, is the distribution with probability density function $f_{\Lap(b)}(x)=\frac{1}{2b}\exp\left(\frac{-|x|}{b}\right)$. We denote a random variable distributed according to this density as $Y\sim \Lap(b)$ or simply $\Lap(b)$.
\end{definition}

\begin{lemma}\label{lem:laplace-concentration-bound}
    Let $Y \sim \Lap(b)$. Then for every $t\geq 0$,
    \begin{align*}
        \Pr[Y<-t\cdot b]= \Pr [Y>t\cdot b]\leq \frac{1}{2}e^{-t}.
    \end{align*}
\end{lemma}

We also need a discrete version of the Laplace distribution over $\Z$ (instead of $\R$), which is defined as follows:
\begin{definition}[Discrete Laplace Distribution]\label{def:discrete-laplace}
    Let $Y$ denote a random variable drawn from the discrete Laplace distribution centered at $0$ with scale parameter $b\in \R^+$, denoted by $Y\sim \DLap(b)$. Then for every $y\in\Z$, we have
    $$\Pr[Y=y] = \frac{e^{1/b} - 1}{e^{1/b} + 1} \cdot e^{-|y|/b}.$$
\end{definition}

\begin{corollary}\label{cor:discrete-laplace-privacy-accuracy}
    For $b\in \R^+$, let $Y\sim \DLap(b)$. The discrete Laplace distribution satisfies the following properties:
    \begin{enumerate}
        \item For every $y\in \Z$, 
        $$e^{-1/b} \Pr[Y=y-1]\leq \Pr[Y=y]\leq e^{1/b} \Pr[Y=y+1].$$
        \item For every $\tau\in \N$,
        \begin{align*}
            \Pr[Y\leq -\tau]= \Pr[Y\geq \tau]\leq \frac{1}{2}e^{-(\tau-1)/b}.
        \end{align*}
    \end{enumerate}
\end{corollary}

\section{Mechanism \texorpdfstring{$\ecc$}{ECC}}\label{sec:ecc}

We now give a formal description of our sparsity-adaptive mechanism for private continual counting for edit neighboring streams $\ecc$ (\Cref{alg:ecc}). The description of $\simecc$ is almost identical, to avoid redundancy we relegate a formal treatment of the latter to \Cref{sec:simecc}. $\ecc$ accepts as input a stream of values $x_1,x_2,\dots$ drawn from $[0,1]\cup \{\bot\}$ and, after each input $x_t$, outputs an estimate $y_t$ of $s_t = \sum_{i\leq t} x_i$, with $\bot$ identified with $0$. 

\begin{algorithm}
    \begin{algorithmic}
        \Function{Initialize}{Privacy parameters $\eps$ and $\delta$}
            \State $\ell \leftarrow 1$
            \State $t_0 \leftarrow 0$
            \State $y_0 \leftarrow 0$
            \State $\intsum_\ell \leftarrow 0$
            
            \State $\sched \leftarrow \slvar.\init()$
            \State $\partt \leftarrow \partt.\init(\eps/51,\tau = 1+\frac{408\ln(\frac{\pi^2}{2 \delta})}{\eps})$
            \State $\cc \leftarrow \cc.\init(4\eps/27,e^{-19\eps/27}\delta/16)$
            \State $\bcc \leftarrow \bcc.\init(4\eps/27,e^{-19\eps/27}\delta/8)$
        \EndFunction
        \Function{Update}{$x_t \in [0,1]$}
            \State $\intsum_\ell \leftarrow \intsum_\ell + x_t$
            \If{$\partt.\upd(x_t) = \top$}
                \State $t_\ell\leftarrow t$
                \State $\widehat{t}_\ell \leftarrow \bcc(t_\ell - t_{\ell-1})$
                \State $\widehat{v}_\ell \leftarrow \cc(\intsum_\ell)$ 
                \State $\sched.\upd(\widehat{t}_\ell, \widehat{v}_\ell)$
                \State $\ell \leftarrow \ell+1$
                \State $\intsum_\ell \leftarrow 0$ 
            \EndIf
            \State $y_t = y_{t-1}$
            \If{$\sched.\chk(t) \neq \bot$}
                \State $y_t \leftarrow \sched.\chk(t)$
            \EndIf
            \State Output $y_t$
        \EndFunction
    \end{algorithmic}
    \caption{$\ecc$ (input data stream $(x_1, x_2, \dots)$, privacy parameter $\eps,\delta$, black-box access to a continual counter $\cc$ and its biased variant $\bcc$)}\label{alg:ecc}
\end{algorithm}

We enumerate the internal variables used in this mechanism which we shall work with in the privacy and accuracy analyses.

\begin{definition}
    \begin{enumerate}
        \item $x_t \in [0,1]\cup \{\bot\}$ - the input received at time-step $t$.
        \item $s_t \in \R$ - the $t$-th prefix sum that equals $\sum_{i\leq t} x_i$, with $x_i=\bot$ treated as $0$.
        \item $a_t \in \{\top,\bot\}$ - the $\ell$-th output of $\partt$ subroutine that indicates whether $t$ should be marked as a checkpoint.
        \item $t_\ell$ - the $\ell$-th \emph{(true) checkpoint} i.e. $t = t_1, t_2, \dots$ are exactly the time steps at which $a_{t}=\top$.
        \item $\widehat{t}_\ell$ - the $\ell$-th \emph{noisy checkpoints} i.e. the privatized value of $t_\ell$ which is generated by $\bcc$.
        \item $\intsum_\ell \in \R$ - the $\ell$-th interval sum, i.e., $\sum_{t=t_{\ell-1}+1}^{t_\ell}x_t$.
        \item $\widehat{v}_\ell$ - a privatized estimate of $v_\ell=s_{t_\ell}$ generated by the subroutine $\cc$.
        \item $\sched$ - an instance of the $\slvar$ data structure used to store noisy checkpoint values.
        \item $y_t \in \R$ - the value returned by $\ecc$ in step $t$.
    \end{enumerate}
\end{definition}

The construction of \ecc can be broken down into three concurrent parts: (i) we use the stream partitioning algorithm $\partt$ of~\cite{dwork2015pure} to privately identify time-steps $t_1, t_2, \dots$ called \emph{true checkpoints} at which the prefix sum $s_t = x_1 + \dots x_{t_\ell}$ has changed significantly compared to the previous most recent checkpoint and an update is necessitated; (ii) we use a private continual counter $\cc$ that is given as input a stream of \emph{interval sums} $\intsum_\ell := x_{t_{\ell-1}+1} + \dots + x_{t_\ell}$ (one for each checkpoint), and generates an output $\widehat{v}_\ell$ which is a privatized value of $s_{t_{\ell}}$ which is the true prefix sum at $t_\ell$; (iii) we use an auxiliary continual counter $\bcc$ to which we feed the difference sequence $t_\ell-t_{\ell-1}$ and get \emph{noisy checkpoints} $\widehat{t}_\ell$, which are privatized estimates of $t_\ell$. $\ecc$ outputs the privatized prefix sums $\widehat{v}_\ell$ at the noisy checkpoints $\widehat{t}_\ell$. The continual counter $\cc$ is used in a black-box manner; we will appeal to the guarantees of the counter described in Lemma~\ref{lem:cc-standard}, and the continual counter $\bcc$ is a biased continual counter such as the one defined in Lemma~\ref{lem:cc-biased}.

Concretely, the mechanism functions as follows. Upon initialization, it does the following:
\begin{enumerate}
    \item It initializes three subroutines in turn; $\partt$, $\cc$, and $\bcc$.
    \item It defines an initial true checkpoint $t_0 := 0$.
    \item It defines the counter $\ell$ (equaling $1$ at initialization) which keeps track of the index of the \emph{next} checkpoint to be declared.
    \item It initializes the first interval sum $\intsum_1$ to $0$.
    \item It initializes a \emph{schedule list}, implemented here as an instance $\sched$ of the $\slvar$ data structure (a dynamically allocated array of lists, empty at initialization) which will store for every time-step $t$ all values $\widehat{v}_i$ for which $\widehat{t}_i = t$. This data structure only receives values that are generated by differentially private mechanisms and is hence just a post-processing step with regard to the privacy analysis.
    \item It defines $y_0 = 0$, the output value at the virtual time-step $t_0$ that simply defines the default value of the sum in the absence of any input.
\end{enumerate}  

When an input $x_t$ is received (i.e. the Update$(x_t)$ subroutine of $\ecc$ is called), the mechanism passes $x_t$ to $\partt$, receiving as output a boolean value $a_t \in \{\top,\bot\}$. If $a_t = \bot$ then it proceeds to check the schedule list to see what the output ought to be (this will be explained further shortly). On the other hand, if $a_t = \top$, then $t$ has been identified to be the $\ell$-th true checkpoint $t_\ell$. Informally, this means that the output ought to be updated, but for reasons of privacy the update is delayed to the corresponding noisy checkpoint $\widehat{t}_\ell$. The value of $\widehat{t}_\ell$ is generated by $\bcc$, which is given as an update the checkpoint-difference $d_\ell= t_\ell-t_{\ell-1}$, and in turn outputs the privatized estimate of $t_\ell = \sum_{i=1}^{\ell} d_i$. Since an update cannot be scheduled for the past, $\bcc$ is constructed to have the property that it only ever outputs overestimates of the sum that it is privatizing. In other words, it is never the case that $\widehat{t}_\ell < t_\ell$. Furthermore, the estimated prefix-sum $\widehat{v}_\ell$ is generated by $\cc$, which is given as an update the interval-sum $\intsum_\ell=s_{t_\ell} - s_{t_{\ell-1}}$, and in turn outputs the privatized estimate  $\widehat{v}_\ell$ of $v_\ell:=s_{t_\ell}=\sum_{i=1}^\ell \intsum_i$. An instance $\sched$ of the schedule list data structure $\slvar$ is then used to schedule this output update $\widehat{v}_\ell$ for time step $\widehat{t}_\ell$ by calling the $\sched.\upd$ method of $\slvar$ on the inputs $\widehat{t}_\ell$ and $\widehat{v}_\ell$.

Finally, $\sched$ is queried to determine the output value $y_t$ for the current time-step. If $\sched.\chk(t)$ returns a $\bot$ value, then the current time-step was not identified as a noisy checkpoint, and we simply return the same value as the previous time-step, i.e. $y_t = y_{t-1}$. On the other hand, if $\sched.\chk(t) \not=\bot$, then the value returned by the data structure is in fact the noisy sum value itself that was scheduled to be released, and $y_t = \sched.\chk(t)$.

The following observation relates the inputs of the subroutines $\cc$ and $\bcc$ to the input of $\ecc$ and the checkpoints $t_\ell$ generated during runtime.

\begin{observation}\label{obs:inputs-cc-bcc}
    Consider an execution of the mechanism $\ecc$ on the input stream $\sigma=(x_1, \dots, x_T)$, and suppose that $t_1, \dots, t_k$ are the true checkpoints indicated by $\partt$ for some $k\in [T]$. Then, by construction, the subroutine $\cc$ receives as input $\intsum=(\intsum_1,\dots,\intsum_k)$ where $\intsum_\ell = \sum_{t =t_{\ell-1} + 1}^{t_\ell} x_t$ for $\ell\in [k]$. Similarly, $\bcc$ receives as input $d=(d_1,\dots,d_k)$ where $d_\ell = t_\ell - t_{\ell-1}$ for $\ell\in [k]$. 
\end{observation}

\medskip
\noindent\textbf{I. Subroutine $\partt$}

\begin{algorithm}
    \begin{algorithmic}
        \Function{Initialize}{Privacy parameter $\eps$, threshold $\tau$}
            \State $\ell\leftarrow 1$
            \State $\intsum_\ell\leftarrow 0$
            \State $\svt_\ell \leftarrow \svt.\init(\eps,\tau+16\ln(\ell)/\eps)$
        \EndFunction
        \Function{Update}{$x_t \in [0,1]$}
            \State $\intsum_\ell \leftarrow \intsum_\ell + x_t$
            \If{$\svt_\ell(\intsum) = \top$}
                \State $\ell \leftarrow \ell + 1$
                \State $\intsum_\ell \leftarrow 0$
                \State $\svt_\ell \leftarrow \svt.\init(\eps,\tau+16\ln(\ell)/\eps)$
                \State Output $\top$
            \Else  
                \State Output $\bot$
            \EndIf
        \EndFunction
    \end{algorithmic}
    \caption{$\partt$}\label{alg:partition}
\end{algorithm}

The mechanism $\partt$, described in Algorithm~\ref{alg:partition}, is initialized with privacy parameter $\eps$ and threshold $\tau$. It first creates an instance of $\svt$, denoted $\svt_1$, with privacy parameter $\eps$ and threshold $\tau + \frac{16\ln(1)}{\eps} = \tau$. Then it initializes a counter $\ell = 1$ and an interval-sum variable $\intsum_1 = 0$. The counter $\ell$ keeps track of the number of SVT instances created so far, and its value is in fact the same as that of the counter $\ell$ in $\ecc$ (we avoid pointers for ease of comprehension). The variable $\intsum_\ell$ stores the cumulative sum of inputs received since the current SVT instance, $\svt_\ell$, was initialized. When this instance outputs $\top$ and terminates, the mechanism $\partt$ increments $\ell \leftarrow \ell + 1$, initializes the next SVT instance $\svt_\ell$ with privacy parameter $\eps$ and threshold $\tau + \frac{16\ln(\ell)}{\eps}$, and sets $\intsum_\ell \leftarrow 0$.

Upon receiving a new input $x_t$ (i.e., when $\ecc$ invokes the \textsc{Update} subroutine of $\partt$), the mechanism $\partt$ proceeds as follows. First, it updates the running sum by setting $\intsum_\ell \leftarrow \intsum_\ell + x_t$. It then feeds $\intsum_\ell$ to the active instance $\svt_\ell$, which in turn conducts a private threshold test and compares $\intsum_\ell$ with $\tau+16\ln(\ell)/\eps$. If $\svt_\ell$ outputs $\top$ (informally, indicating that $\intsum_\ell$ is larger than $\tau+16\ln(\ell)/\eps$), then $\partt$ increases $\ell$ by $1$, initializes $\svt_\ell$ as described above, sets the new interval sum $\intsum_\ell$ to $0$, and returns $\top$ to $\ecc$. Otherwise, if $\svt_\ell$ outputs $\bot$, then $\partt$ returns $\bot$.

\medskip
\noindent\textbf{\texorpdfstring{II. $\slvar$}{schedList} data structure}

\begin{algorithm}
    \begin{algorithmic}
        \Function{Initialize}{}
            \State $\sched \leftarrow \text{a dynamically allocated array}$
            \State return $\sched$
        \EndFunction
        \Function{Update}{$\widehat{t}, \widehat{v}$}
            \State $\sched[\widehat{t}].append(\widehat{v})$
        \EndFunction
        \Function{Check}{$t$}
            \If{$\sched[t] = \emptyset$}
                \State return $\bot$
            \Else
                \State return $\max \{ \widehat{v} : \widehat{v} \in \sched[t]\}$
            \EndIf
        \EndFunction
    \end{algorithmic}
    \caption{$\slvar$ data structure}
\end{algorithm}

$\slvar$ is a simple data structure that we use to keep track of all noisy checkpoints declared, and the respective privatized estimates of the prefix sum that ought to be released when a noisy checkpoint is reached. 
Thus it is a data structure that does not apply any randomization itself and simply post-processes the noisy input it is given. We implement this as a simple array of lists, indexed over time-steps. When a checkpoint $t_\ell$ is reached, the  privatized sum value $\widehat{v}_\ell$ for $t_\ell$  is appended to the list stored at the index of the corresponding noisy checkpoint $\widehat{t}_\ell$. Further, when one checks the data structure to see if any updates have been scheduled for time step $t$ (i.e. in other words, if $t$ were marked as a noisy checkpoint one or more times). The data structure returns a $\bot$ value if $t$ were not marked as a checkpoint, and applies a tie-breaking rule if more than one checkpoints scheduled releases at this time-step. In principle, by the post-processing property of DP, we can use any tie-breaking rule, as long as it does not directly access the true checkpoint time-step values; in this case, we pick the maximum noisy prefix sum.

\subsection{Accuracy Guarantee for \texorpdfstring{\ecc}{ECC}}\label{subsec:ecc-accuracy}

In this section, we will prove the following theorem.

\begin{theorem}
    \label{thm:ecc-accuracy}
    We condition on the event $\calE_{\partt}$ defined in \cref{lem:partt-acc}, and on the $1-\beta_C$ and $1-\beta_B$ events that the error bounds of $\cc$ and $\bcc$ hold. Conditioned on these events, for all $t\in \N$, 
        \[ |s_t - y_t| \leq 2 E_\bcc (s_t) + E_\cc (s_t) + O(\tfrac{1}{\eps}\ln(t/\delta)).\]
\end{theorem}

We recall some notation and internal variables of \ecc that we will use often in our accuracy analysis.
\begin{itemize}
    \item $t_\ell$, $\widehat{t}_\ell$: The mechanism \ecc defines a sequence of true checkpoints $t_\ell$ for $\ell\geq1$. It also generates privatized proxies for each true checkpoint, $\widehat{t}_\ell$.
    \item $s_t, y_t$: The objective of the mechanism is to estimate the prefix sums $s_t = \sum_{k=1}^t x_k$ for $t\geq1$. For every time-step $t$, it generates the output $y_t$, and incurs error $|s_t - y_t|$.
    \item $\intsum_\ell, \widehat{v}_\ell$: The mechanism computes interval sums $\intsum_\ell = \sum_{k=t_{\ell-1}+1}^{t_\ell} x_k$, which corresponds to the sum of all values received after checkpoint $t_{\ell-1}$, and before checkpoint $t_\ell$. It computes $\widehat{v}_\ell$ for $\ell\geq1$, where $v_\ell$ is a privatized proxy of $s_{t_\ell}$ by passing to a continual counter \cc the sequence of interval sums $(\intsum_\ell)_{\ell\geq1}$, and $\widehat{v}_\ell$ is the $\ell$-th value generated.
    \item $a_t$: We will find it convenient to denote the output of the \svt instances by $a_t \in \{\bot,\top\}$. $a_t = \bot$ indicates that $t$ was not marked as a true checkpoint, and $a_t = \top$ indicates that $t$ was marked as a checkpoint.
    \item Given a probability $\beta \in [0,1]$, we let $\beta_t = \tfrac{6}{\pi^2}\tfrac{\beta}{t^2}$. With this definition, $\sum_{t=1}^\infty \beta_t = \beta$.
\end{itemize}

We now prove the accuracy guarantee of $\ecc$. We start first by characterizing $\partt$. More concretely, whether a time step is ($a_t = \top$) or isn't ($a_t = \bot$) marked as a true checkpoint internally by the algorithm allows us to determine, with high probability, bounds on the change in the prefix sum since the last checkpoint was declared. This also allows us to bound the total number of checkpoints declared by time step $t$ in terms of the prefix sum value $s_t$.

\begin{lemma}
    \label{lem:partt-acc}
    For every $\beta\in [0,1]$, there exists an event $\calE_{\partt}$ defined over the random coins of $\partt$, such that $\Pr[\calE_\partt] \geq 1-\beta$, and such that conditioned on $\calE_{\partt}$, for all $t\in \N$, the following statements hold:
    \begin{enumerate}[left=0pt]
        \item If $\ell$ is the index of the active SVT instance at time step $t$ then:
        \begin{enumerate}
            \item If the SVT returns $\top$, then $s_t - s_{t_{\ell-1}} \geq \tau - O(\ln(\ell/\beta)/\eps)$.
            \item If the SVT returns $\bot$, then $s_t - s_{t_{\ell-1}} \leq \tau + O(\ln (t + \lceil1/\eps\rceil)/\eps) + O(\ln(\ell/\beta)/\eps)$.
        \end{enumerate}
        \item If $\ell$ denotes the number of checkpoints declared up to time step $t$, then $\ell \leq s_t$.
    \end{enumerate}
\end{lemma}
\begin{proof}
    Throughout this proof, let $\eps_p \;:=\; \eps/51$
    denote the privacy parameter with which $\ecc$ initializes $\partt$ (see Algorithm~\ref{alg:ecc}), which is also the privacy parameter passed to each internal SVT instance (see Algorithm~\ref{alg:partition}). Since $408 = 8 \cdot 51$, the threshold $\tau$ passed to $\partt$ satisfies
    \[
    \tau \;=\; 1 \;+\; \frac{408\ln(\pi^2/(2\delta))}{\eps} \;=\; 1 \;+\; \frac{8\ln(\pi^2/(2\delta))}{\eps_p},
    \]
    and the threshold of the $\ell$-th SVT instance is
    \[
    \tau(\ell) \;=\; \tau + \frac{16\ln(\ell)}{\eps_p}.
    \]
    A checkpoint $t_\ell$ is declared at the time step at which the $\ell$-th SVT instance first outputs $\top$. Recall $\beta_\ell := \tfrac{6\beta}{\pi^2 \ell^2}$, so that $\sum_{\ell=1}^\infty \beta_\ell = \beta$. By Lemma~\ref{lem:svt_accuracy} and a union bound over $\ell \in \N$, there is an event $\calE_\partt$ of probability at least $1-\beta$ conditioned on which the accuracy guarantee of Lemma~\ref{lem:svt_accuracy} holds simultaneously for every SVT instance with its respective failure probability $\beta_\ell$; we condition on $\calE_\partt$ for the remainder of the proof.

    \begin{enumerate}[left=0pt]
        \item Applying Lemma~\ref{lem:svt_accuracy} to the $\ell$-th SVT instance (privacy parameter $\eps_p$, threshold $\tau(\ell)$, failure probability $\beta_\ell$), we obtain that for every $t \geq t_{\ell-1}+1$:
    \begin{enumerate}
        \item[(a$'$)] if $s_t - s_{t_{\ell-1}} \;\leq\; \tau(\ell) - \tfrac{6\ln(2/\beta_\ell)}{\eps_p} - \tfrac{4\ln(\pi^2/6)}{\eps_p}$, then the SVT returns $\bot$;
        \item[(b$'$)] if $s_t - s_{t_{\ell-1}} \;\geq\; \tau(\ell) + \tfrac{16\ln(t + \lceil12/\eps\rceil)}{\eps_p} + \tfrac{6\ln(2/\beta_\ell)}{\eps_p} + \tfrac{4\ln(\pi^2/6)}{\eps_p}$, then the SVT returns $\top$.
    \end{enumerate}
    Substituting $\ln(2/\beta_\ell) = \ln(\pi^2/3) + 2\ln(\ell) + \ln(1/\beta) = O(\ln(\ell/\beta))$ and taking contrapositives yields the two stated bounds (absorbing $16\ln(\ell)/\eps_p$ from $\tau(\ell)$ into the $O(\ln(\ell/\beta)/\eps)$ error term).
    \item For ease of exposition, we define
    \[
    C_k \;:=\; \frac{6\ln(2/\beta_k)}{\eps_p} + \frac{4\ln(\pi^2/6)}{\eps_p}
    \]
    and
    \[
    \kappa \;:=\; 6\ln(\pi^2/3) + 4\ln(\pi^2/6).
    \]
    For each $k \in [\ell]$, the $k$-th SVT instance returned $\top$ at $t_k$, so the contrapositive of (a$'$) gives
    \[
    s_{t_k} - s_{t_{k-1}} \;>\; \tau(k) - C_k.
    \]
    Since the inputs lie in $[0,1] \cup \{\bot\}$ (with $\bot$ treated as $0$), $s_t$ is non-decreasing in $t$, and hence
    \begin{equation}\label{eq:partt-acc-telescope}
    s_t \;\geq\; s_{t_\ell} \;=\; \sum_{k=1}^{\ell} (s_{t_k} - s_{t_{k-1}}) \;>\; \sum_{k=1}^{\ell} \bigl(\tau(k) - C_k\bigr).
    \end{equation}
    It therefore suffices to show that $\tau(k) - C_k \geq 1$ for every $k \geq 1$; summing over $k \in [\ell]$ then gives $s_t > \ell$.

    Expanding using $\ln(2/\beta_k) = \ln(\pi^2/3) + 2\ln(k) + \ln(1/\beta)$,
    \begin{align*}
        \tau(k) - C_k &\;=\; 1 + \frac{8\ln(\pi^2/(2\delta))}{\eps_p} + \frac{16\ln(k)}{\eps_p} \;-\; \frac{12\ln(k) + 6\ln(1/\beta) + \kappa}{\eps_p} \\
        &\;=\; 1 + \frac{8\ln(\pi^2/(2\delta)) - 6\ln(1/\beta) - \kappa}{\eps_p} + \frac{4\ln(k)}{\eps_p}.
    \end{align*}
    Since $\delta \leq \beta$, we have $\ln(1/\delta) \geq \ln(1/\beta)$, so
    \begin{align*}
        8\ln(\pi^2/(2\delta)) - 6\ln(1/\beta)
    &= 8\ln(\pi^2/2) + 8\ln(1/\delta) - 6\ln(1/\beta)\\
    &\geq 8\ln(\pi^2/2) + 2\ln(1/\beta)\\
    &\geq 8\ln(\pi^2/2).
    \end{align*}
    Numerically, $8\ln(\pi^2/2) > 12$ and $\kappa < 10$, so $8\ln(\pi^2/(2\delta)) - 6\ln(1/\beta) - \kappa > 2 > 0$. Combined with $\ln(k) \geq 0$ for $k \geq 1$, this yields
    \[
    \tau(k) - C_k \;\geq\; 1,
    \]
    as required. Substituting into~\eqref{eq:partt-acc-telescope} gives $s_t > \ell$, completing the proof.
    \end{enumerate} 
\end{proof}

\begin{definition}
    Recall that for a given choice of privacy parameters $\eps,\delta$, and failure probability $\beta$, the continual counters $\cc$ and $\bcc$ are $(\alpha_\cc,\beta)$ and $(\alpha_\bcc,\beta)$-accurate respectively for $\alpha_\cc$ and $\alpha_\bcc$ as defined in \Cref{lem:cc-standard} and \Cref{lem:cc-biased}.
    Let $\beta_C$ denote the error probability of $\cc$, and $\beta_B$ denote the error probability of $\bcc$; let $(\eps_C,\delta_C)$ and $(\eps_B,\delta_B)$ denote the privacy parameters of $\cc$ and $\bcc$ respectively. With this notation, we define the terms $E_\cc$ and $E_\bcc$ as follows:
        \begin{align*}
            E_\cc (i) &= \alpha_\cc(\varepsilon_C,\delta_C,\beta_C, i) \\
            E_\bcc(i) &= \alpha_\bcc(\varepsilon_B,\delta_B,\beta_B,i).
        \end{align*}
    With this notation, we see that in a run of $\ecc$, the instance of $\cc$ is $E_\cc (i)$-accurate for an input stream of length $i$ with probability $1-\beta_C$, and the instance of $\bcc$ is $E_\bcc (i)$-accurate, for an input stream of length $i$ with probability $1-\beta_B$.
\end{definition}

\begin{lemma}
    \label{lem:part-accuracy}
    We condition on the event $\calE_{\partt}$ defined in \cref{lem:partt-acc}, and on the $1-\beta_C$ and $1-\beta_B$ events that the error bounds of $\cc$ and $\bcc$ hold. The following statements hold:
    \begin{enumerate}[left=0pt]
        \item For all $\ell\in \N$, $\widehat{t}_\ell - t_\ell \in [0,E_\bcc(\ell)]$.
        \item  
        For all $\ell\in\N$, $|s_{\widehat{t}_\ell} - \widehat{v}_\ell| \leq E_\bcc (\ell) + E_\cc(\ell)$.
    \end{enumerate}
    \end{lemma}
\begin{proof}
\noindent\textit{(1):} We recall that by definition $\widehat{t}_\ell$ is the last output generated by $\bcc$ when given the input stream $((t_1-t_0),(t_2-t_1),\dots,(t_\ell - t_{\ell-1}))$. It follows from the error bound of $\bcc$ that $ \widehat{t}_\ell - t_\ell \in [0,E_\bcc(\ell)]$, which is the first statement.

\noindent\textit{(2):} Since we are considering a stream of values drawn from $[0,1]\cup{\bot}$, it follows directly from the previous part that $|s_{\widehat{t}_\ell} - s_{t_\ell}| \leq E_\bcc(\ell)$. We recall that $\widehat{v}_\ell$ is the output generated by $\cc$ when given the input stream $(\intsum_1,\dots, \intsum_\ell)$, where $\intsum_i = \sum_{j = t_{i-1}+1}^{t_{i}} x_j$. It follows from the error bound of $\cc$ that $ |\widehat{v}_\ell - s_{t_\ell}| \leq E_\cc(\ell)$. Applying the triangle inequality, we get the second statement. 
\end{proof}

We can now prove \Cref{thm:ecc-accuracy}.

\begin{proof}[Proof of \Cref{thm:ecc-accuracy}]
    We want to bound $|s_t - y_t|$.
    By the pseudocode of \ecc, $y_t$ equals the value released at the most recent noisy checkpoint $\widehat{t}$ at or before $t$, i.e. $y_{\widehat{t}}$. 
    We know that $y_{\widehat{t}} = \max\{\widehat{v} : \widehat{v} \in \sched[\widehat{t}]\}$. Let $\ell$ be the index of the checkpoint whose value $\widehat{v}_\ell$ was released at $\widehat{t}$, i.e. $\widehat{t}_\ell = \widehat{t}$, and when $\sched[\widehat t]$ was accessed at $\widehat t$, $\widehat{v}_\ell$ was the value released. 
    It follows then that $y_t = \widehat{v}_\ell$ and $|s_{\widehat{t}} - y_t| \leq E_\bcc (\ell) + E_\cc (\ell)$. 
    Using that (1) $\ell \leq s_t$, (2) $t\geq t_{\ell}$, (3) $s_t$ is non-decreasing in $t$, and (4) $E_\cc$ and $E_\bcc$ are non-decreasing in their arguments, we can write 
    \begin{align} 
        |s_{\widehat{t}} - y_t| \leq E_\bcc (s_t) + E_\cc (s_t). \label{eqn:ecc-accuracy.1}
    \end{align}
    
    Let $t^\dagger := t - E_\bcc(s_t)$. We have the following two cases:

    \textbf{Case 1:} $\widehat{t} \in (t^\dagger,t]$. Since the input values are drawn from $[0,1]\cup\{\bot\}$, it follows that
        \[ s_t - s_{\widehat{t}} \leq E_\bcc (s_t).\] 
    Applying the triangle inequality on \cref{eqn:ecc-accuracy.1} and the display above, we have
        \[ |s_t - y_t| \leq 2 E_\bcc (s_t) + E_\cc(s_t).\]

    \textbf{Case 2:} $\widehat{t} \not\in (t^\dagger,t]$. First, since $s_t-s_{t-1}\leq1$, we can write
    \begin{align}
        s_t - s_{t^\dagger} \leq E_\bcc (s_t). \label{eqn:ecc-accuracy.2}
    \end{align}
    We claim that $\forall i\in (\widehat{t},t^\dagger]$, $a_{i}= \bot$. This is because if $a_i = \top$ for some such $i$ then by the accuracy guarantee of $\bcc$, the corresponding noisy checkpoint $\widehat{i}$ would be greater than $\widehat{t}$ and smaller than $t$, which contradicts that by definition, $\widehat{t}$ is the last noisy checkpoint to occur at or before $t$.
    
    Let $\ell^*$ denote the index of the most recent true checkpoint at or before $t^\dagger$. Since $t_\ell \leq \widehat{t}$ and $\widehat{t}\leq t^\dagger$, we have $\ell^* \geq \ell$. First suppose $\widehat{t} < t^\dagger$. Then since $\forall i\in(\widehat{t},t^\dagger]$ $a_i = \bot$, we have that the $(\ell^*+1)$-th SVT instance is running at step $t^\dagger$, and this instance outputs $\bot$ at step $t^\dagger$. Applying the accuracy guarantee of $\partt$ (Lemma~\ref{lem:partt-acc}, part 1(b)) to this instance,
    \begin{equation}\label{eqn:ecc-accuracy.3}
        s_{t^\dagger} - s_{t_{\ell^*}} \;\leq\; \tau(\ell^*+1) + \tfrac{16\ln(t + \lceil1/\eps\rceil)}{\eps_p} + \tfrac{6\ln(2/\beta_{\ell^*+1})}{\eps_p} + \tfrac{4\ln(\pi^2/6)}{\eps_p} \;=\; O\!\left(\tfrac{1}{\eps_p}\ln(t/\delta)\right),
    \end{equation}
    where we used $\ell^* \leq s(x,t^\dagger) \leq s(x,t) \leq t$ (Lemma~\ref{lem:partt-acc}) to absorb $\ln((\ell^*+1)/\beta)$ into $\ln(t/\delta)$.
    Since $t_{\ell^*} \leq \widehat{t}$ and the prefix sums are non-decreasing,
    \begin{equation}\label{eqn:ecc-accuracy.4}
        s_{t^\dagger} - s_{\widehat{t}} \;\leq\; s_{t^\dagger} - s_{t_{\ell^*}} \;\leq\; O\!\left(\tfrac{1}{\eps_p}\ln(t/\delta)\right).
    \end{equation}
    Now suppose $\widehat{t} = t^\dagger$; in this case \Cref{eqn:ecc-accuracy.4} holds trivially.
    
    Applying the triangle inequality on \cref{eqn:ecc-accuracy.1}, \cref{eqn:ecc-accuracy.4}, and \cref{eqn:ecc-accuracy.2}, we get
    \begin{align*}
        |s_t - y_t| &\leq |s_{\widehat{t}} - y_t| + (s_{t^\dagger} - s_{\widehat{t}}) + |s_t - s_{t^\dagger}|\\
        &\leq 2 E_\bcc(s_t) + E_\cc(s_t) + O\!\left(\tfrac{1}{\eps_p}\ln(t/\delta)\right).
    \end{align*}
    Since $\eps = \Theta(\eps_p)$, the stated bound follows.
\end{proof}

We now state a simple observation that the prefix sum $s_t$ is always bounded from above by the sparsity $s(x,t)$ of the input stream.

\begin{observation}[Sparsity]\label{obv:sparsity}
    Given a stream $x\in ([0,1] \cup \bot)^*$, recall that $s(x,t)$ denotes the sparsity of $x$ until time $t$, and is defined by the following expression:
        \[ s(x,t) \,:=\, \sum_{i=1}^{t} 1(x_i\not=\bot). \]
    Since $s_t = \sum_{i=1}^t x_i$, and $x_i \leq 1$ (recall that $\bot$ is identified with $0$ for the purposes of defining $s_t$), we have that $s_t \leq s(x,t)$.
\end{observation}

\subsection{Privacy Analysis for \texorpdfstring{\ecc}{ECC}}\label{sec:ecc-privacy}
In this section, we will prove the following privacy guarantee for the mechanism $\ecc$:

\begin{theorem}\label{thm:ecc-privacy}
    Let $\eps>0$ and $0<\delta\le 1$. Suppose there exists a continual counting mechanism $\cc$ that is $(4\eps/27, e^{-19\eps/27}\delta/16)$-DP with respect to the $1$-step $1$-neighbor relation, defined in Definition~\ref{def:k-step-delta-neighboring}. Let $\bcc$ denote the biased version of $\cc$ as constructed in Lemma~\ref{lem:cc-biased}. Then, using $\cc$ and $\bcc$ as black boxes, the continual mechanism $\ecc$, described in Algorithm~\ref{alg:ecc}, is $(\eps,\delta)$-DP with respect to the edit neighbor relation, defined in Definition~\ref{def:queue-neighboring}.
\end{theorem}

Recall that the mechanism $\ecc$ uses the partitioning mechanism $\partt$ to identify checkpoints $(t_1, t_2, \dots)$; executes the continual counting mechanism $\cc$ on the interval sums $(\intsum_1, \intsum_2, \dots)$, where $\intsum_\ell$ denotes the sum of the real-valued inputs received during the time steps $\{t_{\ell-1}+1, \dots, t_\ell\}$, to yield the update values $(\widehat v_1, \widehat v_2, \dots)$; and runs the biased continual counting mechanism $\bcc$ on the checkpoint-differences $(t_1-0, t_2-t_1, \dots)$ to compute noisy checkpoints $(\widehat t_1, \widehat t_2, \dots)$. We observe that the full output stream of $\ecc$ is a deterministic post-processing of the paired sequence $\left((\widehat v_1,\widehat t_1), (\widehat v_2,\widehat t_2), \dots\right)$ generated by $\cc$ and $\bcc$:

\begin{observation}\label{obs:ecc-post-cc-and-bcc}
    For $T\in\N$, suppose the mechanism $\ecc$ is executed for $T$ time steps. Then there exists a deterministic post-processing function $\postfunc$ that maps the output sequences of the subroutines $\cc$ and $\bcc$ during this execution to the output stream of $\ecc$.
\end{observation}

\begin{proof}
    For $k\in [T]$, let $(\widehat v_1,\dots,\widehat v_k)$ and $(\widehat t_1,\dots,\widehat t_k)$ denote the output sequences produced by the subroutines $\cc$ and $\bcc$, respectively, during the execution of $\ecc$ for $T$ steps. We define the function $\postfunc$ to map these sequences to a sequence $(y_1,\dots,y_T)$ as follows: Initialize an empty schedule list $L\leftarrow \slvar.\init$. For each $\ell\in[k]$, update the schedule list by executing $L.\upd(\widehat t_\ell,\widehat v_\ell)$. Set $y_0=0$. Then, for each $t\in[T]$, recursively define
    \[
    y_t \;=\;
    \begin{cases}
    y_{t-1}, & \text{if } L.\chk(t)=\bot,\\
    L.\chk(t), & \text{otherwise}.
    \end{cases}
    \]

    Comparing this procedure with the definition of $\ecc$ in Algorithm~\ref{alg:ecc}, we see that $(y_1,\dots,y_T)$ is exactly the output stream produced by $\ecc$ when its subroutines $\cc$ and $\bcc$ generate the sequences $(\widehat v_1,\dots,\widehat v_k)$ and $(\widehat t_1,\dots,\widehat t_k)$, respectively.
\end{proof}

Ideally, for any two edit-neighboring input streams, we would like the induced input streams of both $\cc$ and $\bcc$ to be $1$-step $1$-neighbors. If this held deterministically, then the privacy guarantees of $\ecc$ would follow from basic composition together with the post-processing lemma. However, this property cannot be guaranteed deterministically since the checkpoints produced by $\partt$ and consequently the inputs of $\cc$ and $\bcc$ are random. To resolve this issue, we carefully design a coupling.

Fix any pair of edit-neighboring input streams $\sigma$ and $\sigma'$ of the same length $T\in\N$. We define an injective mapping $f_{\sigma\to\sigma'}$ that, roughly speaking, maps any feasible sequence of checkpoints $(t_1, \dots, t_k)$ to a corresponding checkpoint sequence $(t_1', \dots, t_k')$ such that the following two conditions are satisfied:
\begin{itemize}[left=0pt]
    \item \textbf{Property (I).} When $\ecc$ is executed on $\sigma$ and $\partt$ indicates $(t_1, \dots, t_k)$ as checkpoints, the resulting input stream to $\cc$ (and $\bcc$) is $1$-step $1$-neighbor to the corresponding input stream obtained when $\ecc$ is executed on $\sigma'$ and $\partt$ indicates $(t_1', \dots, t_k')$ as checkpoints.
    \item \textbf{Property (II).} The probability that $\partt$ indicates $(t_1, \dots, t_k)$ as checkpoints on the input stream $\sigma$ is ``close'' to the probability that it indicates $(t_1', \dots, t_k')$ as checkpoints on $\sigma'$.
\end{itemize}
These two properties, together with the fact that $\cc$ and $\bcc$ are private mechanisms, allow us to show that, for any measurable set $\calY$, the probability that the output sequence of $\ecc$ on input $\sigma$ lies in $\calY$ \emph{and} the subroutine $\partt$ identifies $(t_1,\dots,t_k)$ as checkpoints on input $\sigma$ is close to the corresponding probability for $\sigma'$ and the mapped checkpoints
$(t_1',\dots,t_k')$. We then sum over all feasible checkpoint sequences and use the fact that $f_{\sigma\to\sigma'}$ is injective to conclude the desired privacy guarantee for $\ecc$, as stated in Theorem~\ref{thm:ecc-privacy}.

Next, we will define the function $f_{\sigma\to\sigma'}$ and prove the above properties as intermediate lemmas. The function $f_{\sigma\to\sigma'}$ is intended to be a bijection on the space of all valid checkpoint sequences. However, to guarantee that $f_{\sigma\to\sigma'}$ always outputs a valid checkpoint sequence—namely, a strictly increasing sequence of positive integers—we exclude a small subset of pathological sequences from its domain. These excluded cases are handled later separately by directly bounding their probability of occurrence under $\partt$. 

Recall that two streams $\sigma$ and $\sigma'$ are edit-neighboring if one of them is an insertion neighbor of the other one at some time step. In the definition of $f_{\sigma\to\sigma'}$, we distinguish between the case where $\sigma$ is an insertion neighbor of $\sigma'$ and where $\sigma'$ is an insertion neighbor of $\sigma$. This asymmetry is essential to satisfy the neighborhood guarantees required in Property~(I).

\begin{definition}[Function $f_{\sigma\to\sigma'}$]\label{def:f}
    Let $T\in \N$ and define $t_0=0$. Let
    $$\calC_T=\left\{ (t_1,\dots,t_k) \mid k\in [T], t_1,\dots,t_k\in\N, 0<t_1<\dots<t_k\leq T \right\}$$
    denote the set of all checkpoint sequences, and let
    \begin{align*}
        \calC_T^+ = \left\{ (t_1, \dots, t_k)\in \calC_T \mid t_i-t_{i-1}\ge 2 \text{ for all } i\in[k]\right\}
    \end{align*}
    denote the subset of checkpoint sequences whose consecutive gaps are at least~$2$.

    Let $\sigma=(x_1,\dots,x_T)$ and $\sigma'=(x_1',\dots,x_T')$ be two edit-neighboring sequences in $([0,1]\cup\{\bot\})^T$, meaning that one of them is an insertion neighbor of the other one at some step $i\in [T]$. We define a function $f_{\sigma\to\sigma'}:\calC_T^+\to \calC_T$ as follows:
    \begin{itemize}[left=0pt]
        \item Suppose that $\sigma$ is an insertion neighbor of $\sigma'$ at step $i$. Let $j$ be the smallest index in $\{i,\dots,T\}$  such that $x_j'=\bot$, and set $j=T+1$ if no such index exists. For every $(t_1, \dots, t_k)\in \calC_T^+$ and $t_{k+1}=T+1$, let $p\in [k+1]$ be the index satisfying $t_{p-1}<i\leq t_p$, and let $q\in [k+1]$ be the index satisfying $t_{q-1}<j\leq t_q$. Define $f_{\sigma\to\sigma'}(t_1,\dots,t_k)$ as
        \begin{align*}
            \begin{cases}
                (t_1,\dots,t_p, t_{p+1}-1,\dots,t_{q-2}-1, t_{q-1},\dots,t_k), & \text{if } p<q-2\\[2mm]
                (t_1,\dots,t_k), & \text{o.w.}
            \end{cases}
        \end{align*}
        \item Otherwise, suppose that $\sigma'$ is an insertion neighbor of $\sigma$ at step $i$. Let $j$ be the smallest index in $\{i,\dots,T\}$ such that $x_j=\bot$, and set $j=T+1$ if no such index exists. For every $(t_1, \dots, t_k)\in \calC_T^+$ and $t_{k+1}=T+1$, let $p\in [k+1]$ be the index satisfying $t_{p-1}<i\leq t_p$, and let $q\in [k+1]$ be the index satisfying $t_{q-1}<j\leq t_q$. Define $f_{\sigma\to\sigma'}(t_1,\dots,t_k)$ as
        \begin{align*}
            \begin{cases}
            (t_1,\dots,t_{p}, t_{p+1}+1,\dots,t_{q-2}+1, t_{q-1},\dots,t_k), & \text{if } p<q-2,\\[2mm]
            (t_1,\dots,t_k), & \text{o.w.}
            \end{cases}
        \end{align*}
    \end{itemize}
\end{definition}

\begin{observation}\label{obs:injective}
    Let $\sigma$ and $\sigma'$ be two edit-neighboring sequences of length $T\in \N$. Then, the function $f_{\sigma\to\sigma'}$, defined in Definition~\ref{def:f}, is injective.
\end{observation}
\begin{proof}
    We prove the observation for the case where $\sigma$ is an insertion neighbor of $\sigma'$. The reverse case follows by symmetry. Suppose $f_{\sigma \to \sigma'}$ maps a checkpoint sequence $(t_1, \dots, t_k)$ to $(t_1', \dots, t_k')$. Let $i,j \in [T]$ be the indices in Definition~\ref{def:f}, and define $t_{k+1}' := T+1$. Let $p^* \in [k+1]$ be the (unique) index satisfying
    \[
    t_{p^*-1}' < i \le t_{p^*}',
    \]
    and let $q^* \in [k+1]$ be the (unique) index satisfying
    \[
    t_{q^*-1}' < j \le t_{q^*}'.
    \]
    By construction, the indices $p^*$ and $q^*$---defined with respect to the output sequence $(t_1', \dots, t_k')$---are equal to the indices $p$ and $q$ from Definition~\ref{def:f}---which are determined by the input sequence $(t_1, \dots, t_k)$. Thus, the preimage $(t_1, \dots, t_k)$ corresponding to a given $(t_1', \dots, t_k')$ can be uniquely reconstructed as
    \begin{align*}
        \begin{cases}
            (t_1',\dots,t_{p^*}', t_{p^*+1}'+1,\dots,t_{q^*-2}'+1, t_{q^*-1}',\dots,t_k'), & \text{if } p^*<q^*-2\\[2mm]
            (t_1',\dots,t_k'), & \text{o.w.}
        \end{cases}
    \end{align*}
    Therefore, $f_{\sigma \to \sigma'}$ is injective.
\end{proof}

We use the following notations throughout this section:
\begin{notation}\label{not:privacy}
    Consider an execution of the mechanism $\ecc$ on an input stream $\sigma$. Recall that $\ecc$ invokes $\partt$ on $\sigma$ and sets the checkpoints to the time steps at which $\partt$ outputs $\top$. We denote by $\partt(\sigma)$ the resulting sequence of checkpoints. For convenience, we write
    $$t_{1:k}:=(t_1, \dots, t_k).$$
    Conditioned on the event $\partt(\sigma)=t_{1:k}$, the mechanism $\ecc$ feeds the deterministically defined interval-sum sequence $\intsum$ and checkpoint-difference
    sequence $d$ defined in Observation~\ref{obs:inputs-cc-bcc} into the subroutines $\cc$ and $\bcc$, respectively. For clarity, in this section, we denote these sequences by
    $$\intsum(\sigma, t_{1:k}) \quad\text{ and }\quad d(t_{1:k}).$$
\end{notation}

The following observation shows a slightly modified version of Property~(I) for the mapping $f_{\sigma\to\sigma'}$, where $1$-step $1$-neighboring input sequences are replaced with $c$-step $1$-neighboring ones for some constant $c\in \N$. This minor change is later handled by applying group privacy. \Cref{obs:neighbor-mapping-property} follows from the definitions of edit-neighboring sequences (Definition~\ref{def:queue-neighboring}) and interval-sum and checkpoint-difference sequences (Observation~\ref{obs:inputs-cc-bcc} and Notation~\ref{not:privacy}). The formal proof is provided in Appendix~\ref{app:missing-privacy}.

\begin{observation}\label{obs:neighbor-mapping-property}
    For $T\in\N$, let $\sigma=(x_1,\dots,x_T)$ and $\sigma'=(x_1',\dots,x_T')$ be two edit-neighboring sequences in $\left([0,1]\cup \{\bot\}\right)^T$. Let $f_{\sigma\to\sigma'}:\calC_T^+\to\calC$ be the mapping in Definition~\ref{def:f}. Let $t_{1:k}\in\calC_T^+$, and define
    \[
    t_{1:k}' = f_{\sigma\to\sigma'}(t_{1:k}).
    \]
    Let $i,j\in[T]$ and $p,q\in[k]$ be the indices associated with the definition of $f_{\sigma\to\sigma'}(t_{1:k})$. Let the sequences $\intsum(\sigma, t_{1:k})$ and $d(t_{1:k})$ denote the input streams of $\cc$ and $\bcc$ when $\ecc$ executes on $\sigma$ and $\partt(\sigma)=t_{1:k}$, and let $\intsum(\sigma', t_{1:k}')$ and $d(t_{1:k}')$ denote the same sequences for $\sigma'$ and $t_{1:k}'$ (see Observation~\ref{obs:inputs-cc-bcc} and Notation~\ref{not:privacy}). Then, the following statements hold:
    \begin{itemize}[left=0pt]
        \item The sequences $\intsum(\sigma, t_{1:k})$ and $\intsum(\sigma', t_{1:k}')$ are $4$-step $1$-neighbors (see Definition~\ref{def:k-step-delta-neighboring}). 
        \item The sequences $d(t_{1:k})$ and $d(t_{1:k}')$ are $2$-step $1$-neighbors (see Definition~\ref{def:k-step-delta-neighboring}).
    \end{itemize}
\end{observation}

The next lemma shows Property~(II) for the mapping $f_{\sigma\to\sigma'}$. The proof of this lemma is technical and thus deferred to the end of this section.

\begin{lemma}\label{lem:partition}
    Let $\partt$ be the mechanism described in Algorithm~\ref{alg:partition} with privacy parameter $\eps>0$ and threshold parameter $\tau\geq 3$. Then there exists an event $E$ such that
    $$\Pr[E]\geq 1 - \frac{\pi^2}{4}\cdot e^{-\eps(\tau-1)/8},$$
    and for every $T\in\N$ and every pair of edit-neighboring sequences $\sigma=(x_1,\dots,x_T)$ and $\sigma'=(x_1',\dots,x_T')$ in $\left([0, 1]\cup\{\bot\}\right)^T$, the following statements hold: Define the sets $\calC_T$ and $\calC_T^+$ and the function $f_{\sigma\to\sigma'}$ as in Definition~\ref{def:f}.
    \begin{enumerate}[label=(\roman*)]
        \item \label{item:partition-zero-prob} For every checkpoint sequence $t_{1:k}\in\calC_T\setminus\calC_T^+$,
        \[
        \Pr\left[\partt(\sigma)=t_{1:k}\wedge E\right]= 0 .
        \]
        \item \label{item:partition-dp} For every checkpoint sequence $t_{1:k}\in\calC_T^+$,
        \[
        \Pr\left[\partt(\sigma)=t_{1:k}\wedge E\right]
        \leq
        e^{17\eps/3}
        \Pr\left[\partt(\sigma') = f_{\sigma\to\sigma'}(t_{1:k})\right].
        \]
    \end{enumerate}
\end{lemma}

\begin{proof}[Proof of Theorem~\ref{thm:ecc-privacy}]
    We must show that for every $T\in \N$ and every pair of edit-neighboring input streams $\sigma=(x_1, \dots, x_T)$ and $\sigma'=(x_1', \dots, x_T')$ in $\left([0, 1]\cup\{\bot\}\right)^T$, the distributions of the output sequence of $\ecc$ on $\sigma$ and $\sigma'$, denoted by $\ecc(\sigma)$ and $\ecc(\sigma')$, are $(\eps, \delta)$-indistinguishable. That is, for every measurable set $\calY\subseteq \R^T$,
    \begin{align}\label{eq:privacy}
        \Pr[\ecc(\sigma)\in \calY] \leq e^{\eps}\Pr[\ecc(\sigma')\in \calY] + \delta,
    \end{align}
    and 
    \begin{align}\label{eq:privacy-symmetric}
        \Pr[\ecc(\sigma')\in \calY] \leq e^{\eps}\Pr[\ecc(\sigma)\in \calY] + \delta.
    \end{align}
    Fix $T$, $\sigma$, $\sigma'$, and $\calY$. In the rest of this proof, we will show Inequality~\eqref{eq:privacy} holds. Inequality~\eqref{eq:privacy-symmetric} follows by symmetry. 

    By construction, $\ecc$ instantiates $\partt$ with privacy parameter $\parteps:=\eps/51$ and threshold $\tau:=1+\frac{408\ln(\frac{2\pi^2}{3 \delta})}{\eps}$. Since $\tau>1$, Lemma~\ref{lem:partition} is applicable.
    Let $E$ be the event defined in that lemma. By Fact~\ref{fact:sum-conjunctions} and that $\Pr[E]\geq 1-\frac{\pi^2}{4}\cdot e^{-\parteps(\tau-1)/8}$, we have
    \begin{align*}
        \Pr[\ecc(\sigma)\in \calY] &= \Pr[\ecc(\sigma)\in \calY\wedge E] + \Pr[\ecc(\sigma)\in \calY\wedge \overline{E}] \\
        &\leq \Pr[\ecc(\sigma)\in \calY \wedge E] + \Pr[\overline{E}]\\
        &\leq \Pr[\ecc(\sigma)\in \calY\wedge E] + \frac{\pi^2}{4}\cdot e^{-\parteps(\tau-1)/8}.
    \end{align*}
    Let $\calC_T$ and $\calC_T^+$ be as in Definition~\ref{def:f}. We expand the first term in the above inequality by conditioning on all possible values for the checkpoint sequence produced by $\partt$ when the event $E$ holds. Define 
    $$\mathcal{F}=\left\{t_{1:k}\in\calC_T\mid \Pr[\partt(\sigma)=t_{1:k} \wedge E]>0\right\}.$$ 
    Then, by Fact~\ref{fact:sum-conditions}, we have
    \begin{align*}
        \Pr[\ecc(\sigma)\in \calY] \leq\! \sum_{t_{1:k}\in\mathcal{F}}\!\! \Pr[\partt(\sigma)\!=\!t_{1:k} \wedge E]\!\cdot\! \Pr\big[\ecc(\sigma)\in \calY\!\mid\! \partt(\sigma)\!=t\!_{1:k}\wedge E\big]
        + \frac{\pi^2}{4}\cdot e^{-\parteps(\tau-1)/8}.
    \end{align*}
    The mechanism $\ecc$ executes $\partt$ as black box and only uses the checkpoints indicated by this subroutine, without accessing its local variables or depending on its internal randomness. Therefore, conditioned on the checkpoint sequence $\partt(\sigma)=t_{1:k}$, the output distribution of $\ecc$ is independent of whether the event $E$--which is about the local variables of $\partt$--holds or not. Therefore,
    \begin{equation}\label{eq:remove-condition-E}
    \begin{aligned}
        \Pr[\ecc(\sigma)\in \calY]
        \leq \sum_{t_{1:k}\in\mathcal{F}} \!\Big(\!\Pr[\partt(\sigma)=t_{1:k}\wedge E]
        \cdot \Pr\!\big[\ecc(\sigma)\in \calY\mid \partt(\sigma)=t_{1:k}\big]\Big)+\frac{\pi^2}{4}\cdot e^{-\parteps(\tau-1)/8}.
    \end{aligned}
    \end{equation}      
    For $t_{1:k}\in\mathcal{F}$, let $\intsum(\sigma,t_{1:k})$ and $d(t_{1:k})$ be the interval-sum sequence and checkpoint-difference sequence defined in Observation~\ref{obs:inputs-cc-bcc} and Notation~\ref{not:privacy}. Conditioned on $\partt(\sigma)=t_{1:k}$, these sequences are the input streams to the subroutines $\cc$ and $\bcc$. We denote the corresponding output sequences by $\cc(\intsum(\sigma,t_{1:k}))$ and $\bcc(d(t_{1:k}))$.

    Let $\postfunc$ be the deterministic post-processing function defined in Observation~\ref{obs:ecc-post-cc-and-bcc}. By that observation, knowing that $\intsum(\sigma, t_{1:k})$ and $d(t_{1:k})$ are the input streams of $\cc$ and $\bcc$, the output stream $\ecc(\sigma)$ equals $\postfunc\left(\cc(\intsum(\sigma, t_{1:k})),\bcc( d(t_{1:k}))\right)$. Therefore, since $\partt(\sigma)=t_{1:k}$ deterministically implies input streams $\intsum(\sigma, t_{1:k})$ and $d(t_{1:k})$ for $\cc$ and $\bcc$, we have
    \begin{align*}
        \Pr[\ecc(\sigma)\in \calY\mid \partt(\sigma)=t_{1:k}] &= \Pr[\postfunc\left(\cc(\intsum(\sigma, t_{1:k})),\bcc( d(t_{1:k}))\right)\in \calY].
    \end{align*}
    Consequently, we can reformulate Inequality~\eqref{eq:remove-condition-E} as 
    \begin{equation}\label{eq:lhs}
    \begin{aligned}
        \Pr[\ecc(\sigma)\in \calY]
        &\leq \sum_{t_{1:k}\in\mathcal{F}} \Big(\Pr[\partt(\sigma)=t_{1:k}\wedge E]
        \cdot \Pr[\postfunc\left(\cc(\intsum(\sigma, t_{1:k})),\bcc( d(t_{1:k}))\right)\in \calY]\Big)\\
        &\qquad\qquad +\frac{\pi^2}{4}\cdot e^{-\parteps(\tau-1)/8}.
    \end{aligned}  
    \end{equation}
    We now consider the execution of $\ecc$ on the neighbor input stream $\sigma'$. Define 
    $$\mathcal{F}'=\left\{t_{1:k}'\in\calC_T\mid \Pr[\partt(\sigma')=t_{1:k}']>0\right\}.$$ 
    By Fact~\ref{fact:sum-conjunctions} and the same post-processing argument, we have
    \begin{equation}\label{eq:rhs}
    \begin{aligned}
        \Pr[\ecc(\sigma')\in \calY]
        &=\sum_{t_{1:k}'\in\mathcal{F}'}\!\! \Pr[\partt(\sigma')\!=\!t_{1:k}']\cdot\Pr[\ecc(\sigma')\!\in\!\calY\!\mid\! \partt(\sigma')\!=\!t_{1:k}']\\
        &=\sum_{t_{1:k}'\in\mathcal{F}'} \Pr[\partt(\sigma')=t_{1:k}']
        \cdot\Pr[\postfunc\!\left(\cc(\intsum(\sigma'\!, t_{1:k}')),\bcc( d(t_{1:k}'))\right)\!\in\!\calY].
    \end{aligned}
    \end{equation}  
    By Inequality~\eqref{eq:lhs} and Equality~\eqref{eq:rhs}, to show that Inequality~\eqref{eq:privacy} holds and to complete the proof, it suffices to prove
    \begin{equation}\label{eq:final}
        \begin{aligned}
            &\sum_{t_{1:k}\in\mathcal{F}} \Big(\Pr[\partt(\sigma)=t_{1:k}\wedge E]
            \cdot \Pr[\postfunc\left(\cc(\intsum(\sigma, t_{1:k})),\bcc( d(t_{1:k}))\right)\in \calY]\Big)
            + \frac{\pi^2}{4}\cdot e^{-\parteps(\tau-1)/8}\\
            &\leq e^{\eps} \sum_{t_{1:k}'\in\mathcal{F}'} \Big(\Pr[\partt(\sigma')=t_{1:k}']
            \cdot\Pr[\postfunc\left(\cc(\intsum(\sigma', t_{1:k}')),\bcc( d(t_{1:k}'))\right)\in \calY]\Big) + \delta
        \end{aligned}
    \end{equation}
    By definition, $\mathcal{F}$ only includes checkpoint sequences $t_{1:k}$ satisfying $\Pr[\partt(\sigma)=t_{1:k} \wedge E]>0$. By Lemma~\ref{lem:partition}, for every $t_{1:k}\in \calC_T\setminus\calC_T^+$, we have $\Pr[\partt(\sigma)=t_{1:k} \wedge E]=0$. Thus, $\mathcal{F}\subseteq \calC_T^+$. Let $f_{\sigma\to\sigma'}:\calC_T^+\to\calC_T$ be the function defined in Definition~\ref{def:f}. By Lemma~\ref{lem:partition}, for every $t_{1:k}\in\mathcal{F}$,
    \begin{align}\label{eq:checkpoints}
        \Pr[\partt(\sigma)=t_{1:k}\wedge E]\leq e^{17\parteps/3}\Pr[\partt(\sigma')=f_{\sigma\to\sigma'}(t_{1:k})].
    \end{align}

    By construction, the mechanism $\ecc$ sets the privacy parameters of $\cc$ to $\cceps:=4\eps/27$ and $\delta_C:=e^{-19\eps/27}\delta/16$, and by assumption, $\cc$ is $(\cceps, \delta_C)$-DP with respect to the $1$-step $1$-neighbor relation. By Observation~\ref{obs:neighbor-mapping-property}, the sequences $\intsum(\sigma, t_{1:k})$ and $\intsum(\sigma', f_{\sigma\to\sigma'}(t_{1:k}))$ are $4$-step $1$-neighbors. Thus, by Lemma~\ref{lem:group-privacy}, the output sequences of $\cc$ on these inputs are $(4\cceps, 4\cdot e^{4\cceps}\delta_C)$-indistinguishable. 

    \sloppy
    Furthermore, the mechanism $\ecc$ sets the privacy parameters of $\bcc$ to $\bcceps:=\cceps$ and $\delta_B:=2\delta_C$. By Lemma~\ref{lem:cc-biased}, since $\cc$ is $(\cceps, \delta_C)$-DP with respect to the $1$-step $1$-neighbor relation, $\bcc$ is $(\bcceps, \delta_B)$-DP with respect to the $1$-step $1$-neighbor relation. By Observation~\ref{obs:neighbor-mapping-property}, the sequences $d(t_{1:k})$ and $d(f_{\sigma\to\sigma'}(t_{1:k}))$ are $2$-step $1$-neighbors, and thus by Lemma~\ref{lem:group-privacy}, the output sequences of $\bcc$ on these inputs are $(2\bcceps, 2e^{2\bcceps}\delta_B)$-indistinguishable.

    \sloppy
    Therefore, by Lemma~\ref{lem:basic-composition}, the composed random variables $\left(\intsum(\sigma, t_{1:k}), d(t_{1:k})\right)$ and $\left(\intsum(\sigma', f_{\sigma\to\sigma'}(t_{1:k})), d(f_{\sigma\to\sigma'}(t_{1:k}))\right)$ are $(4\cceps+2\bcceps, 4\cdot e^{4\cceps}\delta_C+2\cdot e^{2\bcceps}\delta_B)$-indistinguishable. Therefore, by Lemma~\ref{lem:post}, the post-processed random variables $\postfunc\left(\cc(\intsum(\sigma, t_{1:k})),\bcc( d(t_{1:k}))\right)$ and $\postfunc\left(\cc(\intsum(\sigma', f_{\sigma\to\sigma'}(t_{1:k}))), \bcc(d(f_{\sigma\to\sigma'}(t_{1:k})))\right)$ are also $(4\cceps+2\bcceps, 4\cdot e^{4\cceps}\delta_C+2\cdot e^{2\bcceps}\delta_B)$-indistinguishable. Thus
    \begin{equation}\label{eq:cc-bcc}
    \begin{aligned}
        &\Pr\left[\postfunc\left(\cc(\intsum(\sigma, t_{1:k})),\bcc( d(t_{1:k}))\right)\in \calY\right]\\
        &\qquad\qquad\leq e^{4\cceps+2\bcceps}\Pr\Big[\postfunc\Big(\cc(\intsum(\sigma', f_{\sigma\to\sigma'}(t_{1:k}))), \bcc(d(f_{\sigma\to\sigma'}(t_{1:k})))\Big)\in \calY\Big]\\
        &\qquad\qquad+4\cdot e^{4\cceps}\delta_C+2\cdot e^{2\bcceps}\delta_B
    \end{aligned}
    \end{equation}
    Combining Inequalities \eqref{eq:checkpoints} and \eqref{eq:cc-bcc}, we obtain 
    \begin{align*}
        &\sum_{t_{1:k}\in\mathcal{F}}\Pr[\partt(\sigma)=t_{1:k}\wedge E] \cdot \Pr[\postfunc\left(\cc(\intsum(\sigma, t_{1:k})),\bcc( d(t_{1:k}))\right)\in \calY]\\
        &\leq \sum_{t_{1:k}\in\mathcal{F}} \bigg( e^{17\parteps/3}\Pr[\partt(\sigma')=f_{\sigma\to\sigma'}(t_{1:k})] \\
        &\qquad\qquad\cdot\Big(e^{4\cceps+2\bcceps}\Pr\Big[\postfunc\big(\cc(\intsum(\sigma', f_{\sigma\to\sigma'}(t_{1:k})))
        , \bcc(d(f_{\sigma\to\sigma'}(t_{1:k})))\big)\in \calY\Big]\\
        &\qquad\qquad\qquad+4\cdot e^{4\cceps}\delta_C+2\cdot e^{2\bcceps}\delta_B\Big)\bigg)\\
        &= e^{17\parteps/3+4\cceps+2\bcceps} \sum_{t_{1:k}\in\mathcal{F}} \bigg( \Pr[\partt(\sigma')=f_{\sigma\to\sigma'}(t_{1:k})]\\
        &\qquad\qquad\qquad\qquad\qquad\cdot \Pr\Big[\postfunc\big(\cc(\intsum(\sigma', f_{\sigma\to\sigma'}(t_{1:k}))), \bcc(d(f_{\sigma\to\sigma'}(t_{1:k})))\big)\in \calY\Big] \bigg)
        \\
        &+(4\!\cdot\! e^{\frac{17}{3}\parteps/3+4\cceps}\delta_C+2\!\cdot\! e^{\frac{17}{3}\parteps+2\bcceps}\delta_B)\!\cdot\!\!\!\!\! \sum_{t_{1:k}\in\mathcal{F}}\!\! \Pr[\partt(\sigma')\!=\!f_{\sigma\to\sigma'}\!(t_{1:k})]
    \end{align*}

    By the choice of $\parteps=\eps/51$, $\cceps=\bcceps=4\eps/27$, $\delta_C=e^{-19\eps/27}\delta/16$, and $\delta_B=2\delta_C$, we have
    $$17\parteps/3+4\cceps+2\bcceps=\eps$$
    and
    $$4\cdot e^{17\parteps/3+4\cceps}\delta_C+2\cdot e^{17\parteps/3+2\bcceps}\delta_B\leq (4+2\times 2)\cdot e^{19\eps/27} \delta_C = \delta/2.$$
    Thus,
    \begin{align*}
        &\sum_{t_{1:k}\in\mathcal{F}} \Big(\Pr[\partt(\sigma)=t_{1:k}\wedge E]\cdot \Pr\Big[\postfunc\big(\cc(\intsum(\sigma, t_{1:k})), \bcc( d(t_{1:k}))\big)\in \calY\Big]\Big)\\
        &\leq e^{\eps} \sum_{t_{1:k}\in\mathcal{F}}\Big( \Pr[\partt(\sigma')=f_{\sigma\to\sigma'}(t_{1:k})]\\
        &\qquad\qquad\quad\cdot\Pr\Big[\postfunc\big(\cc(\intsum(\sigma', f_{\sigma\to\sigma'}(t_{1:k}))),\bcc(d(f_{\sigma\to\sigma'}(t_{1:k})))\big)\in \calY]\Big)\\
        &+\frac{\delta}{2} \cdot \sum_{t_{1:k}\in\mathcal{F}} \Pr[\partt(\sigma')=f_{\sigma\to\sigma'}(t_{1:k})]
    \end{align*}
    By the definition of $\mathcal{F}$, every $t_{1:k}\in\mathcal{F}$ satisfies $\Pr[\partt(\sigma)=t_{1:k}\wedge E]>0$. By Inequality~\eqref{eq:checkpoints}, this implies that $\Pr[\partt(\sigma')=f_{\sigma\to\sigma'}(t_{1:k})]>0$, and therefore $f_{\sigma\to\sigma'}(t_{1:k})\in \mathcal{F}'$.
    By Observation~\ref{obs:injective}, the function $f_{\sigma\to\sigma'}:\calC_T^+\to \calC_T$ is injective. Thus, due to the positivity of every summand, we have
    \begin{align*}
        &\sum_{t_{1:k}\in\mathcal{F}}\Pr[\partt(\sigma')=f_{\sigma\to\sigma'}(t_{1:k})]
        \cdot\Pr\Big[\postfunc\big(\cc(\intsum(\sigma', f_{\sigma\to\sigma'}(t_{1:k}))), \bcc(d(f_{\sigma\to\sigma'}(t_{1:k})))\big)\in \calY\Big]\\
        &\leq \sum_{t_{1:k}'\in\mathcal{F}'} \Pr[\partt(\sigma')=t_{1:k}']\cdot\Pr\Big[\postfunc\big(\cc(\intsum(\sigma', t_{1:k}')), \bcc(d(t_{1:k}'))\big)\in \calY\Big],
    \end{align*}
    and
    \begin{align*}
        \sum_{t_{1:k}\in\mathcal{F}} \Pr[\partt(\sigma')=f_{\sigma\to\sigma'}(t_{1:k})] \leq \sum_{t_{1:k}'\in\mathcal{F}'} \Pr[\partt(\sigma')=t_{1:k}'] = 1
    \end{align*}
    Combining the last three inequalities implies
    \begin{align*}
        &\sum_{t_{1:k}\in\mathcal{F}} \Pr[\partt(\sigma)=t_{1:k}\wedge E]\cdot \Pr\Big[\postfunc\big(\cc(\intsum(\sigma, t_{1:k})),\bcc( d(t_{1:k}))\big)\in \calY\Big]\\
        &\leq e^{\eps}\!\! \sum_{t_{1:k}'\in\mathcal{F}'} \!\!\Pr[\partt(\sigma')=t_{1:k}']\!\cdot\!\Pr\Big[\postfunc\big(\cc(\intsum(\sigma', t_{1:k}')),\bcc( d(t_{1:k}'))\big)\in \calY\Big]\!+\!\frac{\delta}{2}.
    \end{align*}
    Thus, to prove Inequality~\eqref{eq:final} and finish the proof, it suffices to show
    $$\frac{\pi^2}{4}\cdot e^{-\parteps(\tau-1)/8}\leq \frac{\delta}{2},$$
    which is true by the choice of $\tau=1+\frac{408\ln(\frac{\pi^2}{2 \delta})}{\eps}= 1+\frac{4\ln(\frac{\pi^2}{2 \delta})}{\parteps}$.
\end{proof}

It remains to prove Lemma~\ref{lem:partition}. The goal of this lemma is to show that for every pair of edit-neighboring input streams $\sigma$ and $\sigma'$, the probabilities $\Pr[\partt(\sigma)=t_{1:k}]$ and $\Pr[\partt(\sigma')=t_{1:k}']$ are close, where $t_{1:k}\in \calC_T^+$ is a checkpoint sequence and $t_{1:k}'=f_{\sigma\to\sigma'}(t_{1:k})$. 

Recall that the mechanism $\partt$ maintains an instance of $\svt$ and reinitializes it whenever it outputs $\top$, determining the checkpoint steps. Consider the input streams of these $\svt$ instances from initialization until they output $\top$. The high-level idea of the proof is to couple the output sequences of $\partt(\sigma)$ and $\partt(\sigma')$ using the function $f_{\sigma\to\sigma'}$ in Definition~\ref{def:f} and then show that for every pair of coupled output sequences, the corresponding input streams of all $\svt$ instances executed by $\partt(\sigma)$ and $\partt(\sigma')$ must have been identical, except for at most four instances: the input streams of two $\svt$ instances under $\sigma$ and $\sigma'$ were all-step neighbors as in Definition~\ref{def:k-step-delta-neighboring}, and the input streams of the other two were $1$-shift neighbors as in Definition~\ref{def:one-shift-delta-neighbor}. We apply Lemma~\ref{lem:svt-privacy} to show that the probabilities that $\svt$ outputs $\top$ at the end of two neighbor streams of the former type are close. We then formalize the latter notion via a new neighbor relation and show that the probabilities that $\svt$ outputs $\top$ at the end of two such neighboring input streams are close. Then we use this result to compare $\Pr[\partt(\sigma)=t_{1:k}]$ and $\Pr[\partt(\sigma')=t_{1:k}']$.

We recall that, for an input stream $\sigma$ of length $T\in\N$, $\svt(\sigma)\in [T+1]$ denotes the time step at which the mechanism outputs $\top$, with $\svt(\sigma)=T+1$ if the mechanism never outputs $\top$. The following lemma, proved in Section~\ref{sec:svt-consecutive-outputs}, is required in the proof of Lemma~\ref{lem:partition}.

\begin{lemma}[Consecutive Output Distributions]\label{lem:consecutive-output-distributions}
    Let $\svt$ be the mechanism described in Algorithm~\ref{alg:svt} with privacy parameter $\eps>0$ and threshold parameter $\tau>0$. Recall the random variable $Z_0\sim\Lap(2/\eps)$ from that algorithm, and let $Y\sim\Lap(4/\eps)$ be an independent random variable. Define the event
    $$E^*:=\{Z_0\ge -(\tau-1)/2\}\wedge\{Y\le (\tau-1)/2\}.$$
    For $t\in\N$, let $\sigma=(x_1,\dots,x_{t+1})$ and $\sigma'=(x_1',\dots, x_t')$ be two real-valued streams such that $\sigma$ is a $1$-shift $1$-neighbor of $\sigma'$ (see Definition~\ref{def:one-shift-delta-neighbor}). Then,
    \begin{enumerate}[label=(\roman*)]
        \item \label{item:t+1-leq-t} $\Pr[\svt(\sigma)=t+1]\leq e^{3\eps/2}\Pr[\svt(\sigma')=t]$, and
        \item \label{item:t-leq-t+1}$\Pr[\svt(\sigma')=t \wedge E^*]\leq e^{7\eps/6}\Pr[\svt(\sigma)=t+1]$.
    \end{enumerate}
\end{lemma}

\begin{proof}[Proof of Lemma~\ref{lem:partition}]
    {
    Consider an execution of the mechanism $\partt$ with privacy parameter $\eps$ and threshold parameter $\tau$. For $\ell\in \N$, the $\ell$-th $\svt$ instance (potentially) run by this mechanism has privacy parameter $\eps$ and threshold parameter  
    $$\tau(\ell):=\tau+\frac{16\ln(\ell)}{\eps}.$$ 
    Recall the random variables $Z_0\sim\Lap(2/\eps)$ and $Z_1\sim\Lap(4/\eps)$ used by $\svt$ in Algorithm~\ref{alg:svt}. We denote by $Z_0^\ell$ and $Z_1^\ell$ the corresponding noise variables used by the $\ell$-th (potential) instance of $\svt$. Let $Y_\ell\sim \Lap(4/\eps)$ and define the event
    $$E_\ell:=Z_0^\ell\ge -(\tau(\ell)-1)/2\ \wedge \ Y_\ell\le (\tau(\ell)-1)/2 \ \wedge \  Z_1^\ell\leq (\tau(\ell)-1)/2.$$
    Define
    \[
    E:=\bigwedge_{\ell=1}^\infty E_\ell.
    \]
    By Lemma~\ref{lem:laplace-concentration-bound}, we have $\Pr[Z_0^\ell\ge -(\tau(\ell)-1)/2]\ge 1-\tfrac{1}{2}e^{-\eps(\tau(\ell)-1)/4}$, $\Pr[Z_1^\ell\le (\tau(\ell)-1)/2]\ge 1-\tfrac{1}{2}e^{-\eps(\tau(\ell)-1)/8}$, and $\Pr[Y_\ell\leq (\tau(\ell)-1)/2]\ge 1-\tfrac{1}{2}e^{-\eps(\tau(\ell)-1)/8}$. Hence, applying a union bound over all $\ell\in\N$ yields
    \begin{align*}
    \Pr[E] &\ge 1-\sum_{\ell=1}^\infty \frac{3}{2}\cdot e^{-\eps\cdot(\tau(\ell)-1)/8} = 1-\frac{3}{2}\cdot e^{-\eps(\tau-1)/8}\sum_{\ell=1}^\infty e^{-2\ln(\ell)}\\
    &=1-\frac{\pi^2}{4}e^{-\eps(\tau-1)/8},
    \end{align*}
    where we used the fact that $\sum_{\ell=1}^\infty \frac{1}{\ell^2}= \frac{\pi^2}{6}$.
    
    By the design of $\partt$, if an $\svt$ instance executed by this mechanism receives its first input at time $t\in\N$, then this input equals the $t$-th input of $\partt$ when that input is a real value, and equals $0$ when the input is $\bot$. Hence, if all real inputs of $\partt$ are restricted to lie in the range $[0, 1]$, then for every $\ell\in\N$, the first input $x^*$ received by the $\ell$-th (potential) $\svt$ instance has magnitude at most $1$. Hence, if $Z_0^\ell\ge -(\tau(\ell)-1)/2$ and $Z_1^\ell\leq (\tau(\ell)-1)/2$, we have 
    $$x^*+Z_1^\ell \leq 1 + (\tau(\ell)-1)/2 \leq \tau(\ell)- (\tau(\ell)-1)/2 \leq \tau(\ell)+Z_0.$$
    Therefore, if the event $E_\ell$ holds, then the $\ell$-th $\svt$ mechanism does not output $\top$ upon receiving its first input. 
    
    Fix $T\in\N$ and two edit neighboring sequences $\sigma=(x_1, \dots, x_T)$ and $\sigma'=(x_1', \dots, x_T')$ in $([0,1]\cup\{\bot\})^T$. By the definition of $\calC_T^+$, for every $t_{1:k}\in\calC_T\setminus\calC_T^+$, there exists $\ell\in[k]$ such that $t_\ell=t_{\ell-1}+1$. But this implies that if $\partt(\sigma)=t_{1:k}$, then the $\ell$-th $\svt$ instance of $\partt$ must have output $\top$ upon receiving its first input, which is impossible when the event $E$ holds. Consequently, for the input sequence $\sigma$ in $\left([0, 1]\cup\{\bot\}\right)^T$, we have
    \[
    \Pr[\partt(\sigma)=t_{1:k}\wedge E]=0,
    \]
    finishing the proof of part~\ref{item:partition-zero-prob}.
    }
    
    To prove part~\ref{item:partition-dp}, fix $t_{1:k}\in \calC_T^+$. Let $f_{\sigma\to\sigma'}$ be the function defined in Definition~\ref{def:f}, and define indices $i, j\in [T]$ and $p, q\in [k]$ as in the definition of $f_{\sigma\to\sigma'}(t_{1:k})$. Let $t_{1:k}'=f_{\sigma\to\sigma'}(t_{1:k})$. For every $\ell\in [k]$, define the sequences
    \begin{align*}
        S_\ell := \left(\, \sum_{w=t_{\ell-1}+1}^{t_{\ell-1}+h} x_w\, \right)_{h=1}^{t_{\ell}-t_{\ell-1}} \text{ and } \qquad S_\ell' := \left(\, \sum_{w=t_{\ell-1}'+1}^{t_{\ell-1}'+h} x_w'\, \right)_{h=1}^{t_{\ell}'-t_{\ell-1}'},
    \end{align*}
    with $x_w=\bot$ and $x_w'=\bot$ treated as $0$. Similarly, for the (potentially) remaining elements after index $t_k$, define the sequences
    \begin{align*}
        S_{k+1} := \left(\, \sum_{w=t_k+1}^{t_k+h} x_w\, \right)_{h=1}^{T-t_k}  \text{ and } \qquad S_{k+1}' := \left(\, \sum_{w=t_k'+1}^{t_k'+h} x_w'\, \right)_{h=1}^{T-t_k'}.
    \end{align*}
    By construction, $\partt(\sigma) = t_{1:k}$ if and only if the following conditions hold:
    \begin{itemize}[left=0pt]
        \item[(a)] The mechanism $\partt$ instantiates $k+1$ instances of $\svt$.
        \item[(b)] For each $\ell \in [k]$, the $\ell$-th instance processes the sequence $S_\ell$ (of length $t_\ell - t_{\ell-1}$) and outputs $\top$ at its last input, i.e.,
        \[
        \svt(S_\ell) = t_\ell - t_{\ell-1}.
        \]
        \item[(c)] The $(k+1)$-st instance processes $S_{k+1}$ (of length $T - t_k$) and never outputs $\top$, i.e.,
        \[
        \svt(S_{k+1}) = T - t_k + 1.
        \]
        (Note that if $t_k = T$, then the $(k+1)$-st $\svt$ instance is instantiated at time $T$ but receives no inputs. In this case, $S_{k+1}$ is empty, and by Notation~\ref{not:svt}, we have $\Pr[\svt(S_{k+1}) = T - t_k + 1] = 1$.)
    \end{itemize}
    Recall that, by Definition~\ref{def:f}, we set $t_{k+1} := T+1$. With this notation, $S_{k+1}$ has length $t_{k+1} - t_k - 1$, and condition (c) can be equivalently written as
    \[
    \svt(S_{k+1}) = t_{k+1} - t_k.
    \]
    Therefore,
    \begin{align*}
        \partt(\sigma)=t_{1:k} \quad\Leftrightarrow\quad \svt(S_\ell)=t_\ell-t_{\ell-1} \text{ for all } \ell\in [k+1].
    \end{align*}
    Similarly, for $t_{1:k}'\in \calC_T$, we have $\partt(\sigma')=t_{1:k}'$ if and only if $\svt(S_\ell')=t_\ell'-t_{\ell-1}'$ for all $\ell\in [k+1]$.
    Hence, the inequality of part~\ref{item:partition-dp} is equivalent to
    \begin{align*}
        \Pr\left[\bigwedge_{\ell=1}^{k+1} \svt(S_\ell)=t_\ell-t_{\ell-1}\wedge E\right]
        \leq
        e^{17\eps/3} \Pr\left[\bigwedge_{\ell=1}^{k+1} \svt(S_\ell')=t_\ell'-t_{\ell-1}'\right].
    \end{align*}
    \sloppy
    By the independence of $\svt$ instances, the probability term on the right-hand side equals $\prod_{\ell=1}^{k+1} \Pr\!\left[\svt(S_\ell')=t_\ell'-t_{\ell-1}'\right]$. Moreover, the left-hand side is upper bounded by
    \begin{align*}
        &\Pr\left[\bigwedge_{\ell=1}^{k+1} \svt(S_\ell)=t_\ell-t_{\ell-1}\wedge E\right] = \Pr\left[\bigwedge_{\ell=1}^{k+1} \big(\svt(S_\ell)=t_\ell-t_{\ell-1}\wedge E_\ell\big) \wedge \bigwedge_{\ell=k+1}^\infty E_\ell\right]\\
        &\leq \prod_{\ell=1}^{k+1} \Pr\big[\svt(S_\ell)=t_\ell-t_{\ell-1}\wedge E_\ell\big].
    \end{align*}
    Therefore, to prove part~\ref{item:partition-dp}, it suffices to show
    \begin{equation}\label{eq:partition-simple}
        \begin{aligned}
        \prod_{\ell=1}^{k+1} \Pr\big[\svt(S_\ell)=t_\ell-t_{\ell-1}\wedge E_\ell\big]
        \leq e^{17\eps/3}\prod_{\ell=1}^{k+1} \Pr\big[\svt(S_\ell')=t_\ell'-t_{\ell-1}'\big].
    \end{aligned}
    \end{equation}
    For every $\ell\in \{1, \dots, p-1\}\cup \{p+2, \dots, q-2\}\cup \{q+1, \dots, k\}$, by Definition~\ref{def:f}, $t_\ell-t_{\ell-1} = t_\ell'-t_{\ell-1}'$, and by Observation~\ref{obs:svt-inputs-neighboring}~\ref{item:svt-inputs-equal}, $S_\ell=S_\ell'$. Hence,
    \begin{align*}
        \Pr\big[\svt(S_\ell)=t_\ell-t_{\ell-1}\wedge E_\ell\big] &\leq \Pr\big[\svt(S_\ell)=t_\ell-t_{\ell-1}\big] \\
        &= \Pr\big[\svt(S_\ell')=t_\ell'-t_{\ell-1}'\big].
    \end{align*}
    Therefore, to prove Inequality~\eqref{eq:partition-simple}, it suffices to show\footnote{In Inequality~\eqref{eq:partition-final}, we take the product over the set $(\{p, \dots, q\} \setminus \{p+2, \dots, q-2\})\cup\{k+1\}$ instead of the simpler set $\{p, p+1, q-1, q, k+1\}$ to correctly deal with border cases. The product must include the index $k+1$ and the first two and last two indices in the range $\{p, \dots, q\}$, whenever they exist. For example, if $p = q$, then the correct index set is $\{p, k+1\}$. However, the simpler expression $\{p, p+1, q-1, q, k+1\}$ would incorrectly include additional indices (such as $p\pm 1$) that lie outside the valid range.}
    \begin{equation}\label{eq:partition-final}
    \begin{aligned}
        &\prod_{\ell\in (\{p, \dots, q\} \setminus \{p+2, \dots, q-2\})\cup\{k+1\}} \SPACESUM\Pr\big[\svt(S_\ell)=t_\ell-t_{\ell-1}\wedge E_\ell\big] \\
        &\leq e^{17\eps/3}\SPACESUM\prod_{\ell\in (\{p, \dots, q\} \setminus \{p+2, \dots, q-2\})\cup\{k+1\}} \SPACESUM\Pr\big[\svt(S_\ell')=t_\ell'-t_{\ell-1}'\big].
    \end{aligned}
    \end{equation}
    In the rest of this proof, we will show Inequality~\eqref{eq:partition-final} holds. We prove this inequality for two cases: $p\leq q-2$ and $p>q-2$. 
    
    \medskip\noindent\textbf{Case $p\geq q-2$.} In this case, the set $(\{p, \dots, q\} \setminus \{p+2, \dots, q-2\})\cup\{k+1\}$ equals $\{p, k+1\}$ if $p= q$, equals $\{p, q, k+1\}$ if $p= q-1$, and equals $\{p, p+1, q, k+1\}$ if $p= q-2$. In all cases, for every $\ell$ in this set, we have $t_\ell-t_{\ell-1} = t_\ell'-t_{\ell-1}'$ by Definition~\ref{def:f}, and the sequences $S_\ell$ and $S_\ell'$ are all-step $1$-neighbors by Observation~\ref{obs:svt-inputs-neighboring}~\ref{item:svt-inputs-all-step-neighbor-k-plus-1},~\ref{item:svt-inputs-all-step-neighbor-p-and-q}, and~\ref{item:svt-inputs-all-step-neighbor-q-equal-p-plus-2}. Hence, by Lemma~\ref{lem:svt-privacy},
    \begin{align*}
        \Pr\big[\svt(S_\ell)=t_\ell-t_{\ell-1}\wedge E_\ell \big]
        \leq \Pr\big[\svt(S_\ell)=t_\ell-t_{\ell-1}\big]
        \leq e^{\eps}\Pr\big[\svt(S_\ell')=t_\ell'-t_{\ell-1}'\big],
    \end{align*}
    which implies Inequality~\eqref{eq:partition-final} as the cardinality of the set $(\{p, \dots, q\} \setminus \{p+2, \dots, q-2\})\cup\{k+1\}$ is at most $4$.

    \medskip\noindent\textbf{Case $q>p+2$.} In this case, the set $(\{p, \dots, q\} \setminus \{p+2, \dots, q-2\})\cup\{k+1\}$ equals $\{p, p+1, q-1, q, k+1\}$. The sequences $\sigma$ and $\sigma'$ are edit neighbors, meaning that one of them is an insertion neighbor of the other one at step $i$. We prove Inequality~\eqref{eq:partition-final}  separately for each case. In both cases, we apply the result of Lemma~\ref{lem:consecutive-output-distributions} with the following modification: For each $\ell\in\{p+1,q-1\}$, consider an application of Lemma~\ref{lem:consecutive-output-distributions} to the $\ell$-th $\svt$ instance executed by $\partt$. By definition, the event $E_\ell$ implies the event $E^* = Z_0^\ell \ge -(\tau-1)/2 \wedge Y_\ell \le (\tau-1)/2$ in that lemma. Consequently, the left-hand side of Lemma~\ref{lem:consecutive-output-distributions}~\ref{item:t-leq-t+1} does not increase if we replace $E^*$ with $E_\ell$. Therefore, the inequality continues to hold with $E_\ell$ in place of $E^*$.

    \smallskip\noindent\textit{Case 1: $\sigma$ is an insertion neighbor of $\sigma'$.} In this case, by Definition~\ref{def:f}, $t_{p+1}-t_{p} = t_{p+1}'-t_{p}'+1$, and by Observation~\ref{obs:svt-inputs-neighboring}~\ref{item:svt-inputs-shift-neighbor}, the sequence $S_{p+1}$ is a $1$-shift $1$-neighbor of $S_{p+1}'$. Therefore, by Lemma~\ref{lem:consecutive-output-distributions}~\ref{item:t+1-leq-t},
    \begin{align*}
        \Pr\big[\svt(S_{p+1})=t_{p+1}-t_{p}\wedge E_{p+1}\big]
        &\leq \Pr\big[\svt(S_{p+1})=t_{p+1}-t_{p}]\\
        &\leq e^{3\eps/2}\Pr\big[\svt(S_{p+1}')=t_{p+1}-t_{p}-1\big]\\
        &= e^{3\eps/2} \Pr\big[\svt(S_{p+1}')=t_{p+1}'-t_{p}'\big].
    \end{align*}
    Moreover, by Definition~\ref{def:f}, $t_{q-1}-t_{q-2} = t_{q-1}'-t_{q-2}'-1$, and by Observation~\ref{obs:svt-inputs-neighboring}~\ref{item:svt-inputs-shift-neighbor}, the sequence $S_{q-1}'$ is a $1$-shift $1$-neighbor of $S_{q-1}$. Therefore, by Lemma~\ref{lem:consecutive-output-distributions}~\ref{item:t-leq-t+1},
    \begin{align*}
        \Pr\big[\svt(S_{q-1})=t_{q-1}-t_{q-2}\wedge E_{q-1}\big]
        &\leq e^{7\eps/6}\Pr\big[\svt(S_{q-1}')=t_{q-1}-t_{q-2}+1\big]\\
        &=e^{7\eps/6}\Pr\big[\svt(S_{q-1}')=t'_{q-1}-t'_{q-2}\big].
    \end{align*}
    Furthermore, for each $\ell\in \{p, q, k+1\}$, by Definition~\ref{def:f}, $t_\ell-t_{\ell-1} = t_\ell'-t_{\ell-1}'$, and by Observation~\ref{obs:svt-inputs-neighboring}~\ref{item:svt-inputs-all-step-neighbor-k-plus-1},~\ref{item:svt-inputs-all-step-neighbor-p-and-q}, and~\ref{item:svt-inputs-all-step-neighbor-q-equal-p-plus-2}, the sequences $S_\ell$ and $S_\ell'$ are all-step $1$-neighbors. Hence, by Lemma~\ref{lem:svt-privacy}, 
    \begin{align*}
        \Pr\big[\svt(S_\ell)=t_\ell-t_{\ell-1}\wedge E_\ell \big]
        &\leq \Pr\big[\svt(S_\ell)=t_\ell-t_{\ell-1}\big]\\
        &\leq e^{\eps}\Pr\big[\svt(S_\ell')=t_\ell'-t_{\ell-1}'\big].
    \end{align*}
    Combining these three inequalities implies Inequality~\eqref{eq:partition-final}.

    \smallskip\noindent\textit{Case 2:} Assume $\sigma'$ is an insertion neighbor of $\sigma$. Then the proof is identical  to the previous case with Observation~\ref{obs:svt-inputs-neighboring}~\ref{item:svt-inputs-shift-neighbor-prime} being used instead of Observation~\ref{obs:svt-inputs-neighboring}~\ref{item:svt-inputs-shift-neighbor} and the arguments for $p+1$ and $q-1$ being swapped.
\end{proof}
\section{A key property for SVT}\label{sec:svt-consecutive-outputs}
In this section, we will prove the following lemma. We recall that, for an input stream $\sigma$ of length $T\in\N$, $\svt(\sigma)\in [T+1]$ denotes the time step at which the mechanism outputs $\top$, with $\svt(\sigma)=T+1$ if the mechanism never outputs $\top$.

\begin{lemma*}[Restatement of Lemma~\ref{lem:consecutive-output-distributions}]
    Let $\svt$ be the mechanism described in Algorithm~\ref{alg:svt} with privacy parameter $\eps>0$ and threshold parameter $\tau>1$. Recall the random variable $Z_0\sim\Lap(2/\eps)$ from that algorithm, and let $Y\sim\Lap(4/\eps)$ be an independent random variable. Define the event
    $$E^*:=\{Z_0\ge -(\tau-1)/2\}\wedge\{Y\le (\tau-1)/2\}.$$
    For $t\in\N$, let $\sigma=(x_1,\dots,x_{t+1})$ and $\sigma'=(x_1',\dots, x_t')$ be two real-valued streams such that $\sigma$ is a $1$-shift $1$-neighbor of $\sigma'$ (see Definition~\ref{def:one-shift-delta-neighbor}). Then,
    \begin{enumerate}[label=(\roman*)]
        \item
        $\Pr[\svt(\sigma)=t+1]\leq e^{3\eps/2}\Pr[\svt(\sigma')=t]$, and
        \item
        $\Pr[\svt(\sigma')=t \wedge E^*]\leq e^{7\eps/6}\Pr[\svt(\sigma)=t+1]$.
    \end{enumerate}
\end{lemma*}

\begin{proof}
    Recall the noise random variables $Z_0\sim\Lap(2/\eps)$ and $Z_1, Z_2, \dots\sim\Lap(4/\eps)$ used by $\svt$ in Algorithm~\ref{alg:svt}. Consider two executions of $\svt$ on inputs $\sigma$ and $\sigma'$. Define
    $$W :=\bigl\{(z_0,\dots,z_{t+1})\in\R^{t+2} : \svt(\sigma)=t+1 \text{ when } Z_i=z_i \text{ for all } i\in \{0, \dots, t+1\} \bigr\}$$
    and
    $$W' :=\bigl\{(z_0,\dots,z_{t})\in\R^{t+1} : \svt(\sigma')=t \text{ when } Z_i=z_i \text{ for all } i\in \{0, \dots, t\} \bigr\}$$
    
    \medskip\noindent\textbf{\ref{item:t+1-leq-t}:} {
    By definition of $W$ and independence of $Z_0,\dots,Z_{t+1}$,
    \[
    \Pr[\svt(\sigma)=t+1] = \sum_{(z_0,\dots,z_{t+1})\in W} \prod_{i=0}^{t+1}\Pr[Z_i=z_i],
    \]
    which can be reformulated as
    \begin{equation}\label{eq:condition-z1}
    \begin{aligned}
        \Pr[\svt(\sigma)=t+1]
        &= \sum_{z_1'\in\R}\Pr[Z_1=z_1']
        \SPACESUM
        \sum_{\substack{(z_0,z_2,\dots,z_{t+1}):\\(z_0,z_1',z_2,\dots,z_{t+1})\in W}}
        \SPACESUM
        \Pr[Z_0=z_0]\cdot \prod_{i=2}^{t+1}\Pr[Z_i=z_i],
    \end{aligned}
    \end{equation}
    Since $Z_0\sim\Lap(2/\eps)$ and $Z_{t+1}\sim\Lap(4/\eps)$, the Laplace distribution satisfies
    \[
    \Pr[Z_0=z_0]\le e^{\frac{5}{3}\cdot \frac{\eps}{2}}\Pr[Z_0=z_0+\frac{5}{3}],
    \]
    and
    \[
    \Pr[Z_{t+1}=z_{t+1}]\le e^{\frac{8}{3}\cdot \frac{\eps}{4}}\Pr[Z_{t+1}=z_{t+1}+\frac{8}{3}].
    \]
    Applying these bounds gives
    \begin{align*}
        &\Pr[\svt(\sigma)=t+1] \leq e^{3\eps/2} \sum_{z_1'\in\R}\Pr[Z_1=z_1']\\
        &\quad\cdot\SPACESUM
        \sum_{\substack{(z_0,z_2,\dots,z_{t+1}):\\(z_0,z_1',z_2,\dots,z_{t+1})\in W}}
        \SPACESUM \Pr[Z_0=z_0+\frac{5}{3}]\cdot \Pr[Z_{t+1}=z_{t+1}+\frac{8}{3}]\cdot \prod_{i=2}^{t}\Pr[Z_i=z_i].
    \end{align*}
    Fix $z_1'\in\R$. For every $i\in\{2,\dots,t+1\}$, the random variables $Z_i$ and $Z_{i-1}$ are identically distributed. Thus,
    \begin{equation}\label{eq:shift-random-variables}
        \begin{aligned}
        &
        \sum_{\substack{(z_0,z_2,\dots,z_{t+1}):\\(z_0,z_1',z_2,\dots,z_{t+1})\in W}}
        \SPACESUM
        \Pr[Z_0=z_0+\frac{5}{3}]\cdot \Pr[Z_{t+1}=z_{t+1}+\frac{8}{3}]\cdot \prod_{i=2}^{t}\Pr[Z_i=z_i]\\
        &=
        \quad\SPACESUM
        \sum_{\substack{(z_0,z_2,\dots,z_{t+1}):\\(z_0,z_1',z_2,\dots,z_{t+1})\in W}}
        \SPACESUM
        \Pr[Z_0=z_0+\frac{5}{3}]\cdot \Pr[Z_{t}=z_{t+1}+\frac{8}{3}]\cdot \prod_{i=2}^{t}\Pr[Z_{i-1}=z_i]\\
        &= 
        \quad\SPACESUM
        \sum_{\substack{(z_0,z_2,\dots,z_{t+1}):\\(z_0,z_1',z_2,\dots,z_{t+1})\in W}}
        \SPACESUM
        \Pr[(Z_0,\dots,Z_t)=(z_0+\frac{5}{3},z_2,\dots,z_t,z_{t+1}+\frac{8}{3})].
    \end{aligned}
    \end{equation}
    We will show that
    \begin{equation}\label{eq:top-t+1-to-t-mapping}
    (z_0,z_1',z_2,\dots,z_{t+1})\!\in\! W
    \Longrightarrow
    (z_0+\frac{5}{3},\,z_2,\dots,z_t,\,z_{t+1}\!+\!\frac{8}{3})\!\in\! W' .
    \end{equation}
    Assuming~\eqref{eq:top-t+1-to-t-mapping} holds and using the fact that, for fixed $z_1'$, the mapping
    \[
    (z_0,z_1',z_2,\dots,z_{t+1})
    \;\mapsto\;
    (z_0+\frac{5}{3},z_2,\dots,z_t,z_{t+1}+\frac{8}{3})
    \]
    is injective, we conclude that
    \begin{align*}
    &\SPACESUM\sum_{\substack{(z_0,z_2,\dots,z_{t+1}):\\(z_0,z_1',z_2,\dots,z_{t+1})\in W}}
    \SPACESUM\Pr[(Z_0,\dots,Z_t)=(z_0+\frac{5}{3},z_2,\dots,z_t,z_{t+1}+\frac{8}{3})]\\
    &\le
    \sum_{(z_0,\dots,z_t)\in W'}
    \Pr[(Z_0,\dots,Z_t)=(z_0,\dots,z_t)]
    \\
    &=
    \Pr[\svt(\sigma')=t].
    \end{align*}
    Combining this bound with~\eqref{eq:condition-z1} and~\eqref{eq:shift-random-variables} yields
    \begin{align*}
        \Pr[\svt(\sigma)=t+1]
        &\le
        e^{3\eps/2} \sum_{z_1'\in\R}\Pr[Z_1=z_1']\cdot \Pr[\svt(\sigma')=t] = e^{3\eps/2}  \Pr[\svt(\sigma')=t].
    \end{align*}

    It remains to prove that~\eqref{eq:top-t+1-to-t-mapping} holds: Let $(z_0,z_1',z_2,\dots,z_{t+1})\in W$. By definition, (a) $x_i+z_i\le\tau+8\ln\left(i+\frac{12}{\eps}\right)/\eps+z_0$ for all $i\le t$ and (b) $x_{t+1}+z_{t+1}>\tau+8\ln\left(t+1+\frac{12}{\eps}\right)/\eps+z_0$. Using $|x_i-x_{i-1}'|\le 1$, we obtain $x_{i-1}'+z_i\le\tau+8\ln\left(i+\frac{12}{\eps}\right)/\eps+z_0+1$ for all $i\ge2$, and $x_t'+1 +z_{t+1}>\tau+8\ln\left(t+1+\frac{12}{\eps}\right)/\eps+z_0$. It is known that $\ln(y+1)-\ln(y)\leq \frac{1}{y}$ for every $y>0$. Thus, for $i\geq 2$, we have 
    $$\frac{8}{\eps}\left(\ln\left(i+\frac{12}{\eps}\right) - \ln\left(i-1+\frac{12}{\eps}\right)\right) \leq \frac{8}{\eps}\cdot \frac{1}{i-1 + \frac{12}{\eps}}\leq \frac{8}{\eps}\cdot \frac{\eps}{12} = \frac{2}{3}.$$
    Therefore, (a') $x_{i-1}'+z_i\le\tau+8\left(i-1+\frac{12}{\eps}\right)/\eps+z_0+\frac{5}{3}$. Moreover, since $\ln\left(t+1+\frac{12}{\eps}\right)>\ln\left(t+\frac{12}{\eps}\right)$, we have $x_t'+1 +z_{t+1}>\tau+8\ln\left(t+\frac{12}{\eps}\right)/\eps+z_0$. Consequently, (b') $x_t'+ (z_{t+1}+\frac{8}{3})>\tau+8\ln\left(t+\frac{12}{\eps}\right)/\eps+ (z_0+\frac{5}{3})$. Combining (a') and (b') shows that $(z_0+\frac{5}{3},z_2,\dots,z_t,z_{t+1}+\frac{8}{3})\in W'$, finishing the proof of part~\ref{item:t+1-leq-t}.
    }

    \medskip\noindent\textbf{\ref{item:t-leq-t+1}:} {
    Recall event $E^*= Z_0\ge -(\tau-1)/2 \wedge Y\le (\tau-1)/2$. We write
    \begin{align*}
    \Pr[\svt(\sigma')=t \wedge E^*]
    &=\!\!\!\sum_{\substack{y\le (\tau-1)/2\\z_0\ge -(\tau-1)/2\\(z_0,z_1,\dots,z_t)\in W'}}
    \Pr[Y=y]\cdot \prod_{i=0}^{t}\Pr[Z_i=z_i]\\
    &\le
    e^{7\eps/6}\SPACESUM
    \sum_{\substack{y\le (\tau-1)/2\\z_0\ge -(\tau-1)/2\\(z_0,z_1,\dots,z_t)\in W'}}\SPACESUM
    \Pr[Y\!=\!y]\!\cdot\!\Pr[Z_0\!=\!z_0\!+\!1]\!\cdot\! \Pr[Z_t\!=\!z_t\!+\!\frac{8}{3}]\!\cdot\!
    \prod_{i=1}^{t-1}\Pr[Z_i\!=\!z_i],
    \end{align*}
    where the inequality uses the facts that $\Pr[Z_0=z_0]\le e^{\eps/2}\Pr[Z_0=z_0+1]$ and
    $\Pr[Z_{t}=z_{t}]\le e^{\frac{8}{3}\frac{\eps}{4}}\Pr[Z_{t}=z_{t}+\frac{8}{3}]$.

    For every $i\in \{1, \dots, t\}$, the random variables $Z_i$ and $Z_{i+1}$ are identically distributed. By definition, $Z_1$ and $Y$ also have the same distribution. Therefore,
    \begin{align*}
    &\Pr[\svt(\sigma')=t \wedge E^*]
    \\&\le
    e^{7\eps/6}\SPACESUM
    \sum_{\substack{y\le (\tau-1)/2\\z_0\ge -(\tau-1)/2\\(z_0,z_1,\dots,z_t)\in W'}}\SPACESUM
    \Pr[(Z_0,\dots,Z_{t+1})=(z_0+1,y,z_1,\dots,z_{t-1},z_t+\frac{8}{3})].
    \end{align*}
    We will show that
    \begin{equation}
    \label{eq:top-t-to-t+1-mapping}
    \begin{aligned}
        y\le (\tau-1)/2,\; &z_0\ge -(\tau-1)/2,\; (z_0,z_1,\dots,z_t)\in W'\\
        &\Longrightarrow
    (z_0+1,y,z_1,\dots,z_{t-1},\,z_t+\frac{8}{3})\in W .
    \end{aligned}
    \end{equation}
    Assuming~\eqref{eq:top-t-to-t+1-mapping} holds and using the fact that  the mapping
    \[
    (y,z_0,z_1,\dots,z_t)
    \;\mapsto\;
    (z_0+1,y,z_1,\dots,z_{t-1},z_t+\frac{8}{3})
    \]
    is injective, we obtain
    \begin{align*}
    \Pr[\svt(\sigma')=t \wedge E^*]
    &\leq
    e^{7\eps/6}
    \sum_{(z_0,\dots,z_{t+1})\in W}
    \Pr[(Z_0,\dots,Z_{t+1})=(z_0,\dots,z_{t+1})]
    \\
    &=
    e^{7\eps/6}\Pr[\svt(\sigma)=t+1].
    \end{align*}
    It remains to prove~\eqref{eq:top-t-to-t+1-mapping} holds. Let $z_0\ge -(\tau-1)/2$, $y\le (\tau-1)/2$, and $(z_0,z_1,\dots,z_t)\in W'$. Since $\sigma$ is a $1$-shift $1$-neighbor of $\sigma'$, we have $x_1\leq 1$. Thus, 
    $$x_1+y\le x_1+(\tau-1)/2\leq \tau - (\tau-1)/2< \tau + z_0 + 8\ln\left(1+\frac{12}{\eps}\right)/\eps+1.$$
    Moreover,
    as $ (z_0,z_1,\dots,z_t)\in W'$,
    we have (a) $x_i'+z_i\le \tau+8\ln\left(i+\frac{12}{\eps}\right)/\eps+z_0$ for all $i\le t-1$ and (b) $x_t'+z_t>\tau+8\ln\left(t+\frac{12}{\eps}\right)/\eps+z_0$. Using $|x_{i+1}-x_i'|\le 1$, we have $x_{i+1}+z_i\le \tau+8\ln\left(i+\frac{12}{\eps}\right)/\eps+z_0+1$ for all $i\le t-1$ and $x_{t+1}+z_t>\tau+8\ln\left(t+\frac{12}{\eps}\right)/\eps+z_0-1$. We know that $\ln\left(i+\frac{12}{\eps}\right)<\ln\left(i+1+\frac{12}{\eps}\right)$ and, as discussed before, 
    $$\frac{8}{\eps}\left(\ln\left(t+1+\frac{12}{\eps}\right) - \ln\left(t+\frac{12}{\eps}\right)\right)\leq \frac{2}{3}.$$
    Therefore, (a') $x_{i+1}+z_i\le \tau+8\ln\left(i+1+\frac{12}{\eps}\right)/\eps+(z_0+1)$ for all $i\le t-1$ and (b') $x_{t+1}+z_t>\tau+8\ln\left(t+1+\frac{12}{\eps}\right)/\eps+z_0-\frac{5}{3}$. Equivalently, $x_{t+1}+(z_t+\frac{8}{3})>\tau+8\ln\left(t+1+\frac{12}{\eps}\right)/\eps+(z_0+1)$. Hence, if $(Z_0,\dots,Z_{t+1})=(z_0+1,y,z_1,\dots,z_{t-1},z_t+\frac{8}{3})$, then $\svt(\sigma)=t+1$. Thus, $(z_0+1,y,z_1,\dots,z_{t-1},z_t+\frac{8}{3})\in W$, completing the proof of part~\ref{item:t-leq-t+1}.
    }
\end{proof}

\section{The \texorpdfstring{$\simecc$}{SimECC} mechanism}\label{sec:simecc}

In this section, we introduce a simple partitioning subroutine, denoted by $\simpart$, which---unlike the partitioning subroutine $\partt$ from Section~\ref{sec:ecc}---does not access the private inputs of $\ecc$. Similar to $\partt$, the subroutine $\simpart$ outputs either $\top$ or $\bot$ upon each update call and can be used by $\ecc$ to indicate checkpoints. We then explain how to construct the simpler continual counter mechanism $\simecc$ described in \Cref{thm:informal-ub} by replacing $\partt$ in $\ecc$ with $\simpart$.

The mechanism $\simpart$, described in Algorithm~\ref{alg:simple-partition}, is initialized with privacy parameter $\eps > 0$ and step-size parameter $\mu \in \N$. It first draws a sample $Z_1 \sim \DLap(2/\eps)$ and sets the first (future) checkpoint to $t_1 = \max\{1,\, \mu + Z_1\}$.
It also initializes a counter $\ell = 1$, tracking the number of samples drawn so far, and a time variable $t = 0$. We note that since $\ln(1)=0$,
\[
t_1 = \max\{1,\, \mu + \lceil 4\ln(\ell)/\eps\rceil + Z_\ell\}
\]
at initialization.

Upon each update call, $\simpart$ proceeds as follows. It increments the time variable $t \leftarrow t+1$, and checks whether $t$ has reached the next checkpoint $t_\ell$. If $t < t_\ell$, it outputs $\bot$. If $t = t_\ell$, it increments $\ell \leftarrow \ell + 1$, draws a new sample $Z_\ell \sim \DLap(2/\eps)$, sets the next checkpoint to
\[
t_\ell = t + \max\{1,\, \mu + \lceil 4\ln(\ell)/\eps\rceil + Z_\ell\},
\]
and outputs $\top$. Note that, by construction, the current time $t$ is always at most the next checkpoint $t_\ell$.

\begin{notation}
    For $T \in \N$, we denote by $\simpart(T)$ the sequence of time steps at which $\simpart$ outputs $\top$ during its execution for $T$ steps.
\end{notation}

\begin{algorithm}
    \begin{algorithmic}
        \Function{Initialize}{Privacy parameter $\eps>0$, mean step size $\mu\in \N$}
            \State $t\leftarrow 0$
            \State $\ell\leftarrow 1$
            \State $Z_1\sim \DLap(2/\eps)$
            \State $t_\ell \leftarrow \max\{1, \mu + Z_1\}$
        \EndFunction
        \Function{Update}{}
            \State $t \leftarrow t+1$
            \If{$t=t_\ell$}
                \State $\ell\leftarrow \ell+1$
                \State $Z_\ell\sim \DLap(2/\eps)$
                \State $t_\ell \leftarrow t + \max\{1,\, \mu + \lceil 4\ln(\ell)/\eps\rceil + Z_\ell\}$
                \State Output $\top$
            \Else\Comment{$t<t_\ell$}
                \State Output $\bot$
            \EndIf
        \EndFunction
    \end{algorithmic}
    \caption{$\simpart$}\label{alg:simple-partition}
\end{algorithm}

We next show the following lemma for $\simpart$, which is analogous to Lemma~\ref{lem:partition} for the partitioning subroutine $\partt$.

\begin{lemma}\label{lem:simple-partition}
    Let $\simpart$ be the mechanism described in Algorithm~\ref{alg:simple-partition} with privacy parameter $\eps > 0$ and step-size parameter $\mu > 3$. Then there exists an event $E$ such that
    \[
    \Pr[E] \ge 1 - \frac{\pi^2}{12} \cdot e^{-\eps(\mu - 3)/2},
    \]
    and for every $T \in \N$, the following statements hold. Let $\calC_T$ and $\calC_T^+$ be as in Definition~\ref{def:f}, and define
    \[
    \calC_T^{++} = \left\{ (t_1, \dots, t_k) \in \calC_T^+ \mid t_\ell - t_{\ell-1} \ge 3 \text{ for all } \ell \in [k] \right\}.
    \]
    Then:
    \begin{enumerate}[label=(\roman*)]
        \item \label{item:simple-partition-zero-prob} For every checkpoint sequence $t_{1:k}\in\calC_T\setminus\calC_T^{++}$,
        \[
        \Pr\left[\simpart(T)=t_{1:k}\wedge E\right] = 0.
        \]
        \item \label{item:simple-partition-dp} For every pair of edit-neighboring sequences $\sigma$ and $\sigma'$ in $([0, 1]\cup\{\bot\})^T$ and every $t_{1:k}\in\calC_T^{++}$,
        \[
        \Pr\left[\simpart(T)=t_{1:k}\right]
        \leq
        e^{2\cdot\eps}
        \Pr\left[\simpart(T) = f_{\sigma\to\sigma'}(t_{1:k})\right],
        \]
        where $f_{\sigma\to\sigma'}:\calC_T^+\to \calC_T$ is defined in Definition~\ref{def:f}.
    \end{enumerate}
\end{lemma}

\begin{proof}
    Consider an execution of $\simpart$ with privacy parameter $\eps$ and step-size $\mu$. For each $\ell \in \N$, let $Z_\ell$ denote the $\ell$-th (potential) sample drawn by this mechanism from $\DLap(2/\eps)$. Define the event
    \[
        E := \bigwedge_{\ell=1}^{\infty} \left( \mu + Z_\ell + \lceil 4\ln(\ell)/\eps\rceil \ge 3 \right).
    \]
    By Corollary~\ref{cor:discrete-laplace-privacy-accuracy} and a union bound, we have
    \begin{align*}
        \Pr[E]
        &\ge 1 - \sum_{\ell=1}^{\infty} \Pr\!\left[ Z_\ell < -\mu - \lceil 4\ln(\ell)/\eps\rceil + 3 \right] \\
        &\ge 1 - \sum_{\ell=1}^{\infty} \frac{1}{2} \exp\!\left(-\frac{\eps}{2}\left(\mu - 3 + \frac{4\ln(\ell)}{\eps}\right)\right) \\
        &= 1 - \frac{1}{2} e^{-\frac{\eps}{2}(\mu - 3)} \sum_{\ell=1}^{\infty} \frac{1}{\ell^2}
        = 1 - \frac{\pi^2}{12} \cdot e^{-\frac{\eps}{2}(\mu - 3)},
    \end{align*}
    where we used the fact that $\sum_{\ell=1}^\infty \frac{1}{\ell^2}= \frac{\pi^2}{6}$.
    
    \medskip
    Fix $T\in\N$ and consider an execution of $\simpart$ for $T$ steps. By construction, if the event $E$ holds, then the gap between any two consecutive checkpoints is at least $3$. Therefore, for every $t_{1:k} \in \calC_T \setminus \calC_T^{++}$,
    \[
    \Pr\!\left[\simpart(T) = t_{1:k} \wedge E\right] = 0,
    \]
    which proves part~\ref{item:simple-partition-zero-prob}.

    \medskip
    Let $\sigma$ and $\sigma'$ be two edit-neighboring sequences in $([0,1] \cup \{\bot\})^T$, meaning that one of these sequences is an insertion neighbor of the other one. We prove part~\ref{item:simple-partition-dp} for the case where $\sigma$ is an insertion neighbor of $\sigma'$; the other case follows by symmetry.

    Assume $\sigma$ is an insertion neighbor of $\sigma'$ at step $i \in [T]$. Fix a checkpoint sequence $t_{1:k} \in \calC_T^{++}$, and let $p,q \in [k+1]$ be the indices used in the definition of $f_{\sigma \to \sigma'}(t_{1:k})$. If $p > q-2$, then $f_{\sigma \to \sigma'}(t_{1:k}) = t_{1:k}$; hence,
    \[
    \Pr[\simpart(T) = t_{1:k}]
    =
    \Pr[\simpart(T) = f_{\sigma \to \sigma'}(t_{1:k})].
    \]
    Assume now that $p \le q-2$. In this case, $f_{\sigma\to\sigma'}(t_{1:k})=(t_1,\dots,t_p, t_{p+1}-1,\dots,t_{q-2}-1, t_{q-1},\dots,t_k)$. By the design of $\simpart$, $\simpart(T) = t_{1:k}$ if and only if the following conditions hold:
    \begin{itemize}
        \item[(a)] The mechanism $\simpart$ draws exactly $k+1$ samples $Z_1, \dots, Z_{k+1}$ during its $T$ steps of execution.
        \item[(b)] Let $t_0 := 0$. For each $\ell \in [k]$,
        \[
        t_\ell = t_{\ell-1} + \max\!\left\{1,\, \mu + \lceil 4\ln(\ell)/\eps\rceil +Z_\ell \right\}.
        \]
        \item[(c)] The last checkpoint scheduled for the future exceeds $T$, i.e.,
        \[
        t_k + \max\!\left\{1,\, \mu + \lceil4\ln(k+1)/\eps\rceil + Z_{k+1} \right\} > T.
        \]
    \end{itemize}
    Therefore,
    \begin{align*}
        \Pr[\simpart(T)=t_{1:k}]
        &= \Pr\!\left[t_k + \max\{1,\, \mu + \lceil 4\ln(k+1)/\eps\rceil + Z_{k+1}\} > T \right]\\
        &\qquad\cdot \prod_{\ell=1}^k \Pr\!\left[t_{\ell-1} + \max\{1,\, \mu + \lceil4\ln(\ell)/\eps\rceil + Z_\ell\} = t_\ell \right].
    \end{align*}
    Since $t_{1:k}\in\calC_T^{++}$, for every $\ell\in[k]$, $t_\ell-t_{\ell-1}\geq 3>1$. Thus, $\max\{1, \mu + \lceil4\ln(\ell)/\eps\rceil +Z_\ell\}=t_\ell-t_{\ell-1}$ if and only if $\mu + \lceil4\ln(\ell)/\eps\rceil + Z_\ell=t_\ell-t_{\ell-1}$. Therefore,
    \begin{align*}
        \Pr[\simpart(T)=t_{1:k}]
        &= \Pr\!\left[t_k + \max\{1,\, \mu + \lceil4\ln(k+1)/\eps\rceil + Z_{k+1}\} > T \right]
        \cdot \prod_{\ell=1}^k \Pr\!\left[t_{\ell-1} + \mu + \lceil4\ln(\ell)/\eps\rceil + Z_\ell = t_\ell \right]\\
        &= \Pr\!\left[t_k + \max\{1,\, \mu + \lceil4\ln(k+1)/\eps\rceil + Z_{k+1}\} > T \right]
        \cdot \prod_{\ell=1}^k \Pr\!\left[Z_\ell = t_\ell - t_{\ell-1} - \mu - \lceil4\ln(\ell)/\eps\rceil\right]
    \end{align*}
    By Corollary~\ref{cor:discrete-laplace-privacy-accuracy}, for every $\ell\in [k]$, we have
    \begin{align*}
        \Pr\!\left[Z_\ell = t_\ell - t_{\ell-1} - \mu - \lceil4\ln(\ell)/\eps\rceil \right]
        \leq e^{\eps/2}\Pr\!\left[Z_\ell = t_\ell - t_{\ell-1} - \mu - \lceil 4\ln(\ell)/\eps\rceil +1\right],
    \end{align*}
    and
    \begin{align*}
        \Pr\!\left[Z_\ell = t_\ell - t_{\ell-1} - \mu - \lceil4\ln(\ell)/\eps\rceil\right]
        &\leq e^{\eps/2}\Pr\!\left[Z_\ell = t_\ell - t_{\ell-1} - \mu - \lceil4\ln(\ell)/\eps\rceil -1\right].
    \end{align*}
    Hence,
    \begin{align*}
        \Pr[\simpart(T)=t_{1:k}]
        &= \Pr\!\left[t_k + \max\{1,\, \mu + \lceil4\ln(k+1)/\eps\rceil + Z_{k+1}\} > T \right]\\
        &\qquad\times \prod_{\ell=1}^k \Pr\!\left[Z_\ell = t_\ell - t_{\ell-1} - \mu - \lceil4\ln(\ell)/\eps\rceil\right]\\
       &\le e^\eps \, \Pr\!\left[t_k + \max\!\left\{1,\, \mu + \lceil4\ln(k+1)/\eps\rceil + Z_{k+1}\right\} > T \right] \\
        &\qquad \times \Pr\!\left[Z_{p+1} = t_{p+1} - t_p - \mu - \lceil4\ln(p+1)/\eps\rceil - 1\right] \\
        &\qquad \times \Pr\!\left[Z_{q-1} = t_{q-1} - t_{q-2} - \mu - \lceil4\ln(q-1)/\eps\rceil + 1\right] \\
        &\qquad \times \prod_{\ell \in [k]\setminus \{p+1,q-1\}}
        \Pr\!\left[Z_\ell = t_\ell - t_{\ell-1} - \mu - \lceil4\ln(\ell)/\eps\rceil \right].
    \end{align*}
    Since $t_{1:k}\in \calC_T^{++}$, we have (a) $t_\ell-t_{\ell-1}\geq 3$ for every $\ell\in \{1, \dots, p\}\cup\{q, \dots, k\}$; (b) $(t_{p+1} - 1) -t_p\geq 2$; (c) $t_{q-1} -(t_{q-2}-1)\geq 4$; and (d) $(t_\ell-1)-(t_{\ell-1}-1)= t_\ell-t_{\ell-1}\geq 3$ for every $\ell\{p+2, \dots, q-2\}$. Since all of these differences are strictly larger than $1$, the same as above, we can drop the maximum with $1$ in the following equality:
    \begin{align*}
        &\Pr[\simpart(T)= f_{\sigma\to\sigma'}(t_{1:k})]\\
        &=\Pr[\simpart(T)=(t_1,\dots,t_p, t_{p+1}-1,\dots,t_{q-2}-1, t_{q-1},\dots,t_k)]\\
        &= \Pr\!\left[t_k + \max\{1,\, \mu + \lceil4\ln(k+1)/\eps\rceil + Z_{k+1}\} > T \right]
        \\&\qquad \times  \Pr\!\left[Z_{p+1} = (t_{p+1} - 1) - t_p - \mu - \lceil4\ln(p+1)/\eps\rceil\right]
        \\&\qquad \times \Pr\!\left[Z_{q-1} = t_{q-1} - (t_{q-2} - 1) - \mu - \lceil4\ln(q-1)/\eps\rceil\right]
        \\&\qquad \times \prod_{\ell\in [k]\setminus \{p+1,q-1\}}\Pr\!\left[Z_\ell = t_\ell - t_{\ell-1} - \mu - \lceil4\ln(\ell)/\eps\rceil\right]
    \end{align*}
    \sloppy
    Note that since $q \le k+1$ and the mapping $f_{\sigma \to \sigma'}$ leaves all checkpoints from index $q-1$ onward unchanged, the final checkpoint in both sequences $t_{1:k}$ and $f_{\sigma \to \sigma'}(t_{1:k})$ is $t_k$. Consequently, the term $\Pr\!\left[t_k + \max\{1,\, \mu + \lceil4\ln(k+1)/\eps\rceil + Z_{k+1}\} > T \right]$ appears identically in the equality relation for $\Pr[\simpart(T)= f_{\sigma\to\sigma'}(t_{1:k})]$ and the inequality relation for $\Pr[\simpart(T)= t_{1:k}]$. Combining these two implies part~\ref{item:simple-partition-dp}.
\end{proof}

\medskip\noindent\textbf{$\simecc$.} 
The mechanism $\simecc$ as described before is simply an instantiation of $\ecc$ (Algorithm~\ref{alg:ecc}) with the $\partt$ subroutine replaced by a simplified checkpoint generating mechanism $\simpart$. 
At initialization, instead of setting the threshold parameter $\tau=1+\frac{408\ln(\frac{\pi^2}{2 \delta})}{\eps}$, the mechanism $\ecc$ sets a step-size parameter $\mu=3+ \lceil102\ln(\frac{2\pi^2}{3\delta})/\eps\rceil$, and instead of instantiating $\partt$ with parameters $\eps/51$ and $\tau$, it initializes $\simpart$ with parameters $\eps/51$ and $\mu$. The update procedure remains unchanged, except that each call to $\partt(x_t)$ is replaced by a call to $\simpart()$, which does not access the input.

\medskip\noindent\textbf{Privacy analysis for $\simecc$.}
The privacy analysis of the modified $\ecc$ is identical to that in Section~\ref{sec:ecc-privacy}, with Lemma~\ref{lem:partition} replaced by Lemma~\ref{lem:simple-partition}. Although the two lemmas slightly differ in their statements, Lemma~\ref{lem:simple-partition} provides guarantees that are at least as strong as those provided by Lemma~\ref{lem:partition} and used in the original privacy analysis of $\ecc$. Thus, that analysis holds for the modified $\ecc$ as well. We highlight these differences below:
\begin{itemize}
    \item \emph{Probability of the good event.}
    Lemma~\ref{lem:partition} guarantees an event $E$ with
    \[
    \Pr[E] \ge 1 - \frac{\pi^2}{4} e^{-\eps(\tau-1)/8},
    \]
    whereas Lemma~\ref{lem:simple-partition} ensures
    \[
    \Pr[E] \ge 1 - \frac{\pi^2}{12} e^{-\eps(\mu - 3)/2}.
    \]
    For the chosen parameters $\tau=1+\frac{408\ln(\frac{\pi^2}{2 \delta})}{\eps}$ and $\mu = 3+\lceil102\ln(\frac{2\pi^2}{3 \delta})/\eps\rceil$, the latter bound is at least as strong as the former one (i.e., yields a smaller failure probability). Hence, the guarantee on $\Pr[E]$ in Lemma~\ref{lem:simple-partition} implies the one in Lemma~\ref{lem:partition} and suffices for the analysis.
    \item \emph{Stability under neighboring inputs.}
    Lemma~\ref{lem:partition} bounds
    \[
    \Pr[\partt(\sigma)=t_{1:k} \wedge E]
    \le
    e^{17\eps/3}\,
    \Pr\!\left[\partt(\sigma') = f_{\sigma \to \sigma'}(t_{1:k})\right],
    \]
    whereas Lemma~\ref{lem:simple-partition} gives the stronger bound
    \[
    \Pr[\simpart(T)=t_{1:k}]
    \le
    e^{2\eps}\,
    \Pr\!\left[\simpart(T)=f_{\sigma \to \sigma'}(t_{1:k})\right].
    \]
    Since $\Pr[\simpart(T)=t_{1:k} \wedge E] \le \Pr[\simpart(T)=t_{1:k}]$ and $2\eps \le 17\eps/3$, the guarantee of Lemma~\ref{lem:simple-partition} directly implies the bound provided by Lemma~\ref{lem:partition} and required in the original analysis.
\end{itemize}
Putting everything together, we get the following privacy guarantee.
\begin{theorem}\label{thm:simecc-privacy}
    Let $\eps>0$ and $0<\delta\le 1$. Suppose there exists a continual counting mechanism $\cc$ that is $(4\eps/27, e^{-19\eps/27}\delta/16)$-DP with respect to the $1$-step $1$-neighbor relation, defined in Definition~\ref{def:k-step-delta-neighboring}. Let $\bcc$ denote the biased version of $\cc$ as constructed in Lemma~\ref{lem:cc-biased}. Then, using $\cc$ and $\bcc$ as black boxes, the continual mechanism $\simecc$, described in the text above, is $(\eps,\delta)$-DP with respect to the edit neighbor relation, defined in Definition~\ref{def:queue-neighboring}.
\end{theorem}

\medskip\noindent\textbf{Accuracy analysis for \texorpdfstring{$\simecc$}{SimECC}.} The accuracy analysis of the $\ecc$ when executed with $\simpart$ is significantly simpler than when it is executed with \partt.
In \Cref{lem:simpart-acc} we establish that true checkpoints as defined by $\simpart$ cannot be too far apart.
Then, in \Cref{thm:ecc-simpart-accuracy}, we show that it follows from the bound between the gaps between true checkpoints, and the error bound of $\bcc$, that the next noisy checkpoint following any given time-step $t$ cannot occur too late.
In particular, this implies that the most recent update prior to $t$ could not have been too far in the past either, and, by the $1$-sensitivity of the running sum, the error at time-step $t$ is also bounded.
The caveat here is that the error does not adapt to the sparsity of the stream, as it does when we execute $\ecc$ with $\simpart$.

\begin{lemma}\label{lem:simpart-acc}
    Assume $\beta \leq 1/2$, and let $\eps_p$ denote the privacy parameter with which $\simpart$ is executed. There exists a $(1-\beta)$-probability event $\calE_\simpart$ defined over the random 
    coins of $\simpart$ such that, conditioned 
    on $\calE_\simpart$, for every $\ell \in \N$,
    \[
    t_\ell - t_{\ell-1} \;\leq\; \mu + O\!\left(\frac{\ln(\ell/\beta)}{\eps_p}\right).
    \]
\end{lemma}

\begin{proof}
    Recall $t_\ell - t_{\ell-1} = \max\{1,\, \mu + \lceil4\ln(\ell)/\eps_p\rceil + Z_\ell\}$ 
    where $Z_\ell \sim \DLap(2/\eps_p)$ i.i.d. Recall $\beta_\ell := 6\beta/(\pi^2\ell^2)$, 
    so $\sum_{\ell=1}^\infty \beta_\ell = \beta$ and $\beta_\ell \leq \beta \leq 1/2$. Define
    \[
    \xi_\ell \;:=\; \left\lceil 1 + \frac{2}{\eps_p}\ln\!\frac{1}{2\beta_\ell}\right\rceil \in \N,
    \]
    where $\xi_\ell \geq 1$ since $2\beta_\ell \leq 1$. By Corollary~\ref{cor:discrete-laplace-privacy-accuracy} 
    and $\xi_\ell - 1 \geq (2/\eps_p)\ln(1/(2\beta_\ell))$,
    \[
    \Pr[Z_\ell \geq \xi_\ell] \;\leq\; \tfrac{1}{2}\exp\!\left(-(\xi_\ell - 1)\eps_p/2\right) 
    \;\leq\; \tfrac{1}{2} \cdot 2\beta_\ell \;=\; \beta_\ell.
    \]
    Let $\calE_\simpart := \bigwedge_{\ell=1}^\infty \{Z_\ell < \xi_\ell\}$; a union bound 
    gives $\Pr[\calE_\simpart] \geq 1 - \sum_\ell \beta_\ell = 1 - \beta$.
    
    Condition on $\calE_\simpart$, so $Z_\ell \leq \xi_\ell - 1$ (as $Z_\ell, \xi_\ell$ are integer-valued).
    Define $B_\ell := \mu + \lceil4\ln(\ell)/\eps_p\rceil + \xi_\ell - 1$. Since $\mu \geq 3$, $\ln\ell \geq 0$, 
    and $\xi_\ell \geq 1$, we have $B_\ell \geq 3 > 1$, so
    \[
    t_\ell - t_{\ell-1} \;=\; \max\{1,\, \mu + \lceil4\ln(\ell)/\eps_p\rceil + Z_\ell\} \;\leq\; B_\ell.
    \]
    By the definition of $\xi_\ell$, $\xi_\ell - 1 \leq 1 + (2/\eps_p)\ln(1/(2\beta_\ell))$, so
    \[
    B_\ell \;\leq\; \mu + 1 + \left\lceil\frac{4\ln\ell}{\eps_p}\right\rceil + \frac{2}{\eps_p} \ln\!\frac{1}{2\beta_\ell}.
    \]
    Expanding $\ln(1/(2\beta_\ell)) = 2\ln\ell + \ln(1/\beta) + \ln(\pi^2/12)$, it follows that
    $B_\ell = \mu + O(\ln(\ell/\beta)/\eps_p)$.
\end{proof}

\begin{theorem}
    \label{thm:ecc-simpart-accuracy}
    Consider the modified $\ecc$ using $\simpart$ in place of $\partt$ with step-size 
    parameter $\mu$, and assume $\beta \leq 1/2$ and $\delta \leq \beta$. 
    Condition on $\calE_\simpart$ as defined in Lemma~\ref{lem:simpart-acc} and on the 
    $1-\beta_C$ and $1-\beta_B$ events under which $\cc$ and $\bcc$ are accurate. 
    Then for all $t \in \N$,
    \[
    |s_t - y_t| \;\leq\; 2 E_\bcc(t) + E_\cc(t) + \mu + O\!\left(\tfrac{1}{\eps}\ln(t/\delta)\right).
    \]
\end{theorem}

\begin{proof}
    If there is no noisy checkpoint at or before $t$, then  $y_t = 0$ and $|s_t - y_t| = s_t \leq t$. Since $t < \widehat t_1$, 
    $t \;<\; t_1 + E_\bcc(1) \;\leq\; \mu + E_\bcc(t) + O\!\left(\tfrac{1}{\eps}\ln(t/\delta)\right),$
    using $\delta \leq \beta$ and monotonicity. The stated bound follows.
    
    If there is a noisy checkpoint at or before $t$, then let $\widehat t$ denote the most recent such noisy checkpoint. 
    Let  $\ell$ be its index, so $y_t = \widehat v_\ell$ and by Lemma~\ref{lem:part-accuracy}(2), 
    monotonicity of the error terms, and $\ell \leq t_\ell \leq t$,
    \begin{equation}\label{eqn:simpart-accuracy.1}
        |s_{\widehat t} - y_t| \;\leq\; E_\bcc(t) + E_\cc(t).
    \end{equation}
    
    Since $\widehat t_{\ell+1} > t$ by maximality of $\widehat t$, combining 
    Lemmas~\ref{lem:part-accuracy}(1) and~\ref{lem:simpart-acc} with monotonicity of 
    $E_\bcc$ and $\ell \leq t$,
    \[
    t - \widehat t \;<\; \widehat t_{\ell+1} - \widehat t_\ell 
    \;\leq\; (t_{\ell+1} + E_\bcc(\ell+1)) - t_\ell 
    \;\leq\; \mu + E_\bcc(t) + O\!\left(\tfrac{1}{\eps_p}\ln(t/\delta)\right),
    \]
    using $\delta \leq \beta$. By $1$-sensitivity, $|s_t - s_{\widehat t}| \leq t - \widehat t$. 
    Triangle inequality with \eqref{eqn:simpart-accuracy.1} yields
    \[
    |s_t - y_t| \;\leq\; 2 E_\bcc(t) + E_\cc(t) + \mu + O\!\left(\tfrac{1}{\eps_p}\ln(t/\delta)\right).
    \]
    Since $\eps_p = \eps/51$, the theorem statement follows. \qedhere
\end{proof}
\section{Lower bounds}

In this section we will derive lower bounds for private continual counting under edit-neighboring, as well as under \emph{prefix-sum neighboring}, streams.
We give the latter definition (implied by the \emph{sensitivity vector set} characterization in \cite{andersson26improved}; we will return to this) next.
\begin{definition}[Prefix-Sum Neighbor Relation]\label{def:prefix-neighboring}
    For $T\in\N$, two sequences $\sigma=(x_1,\dots,x_T)$ and $\sigma'=(x_1',\dots,x_T')$ in $\mathbb{Z}^T$ (or a subset thereof) are said to be \emph{prefix-sum neighboring}, written $\sigma\sim_p\sigma'$, if one of the following conditions hold:
    \begin{itemize}
        \item For every $t\in[T]$, $\sum_{i=1}^t x_i - \sum_{i=1}^t x_i'\in \{0,1\}$.
        \item For every $t\in[T]$, $\sum_{i=1}^t x_i - \sum_{i=1}^t x_i'\in \{0,-1\}$.
    \end{itemize}
\end{definition}
Intuitively, two sequences are prefix-sum neighbors iff one of the prefix sums always either equals the other one, or trails it by 1.
Similarly, we will use the shorthand \enquote{$\sim_e$} for edit-neighboring streams.
Throughout this section we will also use the notation $\ccgeneric{T}$ to denote the problem of continual counting with neighboring relation $\sim$ and input streams with elements $\mathcal{S}$ of length $T$.
We emphasize that all lower bounds we derive are for streams of a finite length $T\in\N$.

This section is organized as follows:
\begin{enumerate}
    \item First we show that $\ccptrits{T}$ requires error polynomial in $T$.
    \item Secondly, we show that $\ccpbits{T}$ requires error polynomial in $T$.
    \item Lastly, we show that any input-independent mechanism for $\ccebits{T}$ requires error polynomial in $T$.
\end{enumerate}

\subsection{A lower bound for \texorpdfstring{$\ccptrits{T}$}{Count}}

The lower bound here is based on a reduction from the $\cdistinct$ problem.
It follows relatively straightforwardly from \cite{jain23distinct} together with \cite{andersson26improved}.
The only technical challenge is reconciling the subtle differences in how the two works define the problem.
Our version of $\cdistinct$ is defined next.
\begin{definition}[$\cdistinct$]\label{def:cdistinct}
    Fix a universe $\mathcal{U}$, let $\mathcal{U}_{\pm} = (\{+,-\}\times\mathcal{U})\cup\{\bot\}$ and $T\in\N$.
    Consider a stream $\sigma = (x_1, \dots, x_T) \in \mathcal{U}_{\pm}^T$.
    Here $(+, u)$ is interpreted as \enquote{add item $u$},  $(-, u)$ as \enquote{delete item $u$}, and $\bot$ as \enquote{do nothing}.
    We define the problem $\cdistinct$ as the continual release problem, where, on receiving $x_t$, we are to immediately output
    \begin{equation*}
       f(\sigma)_t =  \left\lvert \left\{ u\in\mathcal{U} : \sum_{i=1}^t \indic[x_i = (+, u)] > \sum_{i=1}^t \indic[x_i = (-, u)]\right\}\right\rvert\,.
    \end{equation*}
    We say $\sigma$ and $\sigma'$ are \emph{item-level neighboring}, denoted $\sigma \sim \sigma'$, if one can be derived from the other by replacing all updates involving one item $u\in\mathcal{U}$ by $\bot$'s.
\end{definition}
We make a few remarks before proceeding.
Firstly, this definition of item-level neighboring is the one adopted in \cite{andersson26improved}, and distinct from the one in \cite{jain23distinct}, where instead two streams are item-level neighboring if we can derive one from the other by replacing \emph{any subset} of updates pertaining to any one item.
As remarked by \citeauthor{andersson26improved}, any two streams neighboring by the definition in \cite{jain2023price}, are at most 2-neighboring in $\cdistinct$.
Additionally, while \cite{andersson26improved} allowed for an arbitrary number of updates per step, we restrict to at most one update per step, as in \cite{jain23distinct}.
We will use the following lower bound.
\begin{theorem}[{\cite[Theorem 1.7 (worst-case bounds)]{jain23distinct}}]\label{thm:cdistinct-lb}
      Let $\eps\in(0, 1/2], \delta\in[0, 1]$, and $T\in\N$ be sufficiently large.
      Then any $(\alpha, \frac{1}{100})$-accurate item-level $(\eps, \delta)$-DP mechanism for $\cdistinct$ over $T$ steps satisfies:
      \begin{enumerate}
          \item If $\delta = o(\eps / T)$, then $\alpha = \tilde{\Omega}\left(\min\left\{\frac{T^{1/3}}{\eps^{2/3}}, T\right\}\right)$.
          \item If $\delta = 0$, then $\alpha = \Omega\left(\min\left\{\sqrt{\frac{T}{\eps}}, T\right\}\right)$.
      \end{enumerate}
\end{theorem}
\begin{proof}[Proof sketch.]
    The statement is identical to \cite[Theorem 1.7]{jain23distinct} when setting their \emph{maximum flippancy} parameter $w = \Theta(T)$, with the exception that they allow $\eps \in (0, 1]$.
    Since any neighboring inputs in their setting are guaranteed to be 2-neighbors in ours, we have that any $(\eps, \delta)$-DP mechanism $\cM$ for $\cdistinct$ is also a $(2\eps, 2e^{2\eps}\delta)$-DP for their version of the problem via group privacy.
    Hence, invoking their lower bound for $\eps' = 2\eps$ and $\delta' = 6\delta \geq 2e^{2\eps}\delta$ (valid for $\eps \leq 1/2$), extends their lower bound to $(\eps, \delta)$-DP for $\cdistinct$.
    This restricts the range on $\eps$ accordingly.
    For $\delta$, we can keep its range unchanged.
    If $\delta = o(\eps / T)$, then $6\delta = o(2\eps/T)$ is implied, so the condition on $\delta'$ is satisfied for sufficiently large $T$.
    We have now proved the statement for $\eps\in(0, 1/2], \delta\in[0,1]$.
\end{proof}

\paragraph{Counting on the difference stream.} Towards proving the lower bound on $\ccptrits{T}$, by reducing from $\cdistinct$, we need the following concept.
Let $\sigma \sim \sigma'$  be two neighboring inputs to $\cdistinct$, and define the two \emph{(output) difference streams}
\begin{align*}
    d(\sigma) &= (f(\sigma_1), f(\sigma_2) - f(\sigma_1), \dots, f(\sigma_T) - f(\sigma_{T-1}))\,,\\
    d(\sigma') &= (f(\sigma_1'), f(\sigma_2') - f(\sigma_1'), \dots, f(\sigma_T') - f(\sigma_{T-1}'))\,.
\end{align*}
We will prove the following claim, using a result from \cite{andersson26improved}.
\begin{claim}\label{clm:difference-streams}
   $d(\sigma), d(\sigma') \in \{-1, 0, 1\}^T$ and $d(\sigma)\sim_p d(\sigma')$.
\end{claim}
\begin{lemma}[{\cite{andersson26improved}}]\label{lem:cdistinct-s1}
   Let $\sigma, \sigma'$ be two neighboring inputs to $\cdistinct$, and $d(\sigma), d(\sigma')$ be the corresponding output difference streams.
   Then $d(\sigma) - d(\sigma') \in \mathcal{S}_1$ where
   \begin{equation*}
      \mathcal{S}_1 = \left\{ v\in\mathbb{Z}^T : \max_{i, j \in [T]} \left\lvert \sum_{k=i}^{j} v_k \right\rvert \leq 1 \right\}\,.
   \end{equation*}
\end{lemma}
\begin{proof}[Proof of \Cref{clm:difference-streams}]
    The fact that $d(\sigma), d(\sigma') \in \{-1, 0, 1\}^T$ follows from the definition of $\cdistinct$: at each time step, a single update to an item is made, hence the number of distinct elements from one step to the next can change by $\pm 1$, or not at all (i.e., by zero).

    For the neighboring relation, we will argue about the set $\mathcal{S}_1$ from \Cref{lem:cdistinct-s1}.
    Note that $\mathcal{S}_1$ is the set of all integer vectors of length $T$ with interval sums that are no larger than $1$ in absolute value.
    Let $v\in \mathcal{S}_1$.
    By definition, $v\in\{-1, 0, 1\}^T$, since otherwise there would exist a $k\in[T]$ for which $\lvert v_k \rvert > 1$, and so the singleton interval sum (from $k$ to $k$) would exceed 1 in absolute value.

    Assume, towards a contradiction, that there exists $i, j \in[T]$ for which the prefix sum of $v$ up to $i$ and up to  $j$ assume a value of $-1$ and $1$ respectively.
    Then,
    \begin{equation*}
        \left\lvert \sum_{k=1}^i v_k - \sum_{k=1}^{j} v_k \right\rvert
        = \left\lvert \sum_{k=\min\{i, j\}+1}^{\max\{i,j\}} v_k \right\rvert = 2\,,
    \end{equation*}
    implying $v\notin\mathcal{S}_1$, a contradiction.
    Hence every prefix on $v$ lies in either $\{0, 1\}$ or $\{-1, 0\}$.

    As $d(\sigma) - d(\sigma') \in \mathcal{S}_1$ from \Cref{lem:cdistinct-s1}, we have that the difference in the prefixes on $d(\sigma)$ and $d(\sigma')$ satisfy the same condition.
    This is exactly the condition for $d(\sigma) \sim_p d(\sigma')$, and so we are done.
\end{proof}

The following theorem follows naturally by a reduction from $\cdistinct$. 
\begin{theorem}\label{thm:trits-prefix-lb}
      Let $\eps\in(0, 1/2], \delta\in[0, 1]$, and $T\in\N$ be sufficiently large.
      Then any $(\alpha, \frac{1}{100})$-accurate $(\eps, \delta)$-DP mechanism for $\sim_p$-neighboring continual counting on $\{-1, 0, 1\}$ over $T$ steps satisfies:
      \begin{enumerate}
          \item If $\delta = o(\eps / T)$, then $\alpha = \tilde{\Omega}\left(\min\left\{\frac{T^{1/3}}{\eps^{2/3}}, T\right\}\right)$.
          \item If $\delta = 0$, then $\alpha = \Omega\left(\min\left\{\sqrt{\frac{T}{\eps}}, T\right\}\right)$.
      \end{enumerate}
\end{theorem}
\begin{proof}
    We prove the claim by a reduction from $\cdistinct$, following \citeauthor{jain23distinct} and \citeauthor{andersson26improved}.
    Let $\cc$ be an $(\eps, \delta)$-DP mechanism for $\ccptrits{T}$ and suppose that $\cc$ is $(\alpha, 1/100)$-accurate.
    We construct a mechanism $\cM$ for $\cdistinct$ as follows.
    On input stream $\sigma$, compute the output difference stream $d(\sigma)\in\{-1,0,1\}^T$, and output $\cc(d(\sigma))$.
    This can be done in an online manner, as $d(\sigma)_t = f(\sigma)_t - f(\sigma)_{t-1}$.
    Let $\sigma, \sigma'$ be neighboring inputs to $\cdistinct$.
    By \Cref{clm:difference-streams}, we have that $d(\sigma), d(\sigma') \in\{-1,0,1\}^T$ and that $d(\sigma)\sim_p d(\sigma')$.
    Hence, $d(\sigma), d(\sigma')$ are valid, neighboring inputs to $\ccptrits{T}$.
    Thus, as $\cc$ is an $(\eps,\delta)$-DP mechanism for $\ccptrits{T}$, $\cM$ is an $(\eps,\delta)$-DP mechanism for $\cdistinct$.
    Moreover, as prefix sums on $d(\sigma)$ equals the distinct-count sequence $f(\sigma)$, we have, by the $(\alpha, 1/100)$-accuracy of $\cc$, $(\alpha, 1/100)$-accuracy for $\cM$ for $\cdistinct$.
    Consequently, $\cM$ inherits the lower bounds from \Cref{thm:cdistinct-lb}, and the stated bounds on $\alpha$ follow.
\end{proof}

\subsection{A lower bound for \texorpdfstring{$\ccpbits{T}$}{Count}}

Having proved a lower bound for $\sim_p$-neighboring continual counting on $\{-1, 0, 1\}$-streams, we will next extend it to binary streams.
The key idea will be to encode strings $x,x'\in\{-1,0,1\}^T$ into strings $s,s'\in\{0, 1\}^{2T}$ in such a way that:
\begin{enumerate}
    \item Prefix sums on $x$ and $x'$ can be derived from $s$ and $s'$ respectively.
    \item If $x$ and $x'$ are prefix-sum neighboring, then so are $s$ and $s'$.
\end{enumerate}
Equipped with such an encoding, the reduction from $\ccptrits{T}$ to $\ccpbits{T}$ proceeds naturally.
We give the encoding next.

\paragraph{Subroutine $\pmonetobin$.}

The procedure $\pmonetobin$ maps a sequence $x \in \{-1,0,1\}^{T}$ into a binary sequence $s \in \{0,1\}^{2T}$, where for each $i \in [T]$,
\begin{equation*}
    (s_{2i-1}, s_{2i}) =
    \begin{cases}
        (0,1), & \text{if } x_i = 0, \\[4pt]
        (1,1), & \text{if } x_i = 1, \\[4pt]
        (0,0), & \text{if } x_i = -1,
    \end{cases}
    \qquad \forall i \in [T].
\end{equation*}
As can be shown by a straightforward induction, by construction, for every $i \in [T]$,
\begin{equation}
    \sum_{j=1}^{2i} s_j = \sum_{j=1}^i x_j + i.\label{eq:bits-to-trit-prefix}
\end{equation}

\begin{lemma}\label{lem:neighbor-pmonetobin}
    Let $x,x' \in \{-1,0,1\}^{T}$ be $\sim_p$-neighboring sequences.
    Define $s=\pmonetobin(x)$ and $s'=\pmonetobin(x')$.
    Then $s, s'\in\{0, 1\}^{2T}$ are $\sim_p$-neighboring sequences.
\end{lemma}
\begin{proof}
    For $i\in[T]$ and $k\in[2T]$, define the prefix-differences
    \begin{equation*}
        P_i := \sum_{t=1}^i x_t - \sum_{t=1}^i x'_t,
        \qquad D_k := \sum_{j=1}^k s_j - \sum_{j=1}^k s_j',
    \end{equation*}
    and also set $P_0=D_0 := 0$.
    
    \paragraph{Even prefixes.}
    By the defining property of $\pmonetobin$, for every $i\in[T]$,
    \begin{equation*}
        \sum_{j=1}^{2i} s_j = \sum_{t=1}^i x_t + i
        \qquad\text{and}
        \qquad\sum_{j=1}^{2i} s_j' = \sum_{t=1}^i x_t' + i,
    \end{equation*}
    and hence
    \begin{equation}\label{eq:even-prefix-diff}
    D_{2i}
    = \Big(\sum_{t=1}^i x_t + i\Big) - \Big(\sum_{t=1}^i x_t' + i\Big)
    = P_i.
    \end{equation}
    
    \paragraph{Odd prefixes.}
    Write $a_i := s_{2i-1}$ and $a'_i := s'_{2i-1}$.
    By the definition of $\pmonetobin$, $a_i=\indic[x_i=1]$ and $a'_i=\indic[x'_i=1]$.
    Moreover,
    \begin{equation*}
        \sum_{j=1}^{2i-1} s_j = a_i + \sum_{j=1}^{2(i-1)} s_j,
        \qquad\sum_{j=1}^{2i-1} s'_j = a_i' + \sum_{j=1}^{2(i-1)} s'_j,
    \end{equation*}
    so
    \begin{equation}\label{eq:odd-prefix-diff}
        D_{2i-1}
        = D_{2(i-1)} + (a_i-a'_i)
        = P_{i-1} + \big(\indic[x_i=1]-\indic[x'_i=1]\big).
    \end{equation}
    
    We now verify the prefix-neighbor condition for $s,s'$.
    Since $x\sim_p x'$, we are in one of the following two cases.
    
    \paragraph{Case 1: $P_i\in\{0,1\}$ for all $i\in[T]$.}
    Then by~\eqref{eq:even-prefix-diff}, $D_{2i}\in\{0,1\}$ for all $i$.
    It remains to show $D_{2i-1}\in\{0,1\}$.
    From~\eqref{eq:odd-prefix-diff}, we have $D_{2i-1}\in\{-1,0,1,2\}$, so it suffices to rule out the values $-1$ and $2$.
    
    Suppose $D_{2i-1}=-1$ for some $i$.
    Then~\eqref{eq:odd-prefix-diff} forces $P_{i-1}=0$ and $\indic[x_i=1]-\indic[x'_i=1]=-1$,
    i.e., $x_i\neq 1$ and $x_i' = 1$.
    Hence $x_i-x'_i \leq -1$, and therefore
    \begin{equation*}
        P_i = P_{i-1} + (x_i-x'_i) \leq 0 + (-1) = -1,
    \end{equation*}
    contradicting $P_i\in\{0,1\}$.
    
    Suppose instead that $D_{2i-1}=2$ for some $i$.
    Then~\eqref{eq:odd-prefix-diff} forces $P_{i-1}=1$ and $\indic[x_i=1]-\indic[x'_i=1]=1$,
    i.e., $x_i=1$ and $x'_i\neq 1$.
    Hence $x_i-x'_i\ge 1$, and therefore
    \begin{equation*}
        P_i = P_{i-1} + (x_i-x'_i) \ge 1 + 1 = 2,
    \end{equation*}
    contradicting $P_i\in\{0,1\}$.
    
    Thus $D_{2i-1}\notin\{-1,2\}$, so $D_{2i-1}\in\{0,1\}$ for all $i$.
    Combined with $D_{2i}\in\{0,1\}$, this implies $D_k\in\{0,1\}$ for all $k\in[2T]$,
    i.e., $s\sim_p s'$.
    
    \paragraph{Case 2: $P_i\in\{0,-1\}$ for all $i\in[T]$.}
    Then by~\eqref{eq:even-prefix-diff}, $D_{2i}\in\{0,-1\}$ for all $i$.
    Again it remains to show $D_{2i-1}\in\{0,-1\}$.
    From~\eqref{eq:odd-prefix-diff}, we have $D_{2i-1}\in\{-2,-1,0,1\}$, so it suffices to rule out $1$ and $-2$.
    
    If $D_{2i-1}=1$, then~\eqref{eq:odd-prefix-diff} forces $P_{i-1}=0$ and $\indic[x_i=1]-\indic[x'_i=1]=1$,
    so $x_i=1$ and $x'_i\neq 1$, implying $x_i-x'_i\ge 1$ and hence
    \begin{equation*}
        P_i = P_{i-1} + (x_i-x'_i) \ge 1,
    \end{equation*}
    contradicting $P_i\in\{0,-1\}$.
    
    If $D_{2i-1}=-2$, then~\eqref{eq:odd-prefix-diff} forces $P_{i-1}=-1$ and $\indic[x_i=1]-\indic[x'_i=1]=-1$,
    so $x'_i=1$ and $x_i\neq 1$, implying $x_i-x_i' \leq -1$ and hence
    \begin{equation*}
        P_i = P_{i-1} + (x_i-x'_i) \leq -2,
    \end{equation*}
    contradicting $P_i\in\{0,-1\}$.
    
    Thus $D_{2i-1}\notin\{1,-2\}$, so $D_{2i-1}\in\{0,-1\}$ for all $i$.
    Combined with $D_{2i}\in\{0,-1\}$, this implies $D_k\in\{0,-1\}$ for all $k\in[2T]$,
    i.e., $s\sim_p s'$.
    
    In both cases $s\sim_p s'$, proving the statement.
\end{proof}

The proof of the following lower bound is a natural reduction from \Cref{thm:trits-prefix-lb}.
\begin{theorem}\label{thm:bits-prefix-lb}
      Let $\eps\in(0, 1/2], \delta\in[0, 1]$, and $T\in\N$ be sufficiently large.
      Then any $(\alpha, \frac{1}{100})$-accurate $(\eps, \delta)$-DP mechanism for $\sim_p$-neighboring continual counting on $\{0, 1\}$ over $T$ steps satisfies:
      \begin{enumerate}
          \item If $\delta = o(\eps / T)$, then $\alpha = \tilde{\Omega}\left(\min\left\{\frac{T^{1/3}}{\eps^{2/3}}, T\right\}\right)$.
          \item If $\delta = 0$, then $\alpha = \Omega\left(\min\left\{\sqrt{\frac{T}{\eps}}, T\right\}\right)$.
      \end{enumerate}
\end{theorem}
\begin{proof}
    Assume wlog.~that $T$ is even and let $n = T/2$.
    We prove the claim by a reduction from $\ccptrits{n}$.
    Let $\cc$ be an $(\eps, \delta)$-DP mechanism for $\ccpbits{T}$, and suppose that $\cc$ is $(\alpha, 1/100)$-accurate.
    Let $x \in\{-1,0,1\}^n$ be the input to an instance of $\ccptrits{n}$, and let $s = \pmonetobin(x)\in\{0,1\}^{T}$.
    We define $\cM$ as the mechanism performing the following routine on receiving $x_t$:
    \begin{enumerate}
        \item Construct $(s_{2t-1}, s_{2t})$ from $x_t$.
        \item Feed $s_{2t-1}$ followed by $s_{2t}$ to $\cc$, yielding corresponding outputs $a_{2t-1}$ and $a_{2t}$ from $\cc$.
        \item Output $y_t = a_{2t} - t$.
    \end{enumerate}
    Let $x, x'\in\{-1, 0, 1\}^n$ be neighboring inputs to our problem instance, i.e., $x\sim_p x'$.
    Let $s' = \pmonetobin(x')\in\{0,1\}^{T}$.
    By \Cref{lem:neighbor-pmonetobin}, we have that $s\sim_p s'$.
    Hence $s, s'$ are valid neighboring inputs to $\ccpbits{T}$.
    Therefore, since $\cc$ is $(\eps,\delta)$-DP for $\ccpbits{T}$, and $\cM$ outputs a postprocessing of $\cc$, $\cM$ is an $(\eps,\delta)$-DP mechanism for $\ccptrits{n}$.
    
    Moreover, $y_t = a_{2t} - t$ is a noisy estimate of $\sum_{i=1}^t x_i$ by \eqref{eq:bits-to-trit-prefix}.
    Hence, by the $(\alpha, 1/100)$-accuracy of $\cc$, we also have $(\alpha, 1/100)$-accuracy for $\cM$.
    Consequently, $\cM$ inherits the lower bounds from \Cref{thm:trits-prefix-lb}, where 
    \begin{enumerate}
          \item If $\delta = o(\eps / n) = o(\eps / T)$, then $\alpha = \tilde{\Omega}\left(\min\left\{\frac{n^{1/3}}{\eps^{2/3}}, n\right\}\right) = \tilde{\Omega}\left(\min\left\{\frac{T^{1/3}}{\eps^{2/3}}, T\right\}\right)$.
          \item If $\delta = 0$, then $\alpha = \Omega\left(\min\left\{\sqrt{\frac{n}{\eps}}, n\right\}\right) = \Omega\left(\min\left\{\sqrt{\frac{T}{\eps}}, T\right\}\right)$.
    \end{enumerate}
    since $n = \Theta(T)$.
    Finally, if $T$ is odd then repeat the proof for $T - 1 = \Theta(T)$ for the corresponding lower bound statement.
\end{proof}

\subsection{A lower bound for \texorpdfstring{$\ccebits{T}$}{Count}}
Fix $T\in\N$ for the rest of this section.
We will prove a lower bound for $\ccebits{T}$ for a particular class of mechanisms, namely those that are \emph{data-independent}.
We give a definition next in the context of continual release problems, which we first formally define.
\begin{definition}[Continual Release Problem]\label{def:continual-release-problem}
    Let $\mathcal{U}$ denote an update universe, and let $\mathcal{D}$ denote a dataset space.
    Let $U:\mathcal{D}\times\mathcal{U}\to\mathcal{D}$ be an update operator, and fix an initial dataset $D_0 = \emptyset$.
    For an input stream $\sigma=(u_1,\dots,u_T)\in\mathcal{U}^T$, define the induced datasets $(D_t(\sigma))_{t=0}^T$ recursively by
    \begin{equation*}
        D_t(\sigma) := U(D_{t-1}(\sigma),u_t)\qquad\text{for }t\in[T].
    \end{equation*}
    Let $f:\mathcal{D}\to\R$ be a real-valued function (a \emph{query}).
    We define the induced continual query $f:\mathcal{U}^T\to\R^T$ (with a slight abuse of notation) by
    \begin{equation*}
        f(\sigma) := \left(f(D_1(\sigma)),\dots,f(D_T(\sigma))\right).
    \end{equation*}
    Let $\sim$ be a neighboring relation over input streams in $\mathcal{U}^T$.
    We call $P=(f,\sim)$ a \emph{continual release problem}.
    A randomized mechanism $\cM$ is \emph{$(\eps,\delta)$-DP for $P$} if it satisfies $(\eps,\delta)$-DP, under continual observation, with respect to the neighboring relation $\sim$ on $\mathcal U^T$ for the mapping $\sigma \mapsto \cM(\sigma)\in\R^T$.
\end{definition}
For convenience, when discussing a continual release problem $P=(f, \sim)$, we assume the corresponding surrounding notation (e.g., input streams are denoted by $\sigma$).
We next formalize what it means to be data-independent for such problems.
\begin{definition}[Data-Independent Mechanism]\label{def:input-independent-additive-error}
    Let $P = (f, \sim)$ be a continual release problem.
    Let $\cM(\sigma)$ denote the vector of outputs produced by $\cM$ on input stream $\sigma$ to $P$.
    We say that $\mathcal{M}$ is a \emph{data-independent} mechanism, if there exists a distribution $\mu$ on $\mathbb{R}^T$, such that, for any input stream $\sigma$ of length $T$,
    \begin{equation*}
        \cM(\sigma) \stackrel{d}{=} f(\sigma) + Z,
    \end{equation*}
    where $Z$ is drawn from $\mu$, and the law of $Z$ is \emph{independent} of $\sigma$.
\end{definition}
Intuitively, a data-independent mechanism is distributionally equivalent to a mechanism that adds input-independent, possibly correlated across time steps, noise to the true function values.
Notably, it includes all mechanisms based on \emph{matrix-factorization} techniques \citep{li2015matrix,edmonds2020power}.

We remark that the notion of data-independent mechanism, together with related ideas in this section, appear in the literature on differentially private linear queries, see e.g., \citep{HardtT10geometry,li2015matrix,bhaskara2012unconditional,edmonds2020power,awan21structure,nikolov13geometry} and references therein.
These works primarily deal with $\ell_1$ neighboring inputs, and not the more elaborate neighboring relations we do.
Nonetheless, our key lemma in this section (\Cref{lem:sens-polytope-relation}) follows from basic properties of data-independent mechanism and \emph{sensitivity sets} (sometimes also referred to as a \emph{sensitivity polytope}, when appropriate).
We define this last notion next.

\begin{definition}[Sensitivity Set]\label{def:sensitivity-set}
    Let $P=(f, \sim)$ be a continual release problem. The \emph{sensitivity set} of $P$ is defined as
    \[
    S_{P} := \{\, f(\sigma) - f(\sigma') \mid \sigma \sim \sigma' \,\}.
    \]
\end{definition}

Rather than constructing a full reduction, as we did previously, we will prove a lower bound based on \emph{sensitivity-set containment}.
More precisely, we show that if the sensitivity set of $P'=(f', \sim')$ is contained in the sensitivity set of $P=(f, \sim)$, then for every data-independent mechanism for $P$ with additive error $\alpha$, there exists a mechanism for $P'$ with the same error $\alpha$.
Hence, we are able to lift a lower bound from $P'$ to any data-independent mechanism for $P$.
While not all mechanisms are data-independent, it shows inherent limitations for a broad and natural class of constructions (notably including factorization mechanisms).
It also highlights that any hope of achieving accuracy beyond these bounds requires mechanisms whose additive error depends on the input stream.

The following lemma formalizes this observation.
\begin{lemma}\label{lem:sens-polytope-relation}
    Consider two continual release problems $P=(f, \sim)$ and $P'=(f', \sim')$.
    Let $S_{P}$ and $S_{P'}$ denote the corresponding sensitivity sets, as defined in \Cref{def:sensitivity-set}.
    Suppose that
    \begin{equation*}
        S_{P'} \subseteq S_{P}.
    \end{equation*}
    For $\eps \geq 0, \delta\in[0, 1), \beta\in(0, 1)$, let $\cM$ be a \emph{data-independent} $(\eps, \delta)$-DP mechanism for $P$ that is $(\alpha, \beta)$-accurate.
    Then there exists an $(\eps, \delta)$-DP mechanism $\cM'$ for $P'$, that is also $(\alpha, \beta)$-accurate.
\end{lemma}
\begin{proof}
    By $\cM$ being data-independent, there exists a distribution $\mu$ on $\mathbb{R}^T$, such that for every $\sigma$ of length $T$, $\cM(\sigma) \stackrel{d}{=} f(\sigma) + Z$ with $Z$ distributed according to $\mu$.
    Wlog assume that $\cM(\sigma) = f(\sigma) +Z$.
    By \Cref{def:dp}, $\cM$ is $(\eps,\delta)$-DP for $P$ if and only if, for all $\sigma\sim\sigma'$, and all measurable $O\subseteq \R^T$:
    \begin{align*}
        \Pr[f(\sigma) + Z \in O] &\leq e^{\eps}\Pr[f(\sigma') + Z \in O] + \delta,\\
        \Pr[f(\sigma') + Z \in O] &\leq e^{\eps}\Pr[f(\sigma) + Z \in O] + \delta.
    \end{align*}
    Let $\sigma\sim\sigma'$ and measurable $O\subseteq\R^T$ be arbitrary.
    Let $O' := \{y - f(\sigma) : y\in O\}$.
    Then,
    \begin{align*}
        f(\sigma) + Z \in O &\iff Z\in O'\,,\\
        f(\sigma') + Z \in O &\iff Z + f(\sigma') - f(\sigma)\in O'\,.
    \end{align*}
    Moreover, $O\mapsto O'$ is a bijection on all measurable subsets of $\mathbb{R}^T$.
    Hence quantifying over all measurable $O'$ is equivalent to quantifying over all measurable $O$.
    Therefore (after relabeling), we have equivalent inequalities:
    \begin{align*}
        \Pr[Z \in O] &\leq e^{\eps}\Pr[Z + f(\sigma') - f(\sigma) \in O] + \delta,\\
        \Pr[Z + f(\sigma') - f(\sigma) \in O] &\leq e^{\eps}\Pr[Z \in O] + \delta.
    \end{align*}
    Observe that the dependence on $\sigma, \sigma'$ only appears in the term $f(\sigma) - f(\sigma') \in S_{P}$.
    It follows that $\cM$ is $(\eps,\delta)$-DP for $P$, if and only if, for all $\lambda \in S_{P}$, and all measurable $O\subseteq \R^T$:
    \begin{align*}
        \Pr[Z \in O] &\leq e^{\eps}\Pr[Z - \lambda \in O] + \delta,\\
        \Pr[Z - \lambda \in O] &\leq e^{\eps}\Pr[Z \in O] + \delta.
    \end{align*}
    Repeating the argument for a data-independent $\cM'$ for problem $P'$, and additive noise $Z'\in\mathbb{R}^T$, you arrive at corresponding equations.
    For completeness: $\cM'$ is $(\eps,\delta)$-DP for $P'$, if and only if, for all $\lambda \in S_{P'}$, and all measurable $O\subseteq \R^T$:
    \begin{align*}
        \Pr[Z' \in O] &\leq e^{\eps}\Pr[Z' - \lambda \in O] + \delta,\\
        \Pr[Z' - \lambda \in O] &\leq e^{\eps}\Pr[Z' \in O] + \delta.
    \end{align*}
    By the set inclusion $S_{P'} \subseteq S_{P}$, if the additive-noise $Z$ provides privacy for $P$, then using $Z'= Z$ would provide privacy for $P'$.
    As the noise distributions $Z$ and $Z'$ uniquely determine the additive error of $\cM$ and $\cM'$ respectively, the statement follows.
\end{proof}

To prove our lower bound, we will show that the sensitivity set for continual counting on \emph{edit-neighboring} streams, is at least as expressive as that for continual counting under \emph{prefix-sum neighboring} streams.
\begin{lemma}\label{lem:sensitivity-sets-prefix-and-queue-neighboring}
    Let $P = \ccebits{T}$ and $P' = \ccpbits{T}$.
    Then,
    \begin{enumerate}
        \item $S_{P} \supseteq \{0,1\}^T \;\cup\; \{0,-1\}^T$.
        \item $S_{P'} = \{0,1\}^T \;\cup\; \{0,-1\}^T$.
    \end{enumerate}
\end{lemma}

\begin{proof}
For the input stream $\sigma = (x_1, \dots, x_T)$ in $\{0, 1, \bot\}^T$, recall that $f(\sigma)_t = f(D_t) = \sum_{i=1}^t x_i$, where $\bot$ summands are treated as zeros.
We prove each item separately.
\paragraph{Item 1.}
    We begin by proving that
    \begin{equation*}
        S_{P} \supseteq \{0,1\}^T.
    \end{equation*}
    Fix $y\in\{0, 1\}^T$.
    Consider the two sequences $\sigma = (1, z_1, \dots, z_{T-1})$ and $\sigma'=(z_1, \dots, z_T)$ from $\{0, 1\}^T$ (we do not use any $\bot$'s), where $\sigma\sim_e\sigma'$ ($\sigma$ is an edit-insert neighbor at $i=1$).
    For any $t\in[T]$, we have that
    \begin{equation*}
        f(\sigma)_t - f(\sigma')_t = 1 - z_t
    \end{equation*}
    where, on setting $z_t = 1 - y_t \in\{0, 1\}$, we get that $y=f(\sigma)-f(\sigma')\in S_{P}$.
    Hence $S_{P} \supseteq \{0, 1\}^T$.
    
    The argument for proving $S_{P} \supseteq \{0,-1\}^T$ proceeds similarly.
    Fix $y\in\{0, -1\}^T$, and instead consider $\sigma = (0, z_1, \dots, z_{T-1})$ and $\sigma'=(z_1, \dots, z_T)$ from $\{0, 1\}^T$.
    Now set $z_1, \dots, z_T$ based on
    \begin{equation*}
        y_t = f(\sigma)_t - f(\sigma')_t = -z_t,
    \end{equation*}
    yielding $y = f(\sigma) - f(\sigma') \in S_{P}$, and so $S_{P} \supseteq \{0, -1\}^T$.
    
    It follows that $S_{P} \supseteq \{0, 1\}^T \cup \{0, -1\}^T$, completing the proof of the first item.
\paragraph{Item 2.}
    The definition of prefix-sum neighboring (\Cref{def:prefix-neighboring}) explicitly enforces that:
    \begin{equation*}
        S_{P'} \subseteq \{0,1\}^T \cup \{0, -1\}^T.
    \end{equation*}
    Towards proving the other direction, we will show that
    \begin{equation*}
        S_{P'} \supseteq \{0,1\}^T.
    \end{equation*}
    Fix $y \in \{0, 1\}^T$.
    Next, define the sequences $\sigma=(x_1, \dots, x_T), \sigma'=(x_1', \dots, x_T')$ in $\{0, 1\}^T$ via
    \begin{equation*}
        x_t = \begin{cases}
           1\quad &\text{if}\; y_t - y_{t-1} = 1\\
           0\quad &\text{otherwise}
        \end{cases}
        \qquad
        x_t' = \begin{cases}
           1\quad &\text{if}\; y_t - y_{t-1} = -1\\
           0\quad &\text{otherwise}
        \end{cases}
    \end{equation*}
    where we define $y_0 = 0$.
    We will argue that $\sigma \sim_p \sigma'$.
    Indeed, for any $t\in[T]$:
    \begin{align*}
        f(\sigma)_t - f(\sigma')_t
        &= \sum_{i=1}^t x_i - \sum_{i=1}^{t} x_i'
        = \sum_{i=1}^{t} (x_i - x_i')\\
        &= \sum_{i=1}^{t} (y_i - y_{i-1})
        = y_t\in\{0, 1\}.
    \end{align*}
    Hence, we have that $y = f(\sigma) - f(\sigma') \in S_{P'}$, proving $S_{P'} \supseteq \{0, 1\}^T$.
    
    By a symmetrical argument, we also have that $S_{P'} \supseteq \{0, -1\}^T$, allowing us to conclude:
    \begin{equation*}
        S_{P'} \supseteq \{0,1\}^T \cup \{0, -1\}^T.
    \end{equation*}
    This finishes the proof of the second item, and the lemma.
\end{proof}

We are now ready to give our lower bound for $\ccebits{T}$.
\begin{theorem}
      Let $\eps\in(0, 1/2], \delta\in[0, 1)$, and $T\in\N$ be sufficiently large.
      Then any $(\alpha, \frac{1}{100})$-accurate $(\eps, \delta)$-DP \emph{data-independent} algorithm (see~\Cref{def:input-independent-additive-error}) for $\sim_e$-neighboring continual counting on $\{0, 1, \bot\}$ over $T$ steps satisfies:
      \begin{enumerate}
          \item If $\delta = o(\eps / T)$, then $\alpha = \tilde{\Omega}\left(\min\left\{\frac{T^{1/3}}{\eps^{2/3}}, T\right\}\right)$.
          \item If $\delta = 0$, then $\alpha = \Omega\left(\min\left\{\sqrt{\frac{T}{\eps}}, T\right\}\right)$.
      \end{enumerate}
\end{theorem}
\begin{proof}
    Let $P = \ccebits{T}$ and $P' = \ccpbits{T}$.
    Consider a data-independent $(\eps, \delta)$-DP mechanism $\cM$ for $P$ that is $(\alpha, 1/100)$-accurate.
    By \Cref{lem:sensitivity-sets-prefix-and-queue-neighboring} $S_{P'} \subseteq S_P$, and so by \Cref{lem:sens-polytope-relation} there exists an $(\eps, \delta)$-DP mechanism $\cM'$ for $P'$ that is $(\alpha, 1/100)$-accurate.
    Applying \Cref{thm:bits-prefix-lb} to $\cM'$ yields the stated lower bound on $\alpha$.
\end{proof}
\section{Empirical Work}\label{sec:empirical-work}

\begin{figure*}[t]
  \centering
  \begin{subfigure}[b]{0.32\linewidth}
    \centering
    \includegraphics[width=\linewidth]{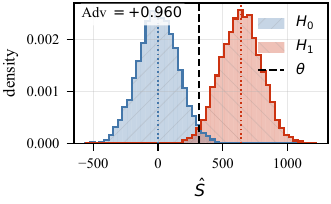}
    \caption{\basecc at $\eps = 2.78$.}
    \label{fig:hist-basecc}
  \end{subfigure}
  \hfill
  \begin{subfigure}[b]{0.32\linewidth}
    \centering
    \includegraphics[width=\linewidth]{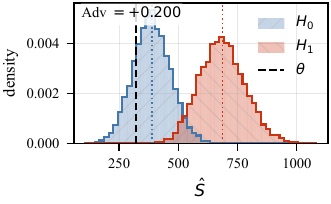}
    \caption{\partitioncc at $\eps = 2.78$.}
    \label{fig:hist-partitioncc}
  \end{subfigure}
  \hfill
  \begin{subfigure}[b]{0.32\linewidth}
    \centering
    \includegraphics[width=\linewidth]{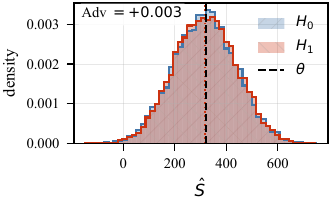}
    \caption{\simecc at $\eps = 150$.}
    \label{fig:hist-simpleecc}
  \end{subfigure}
  \caption{
    Plots of (normalized) histograms over the statistic $\hat{S}$ from the experiment setting in \Cref{fig:adv-vs-rmse-main}.
    Each plot shows the corresponding histogram over the statistic $\hat{S}$ under $H_0$ (blue; left-hatched) vs.\ $H_1$ (red, right-hatched), based on drawing $N=20000$ streams, with the strong vertical dashed line denoting the decision threshold $\theta$.
    The weaker dashed lines correspond to the mean value of the statistic under $H_0$ and $H_1$.
    No weak dashed lines are visible in \Cref{fig:hist-simpleecc} as they overlap with the line for $\theta$.
    }
  \label{fig:stat-hist}
\end{figure*}

We empirically substantiate the claim that swap-private continual counters are not a good substitute for an edit-private mechanism: when calibrated to defend against a natural distinguishing attack on edit-neighboring streams, they incur meaningfully larger error than our mechanisms for comparable attack success probability.\footnote{Our source code is available: \url{https://github.com/jodander/CCS26}.} More specifically, we construct the attack on a parametric family of structured Bernoulli streams, and sweep privacy parameters to trace the trade-off between how often an attack succeeds, and the respective errors incurred by different continual counters.

\paragraph{Formal context for experiments.}
We formulate the distinguishing attack as a guessing game, a standard approach in the DP auditing literature~\cite{jagielskiUO20,nasr21adversary,steinke23auditing}. 
The defender draws a stream $\sigma^{(0)}\sim \mathcal{D}$ (a distribution over input streams described below), constructs an edit-neighbor of $\sigma^{(0)}$ denoted $\sigma^{(1)}$, samples a challenge bit $b\in\{0,1\}$, and releases $y = \cM(\sigma^{(b)})$.
The attacker, given knowledge of $y, \mathcal{D}$, $\cM$ and how $\sigma^{(1)}$ is constructed for any choice of $\sigma^{(0)}$, outputs a guess $\hat{b}\in\{0, 1\}$ for the value of the challenge bit.
We measure the attacker's performance via their \emph{advantage}, as done in~\cite{cherubin17bayes,swanberg25beyond}.
We use the signed attack advantage defined by the expression
\begin{equation*}
    Adv = 2\Pr[\hat{b} = b] - 1 \in[-1, 1].
\end{equation*}
Note that the advantage is defined over the randomness of the stream(s), in addition to the randomness of the challenge bit, and the privatizing mechanism.
As such, it captures a notion of \enquote{average} performance of the attack for a distribution over streams.
In contrast with work on DP auditing, we use advantage as a descriptive statistic for attacker performance under $\mathcal{D}$, not as a tool for refuting DP guarantees. It captures how often  a given attacker distinguishes two streams.

We study the special case where $\mathcal{D}$ is a \emph{Bernoulli stream} of rate $p=(p_1, \dots, p_T)\in[0,1]^T$, i.e.,~the stream $\sigma^{(0)} = (x_1, \dots, x_T)$ is drawn i.i.d.\ via $x_t \sim \mathrm{Bernoulli}(p_t)$.
For constructing a neighboring stream, we use $\sigma^{(1)} = f(\sigma^{(0)}) = (1, x_1, \dots, x_{T-1})$, i.e., we insert a 1 at the first position in the stream and discard the last element.
Intuitively, since we permit the attacker knowledge of $p$, the attacker can leverage the fact that the probability of seeing a 1 at any given point in the stream may differ across the two streams.
More precisely, for $t\geq 2$, $\Pr[\sigma^{(0)}_t = 1] = p_t$ and $\Pr[\sigma^{(1)}_t = 1] = p_{t-1}$.

We will formalize this with a natural correlation attack.
Let $\Delta p, \Delta y \in \mathbb{R}^{T}$ be defined via $\Delta p_t = p_{t} - p_{t-1}$ and $\Delta y_t = y_t - y_{t-1}$ (with $p_0 = y_0 = 0$).
The attacker computes the values $\hat{S}(y, p)$ and $\theta(p)$ defined by the expressions
\begin{align*}
    \hat{S}(y, p) \,:=\, \sum_{t=2}^{T} (p_t - \Delta y_t)\Delta p_t,\quad \theta(p) \,:=\, \frac{1}{2}\sum_{t=2}^{T} (\Delta p_t)^2.
\end{align*}
The attacker then outputs $\hat{b} = \mathbbm{1}[ \hat{S}(y, p) \geq \theta(p)]$.
To motivate the attack, consider an unbiased additive-noise mechanism, so $\mathbb{E}[\Delta y_t | \sigma^{(b)}] = \sigma^{(b)}_t$.
Under $b=0$ we have $\sigma_t^{(b)} = x_t$, while for $b=1$ and $t\geq 2$, we have that $\sigma_t^{(b)} = x_{t-1}$.
Hence the residual $r_t := p_t - \Delta y_t$ has mean $0$ under $b=0$, and mean $\Delta p_t$ under $b=1$.
The attacker can thus exploit a known signal in the noise, and $\hat{S}$ is a natural linear test correlating $r$ against the signature $\Delta p$; the threshold $\theta(p)$ represents the midpoint of the expected value of $\hat{S}$ under the two possible values of $b$.
Restricting the sum to $t\geq 2$ makes the attack independent of the inserted bit's value: it only measures the shift in the rate profile arising from the edit.

\paragraph{Experimental setup.}
We have implemented the following $(\eps, \delta)$-DP (unbounded) continual counters:
\begin{enumerate}
    \item \basecc: Implementation of the 1-step-1-neighboring continual counter from \Cref{lem:cc-standard}, internally using the near-optimal bounded continual counter from~\cite{fichtenberger2022constant,henzinger2023almost}.
    \item \partitioncc: Implementation of the 1-step-1-neighboring continual counter in~\cite[Theorem~3.1]{dwork2015pure}, internally using \basecc.
    For setting the threshold for their partitioning mechanism, we allow the mechanism knowledge of the stream length $T$.
    \item \ecc: \Cref{alg:ecc}, with \basecc as the \enquote{base} 1-step-1-neighboring continual counter.
    \item \simecc: Implementation of \Cref{alg:ecc}, but based on the \enquote{simple partitioning} (\Cref{alg:simple-partition}) with \basecc as the \enquote{base} 1-step-1-neighboring continual counter.
    \item \gausscc: An edit-private continual counter, based on at each step making a fresh release of the running sum with the Gaussian mechanism, and using \emph{Gaussian DP}~\cite{dong2022gdp} for the composition.
    It is presented in detail as \Cref{alg:edit-gaussian} in \Cref{app:empirical-work-details}.
\end{enumerate}
The full pipeline - from stream generation to mechanism output to distinguishing attack - uses common random numbers.
Each (pseudorandomly) generated trial $(b, \sigma^{(0)})$ is fed to each of the mechanisms being tested, i.e.,~each mechanism receives the same input stream but the pseudorandomness across mechanisms is not shared.

To compare the error of different mechanisms we use the \emph{root-mean-squared error (RMSE)}.
For a continual counter, taking as input $\sigma\in\{0,1\}^T$ and producing an output $y\in\mathbb{R}^T$, the RMSE of its output is defined as $\sqrt{\frac{1}{T}\sum_{t=1}^{T}(y_t - \sum_{j=1}^{t} \sigma_j )^2}$.
We write \emph{per-run} RMSE when we want to emphasize that the error we report is the average RMSE over many trials.

\paragraph{Attack advantage vs.\ error}
Our main empirical result is shown in~\Cref{fig:adv-vs-rmse-main}.
The plot shows the outcome of running and attacking each mechanism (with the exception of \ecc) on inputs from Bernoulli streams with a block structure.\footnote{Including \ecc in the comparison would yield greater error than \simecc because (1) sparsity adaptivity incurs constant-factor overhead, and (2) the stream in \Cref{fig:adv-vs-rmse-main} is dense (expected $\approx T/2$ ones). For sufficiently long sparse streams, \ecc is expected to outperform \simecc.}
Specifically, we consider streams of length $T=10^4$ with a Bernoulli rate $p \in \{0.1, 0.9\}^T$ where the rate changes every 10 steps.
Each point shown corresponds to running one mechanism with privacy parameters $\eps\in[0.1, 150], \delta=10^{-5}$, with per-run RMSE and attack advantage estimated over $N=20000$ runs.
We proceed to make some observations.

Firstly, the results suggest that our correlation attack is effective against both 1-step-1-neighboring counters \basecc and \partitioncc.
This is not too surprising in the case of \basecc\ - an unbiased continual counter, against which the attack derivation directly applies - but it is arguably more surprising in the case of \partitioncc.
While \partitioncc has biased error, the time steps at which it updates its output are nevertheless positively correlated with the Bernoulli rate.
Our experiments indicate that this correlation is sufficient to keep the attack effective.

Secondly, the same attack proves less effective against the edit-neighboring counters: when we fix a bound on the attack advantage and read off the error each mechanism incurs at that bound, the gap is substantial.
\begin{table}[ht]
  \centering
  \caption{Per-run RMSE attainable for given attack advantage 0.1 or 0.05 at 95\% confidence for the experimental setup in \Cref{fig:adv-vs-rmse-main}. Subscript of the reported error denotes the standard error (SE). Across the privacy sweep $\eps\in[0.1, 150]$ \basecc did not achieve an advantage below $0.05$.}
  \label{tab:adv-vs-rmse}
  \begin{tabular}{@{}lcc@{}}
    \toprule
    {Mechanism} & {RMSE @ $Adv \leq 0.1$} & {RMSE @ $Adv \leq 0.05$}\\
    \midrule
    \basecc & \err{280.6}{0.6} & N/A\\
    \partitioncc & \err{107.5}{0.3} & \err{188.8}{0.4}\\
    \gausscc & \err{84.081}{0.005} & \err{325.98}{0.02}\\
    \simecc & \err{31.32}{0.02} & \err{31.32}{0.02}\\
    \bottomrule
  \end{tabular}
\end{table}
From~\Cref{tab:adv-vs-rmse}, we see that on fixing a bound on the attack advantage, the error incurred by \simecc is significantly lower than that incurred by its competitors.
Asymptotically, both \basecc and \gausscc are input-independent additive-noise mechanisms, so by~\Cref{thm:lb-edit-informal}, any calibration that achieves privacy for edit-neighboring streams must incur polynomial $\ell_\infty$ error in $T$.
Our edit-neighboring mechanisms by contrast achieve polylogarithmic error, so a separation at sufficiently large $T$ is predicted by theory; the empirical contribution is to show that the separation manifests at realistic $T$ against a single, simple attack.

Our particular choice of $T=10^4$ is motivated in part by resource constraints: running the setup given in \Cref{fig:adv-vs-rmse-main} over $N=20000$ trials takes a few hours on a commercial MacBook Air (M4).
At smaller values of $T$, we observe that the trade-off between advantage and error is less pronounced, which is primarily explained by our large constant factors.
Nevertheless, we note that our findings are largely robust to the precise parameter setting, and refer the reader to \Cref{fig:robustness} in \Cref{app:empirical-work-details} for evidence to this effect.

To give additional insight into why the attack works, we also plot the associated (normalized) histogram over the statistic in \Cref{fig:stat-hist}. In short, the statistic $\hat{S}$ is effective at separating the signal based on the secret bit $b\in\{0,1\}$ when applied to the 1-step-1-neighboring counters, but in the case of \partitioncc the decision threshold $\theta$ could be fine-tuned; doing so would only widen the gap with \simecc.
When applied to \simecc, the statistic fails to separate the hypotheses, even at $\eps=150$.
We emphasize that this is not to suggest that \simecc is unconditionally robust: it is resistant to this particular, natural attack.

\paragraph{Comparing edit-neighboring counters.}
We finish our empirical section by briefly remarking upon the relative error of the counters for edit-neighboring streams, treated in \Cref{fig:rmse-vs-t-main}.
Here we make the comparison at $(\eps=10, \delta=10^{-5})$-DP, on Bernoulli streams of length $T\in\{10^3, 10^4, 10^5, 10^6\}$ and constant rate 0.01, and over $N=100$ runs per point.
\simecc achieves substantially lower error than \gausscc and lower error than \ecc, except at the smallest values of $T$, where \ecc produces few or no updates.

\subsection{Technical Details}\label{app:empirical-work-details}

\paragraph{Details on \texorpdfstring{\texttt{GaussianCC}}{GaussianCC}}
The baseline for edit-DP we compare against is that of at each time step releasing the full prefix sum with the Gaussian mechanism (\Cref{lem:gaussian-mech}), using that the $\ell_2$ sensitivity under edit neighbors for each prefix sum is $1$.
For this to be private on streams of unbounded lengths, however, we need to set the privacy parameters of each individual release such that they in the limit converge to our target $(\eps, \delta)$-DP guarantee.
\Cref{alg:edit-gaussian} does exactly this.
\begin{algorithm}[b]
\caption{\textsc{GaussianCC}}
\label{alg:edit-gaussian}
\begin{algorithmic}[1]
\Require Privacy parameters $\varepsilon, \delta$; stream $x_1, x_2, \dots \in \mathbb{R}$
\Require Square-summable sequence $(a_t)_{t \ge 1}$ with $\sum_{t \ge 1} a_t^2 \le 1$
\State $\mu \gets$ $\mu$-GDP parameter calibrated s.t.\ $\mu$-GDP implies $(\varepsilon,\delta)$-DP
\State $S_0 \gets 0$
\For{$t = 1, 2, \dots$}
    \State receive $x_t$
    \State $S_t \gets S_{t-1} + x_t$
    \State $\mu_t \gets \mu \cdot a_t$
    \State sample $Z_t \sim \mathcal{N}\!\left(0,\, 1/\mu_t^{2}\right)$
    \State \textbf{release} $\widehat{S}_t \gets S_t + Z_t$
\EndFor
\end{algorithmic}
\end{algorithm}
The only facts of Gaussian DP we leverage is that (1) the Gaussian mechanism with variance $1/\mu^2$ is $\mu$-GDP for sensitivity-1 queries and (2) the composition of $n$ releases at privacy levels $\mu_1, \dots, \mu_n > 0$ satisfy $\sqrt{\sum_{k=1}^n \mu_k^2}$-GDP.
By the square-summability of the sequence $(a_t)_{t\ge 1}$, $\sqrt{\sum_{t=1}^{\infty} \mu_t^2} \leq \mu$.
Hence \Cref{alg:edit-gaussian} satisfies $\mu$-GDP, and therefore also $(\eps, \delta)$-DP.
For the experiments, we chose the sequence $(a_t)_{\geq 1}$ where $a_t = 1/(1.85 \sqrt{t}\ln(t+1))$.

\paragraph{On the applicability of Clopper--Pearson}
For the advantage error bars, we report 95\% Clopper--Pearson confidence intervals on the success rate $q$ of the attacker, and translate this into the advantage via the affine map $Adv = 2q -1$.
This is motivated by the fact that each trial is composed of three independent sampling steps (1) sampling the Bernoulli stream (2) sampling a challenge bit, and (3) sampling noise for the mechanism.
As each of these sources of randomness are kept separate, and the adversary does not adapt across trials, the success probability of each attack is $\mathrm{Bern}(q)$ for some $q$, hence the total number of successful attacks is $\mathrm{Bin}(N, q)$, which justifies Clopper--Pearson on $q$, and the confidence interval carrying over to $Adv$.

\paragraph{On the robustness of our results.}
To demonstrate that our results are not sensitive to the precise parameter setting, we re-compute the results in \Cref{fig:adv-vs-rmse-main} at $T=1000$, and vary other additional parameters.
The results are given in \Cref{fig:robustness}.
\begin{figure*}[t]
  \centering
  \begin{subfigure}[b]{0.49\linewidth}
    \centering
    \includegraphics[width=\linewidth]{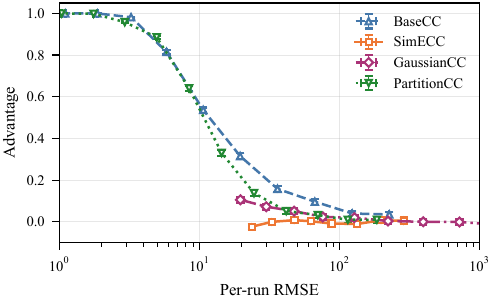}
    \caption{Same parameters, except $T=1000$.}
    \label{fig:exp3-same}
  \end{subfigure}
  \hfill
  \begin{subfigure}[b]{0.49\linewidth}
    \centering
    \includegraphics[width=\linewidth]{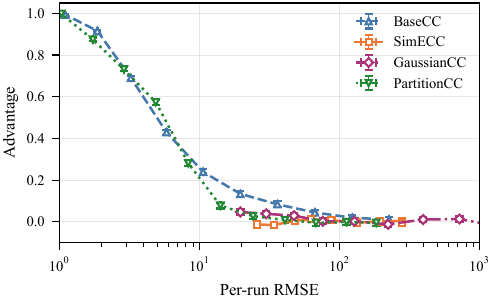}
    \caption{Bernoulli rate flips every 50 steps.}
    \label{fig:exp4-bigger-blocks}
  \end{subfigure}\\
  \begin{subfigure}[b]{0.49\linewidth}
    \centering
    \includegraphics[width=\linewidth]{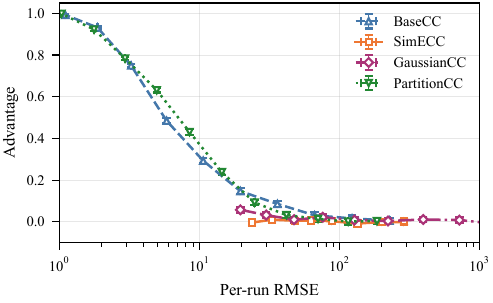}
    \caption{Bernoulli rate $p\in\{0.3, 0.7\}^T$.}
    \label{fig:exp5-smaller-rate-spread}
  \end{subfigure}
  \hfill
  \begin{subfigure}[b]{0.49\linewidth}
    \centering
    \includegraphics[width=\linewidth]{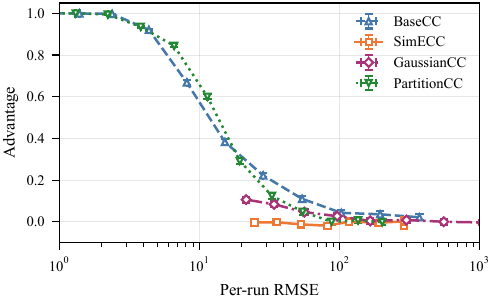}
    \caption{$\delta = 10^{-9}$.}
    \label{fig:exp6-smaller-delta}
  \end{subfigure}\\
  \caption{
    Plots re-running the setup of \Cref{fig:adv-vs-rmse-main} but at $T=10^3$, different seeded randomness and slight parameter variations.
    The base setup is $T=1000$, $\delta=10^{-5}$, alternating Bernoulli rates $p\in\{0.1, 0.9\}^T$ on blocks of width $10$, with each point shown being the average over $N=20000$ runs for one value of $\varepsilon\in\{0.10, 0.19, 0.38, 0.73, 1.4, 2.8, 5.4, 10, 20, 40, 77, 150\}$.
    The error bars on the per-run RMSE are showing the standard error (SE), and the error bars on the advantage are 95\% Clopper--Pearson confidence intervals; $N$ has been set such that the error bars are small enough to be barely visible, if at all.
    \Cref{fig:exp3-same} runs this precise setup, whereas each of \Cref{fig:exp4-bigger-blocks,fig:exp5-smaller-rate-spread,fig:exp6-smaller-delta} change the value of one parameter.
    }
  \label{fig:robustness}
\end{figure*}
The main point we want to highlight is that all four plots are qualitatively quite similar.
While the curves do change shape based on the parameter tuning, they do not change by much.
Compared to \Cref{fig:adv-vs-rmse-main}, the major change at $T=10^3$ is that the benefit of \simecc is not as pronounced.
As we have discussed, this is not surprising: we expect for the polylogarithmic error scaling of \simecc to become more impactful at larger $T$.

\section*{Acknowledgements}
    \noindent\begin{minipage}[t]{\dimexpr\linewidth-5.5cm\relax}%
    \vspace{0pt}%
    \indent This research was supported by the European Research Council (ERC) under the European Union's Horizon 2020 research and innovation programme (Grant agreement No.\ 101019564), and the Austrian Science Fund (FWF) under grant DOI 10.55776/Z422.
    \end{minipage}
    \begin{minipage}[t]{5cm}%
    \vspace{0pt}%
    \includegraphics[width=5cm]{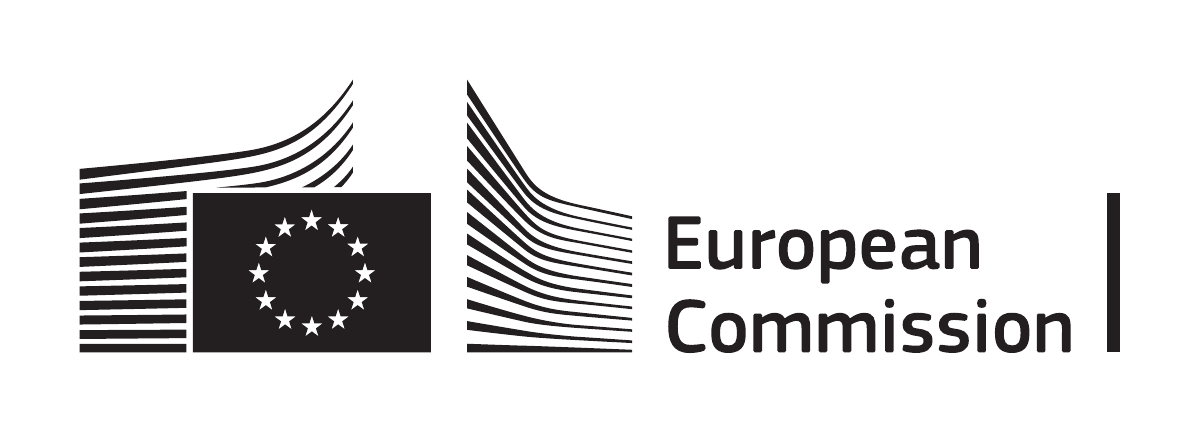}%
    \end{minipage}
    
  \noindent For open access purposes, the authors have applied a CC BY public copyright license to any author-accepted manuscript version arising from this submission. Views and opinions expressed are however those of the author(s) only and do not necessarily reflect those of the European Union or the European Research Council Executive Agency. Neither the European Union nor the granting authority can be held responsible for them.

  \medskip
  \noindent Anamay Chaturvedi was supported by an ISTA Fellowship.

\bibliography{main}
\appendix
\newpage

\section{Results for \enquote{Standard} Continual Counters}\label{app:cc-proofs}

In this section we briefly prove results for continual counting under $1$-step 1-neighboring.
We begin by stating the explicit construction with the tightest known bound on the (maximum) variance for the case of bounded streams.
\begin{lemma}\label{lem:cc-standard-bounded}[\cite{henzinger25normalized}]
    For every $\eps > 0$, $\delta \in (0,1)$ and $T\in\mathbb{N}$, there exists a continual counting mechanism $\cc$ that is $(\eps,\delta)$-DP with respect to the 1-step $1$-neighbor relation for streams of length $T$.
    Given any input stream of length $T$, let $z_t$ denote the error of $\cc$ at time step $t\in[T]$.
    Then,
    \begin{equation*}
        z_t \sim \mathcal{N}(0, \sigma_t^2)\qquad\text{where}\qquad\sigma_t \leq C_{\eps, \delta}\left( 0.846 + \frac{\ln(T)}{\pi} + o(1) \right).
    \end{equation*}
    Here $C_{\eps,\delta}$ is the Gaussian mechanism constant from \Cref{lem:gaussian-mech} and the $o(1)$ term goes to zero as $T\to \infty$.
\end{lemma}
Our work requires continual counters supported on unbounded streams; \Cref{lem:cc-standard} describes such a continual counter.
Before proving the lemma, we remark that the result can be derived given access to any $(\eps,\delta)$-DP continual counter with error matching \Cref{lem:cc-standard-bounded} up to multiplicative constants - e.g., the binary tree mechanism~\cite{chan2011private,dwork2010differentially} based on Gaussian noise - at a cost of worse constants in the error.
\begin{proof}[Proof of \Cref{lem:cc-standard}]
    The construction is based on the so-called \enquote{doubling trick} from \cite{chan2011private}, and essentially matches their Algorithm~4 with the exception that we consider \emph{approximate} DP, and so uses Gaussian rather than Laplace noise.
    We sketch the argument; constants can be improved by optimizing the privacy budget split.

    \paragraph{Construction.}
    Divide the timeline into dyadic intervals $\mathcal{I}_1, \mathcal{I}_2, \dots$, where $\mathcal{I}_{k} = [2^{k-1}, 2^{k}-1]$.
    We define 
    \begin{equation*}
        p_k := \sum_{t\in\mathcal{I}_k} x_t + z_k^{(psum)}\qquad\text{where}\qquad z_k^{(psum)} \sim N(0, C_{\eps/2, \delta/2}^2)
    \end{equation*}
    as the partial sum on $\mathcal{I}_k$, released with the Gaussian mechanism.
    Define the \emph{logarithmic mechanism} $\mathcal{L}$ as the mechanism, that, on receiving $x_t$, outputs
    \begin{equation*}
       \ell_t :=  \sum_{k=1}^{\lfloor \log(t+1) \rfloor} p_k.
    \end{equation*}
    In other words, $\ell_t$ is a noisy prefix sum up to the last completed interval.

    While $\mathcal{L}$ is an unbounded mechanism, its error grows linearly with time.
    To improve it, we run an instance of the continual counter from \Cref{lem:cc-standard-bounded} on each unfinished interval.
    Whenever we enter a new interval $\mathcal{I}_k$, we initialize a fresh $(\eps/2, \delta/2)$-DP continual counter for streams of length $2^{k-1}$ that receives inputs $(x_t)_{t\in\mathcal{I}_k}$, and produces corresponding outputs $(q_t)_{t\in\mathcal{I}_k}$.
    We define our unbounded $(\eps, \delta)$-DP continual counter $\cc$ as the one that, on receiving $x_t$, outputs
    \begin{equation*}
        y_t = \begin{cases}
            \ell_t\quad&\text{if } t = 2^j - 1\text{ for some }j\in\N\\
            \ell_t + q_t\quad&\text{o.w.}
        \end{cases}
    \end{equation*}

    \paragraph{Privacy.}
    Since each partial sum used by $\mathcal{L}$ is (1) a $(\eps/2, \delta/2)$-DP release with the Gaussian mechanism and (2) disjoint, $\mathcal{L}$ itself is $(\eps/2, \delta/2)$-DP via postprocessing.
    Additionally, any input $x_t$ is used as input to exactly one time-bounded $(\eps/2, \delta/2)$-DP continual counter.
    By simple composition, the unbounded mechanism $\cc$ is $(\eps, \delta)$-DP.

    \paragraph{Accuracy.}
    Fix $t \in \N$ where $t\in\mathcal{I}_k$ for $k=\lfloor\log(t+1)\rfloor$.
    Note that by construction, the output can be written as
    \begin{equation*}
        y_t = \sum_{j=1}^{t} x_j + z_t
    \end{equation*}
    where $z_t\sim \mathcal{N}(0, \sigma_t^2)$.
    Assume $t = 2^{j} - 1$ for some $j\in\N$.
    Then by construction:
    \begin{equation*}
        y_t = \sum_{j=1}^{t} x_j + \sum_{i=1}^{\lfloor\log(t+1)\rfloor} z_i^{(psum)}
    \end{equation*}
    and so $\sigma_t = O(C_{\eps/2, \delta/2} \sqrt{\log t})$ in this case.
    For the second case, we also get a contribution to the error from the time-bounded continual counter initialized for $T=\lvert \mathcal{I}_k \rvert = 2^{k} = O(t)$:
    \begin{equation*}
        y_t = \ell_t + q_t = \sum_{j=1}^{t} x_j + \sum_{i=1}^{\lfloor\log(t+1)\rfloor} z_i^{(psum)} + \hat{z}_t
    \end{equation*}
    where $\hat{z}_t$ is the zero-mean Gaussian noise from \Cref{lem:cc-standard-bounded} of standard deviation $O(C_{\eps/2, \delta/2} \log t)$.
    As sums of Gaussians are Gaussian, and the intra-interval error dominates, we can conclude that $\sigma_t = O(C_{\eps/2, \delta/2} \log t)$ for every $t$.

    To give an $\ell_\infty$-error bound and finish the argument, we begin by using a standard Gaussian tail bound.
    For a fixed $t\in\mathbb{N}$ and $\beta\in(0,1)$, we have that:
    \begin{equation*}
        \Pr[ \lvert z_t \rvert \geq \sigma_t \sqrt{2\ln(2/\beta)}] \leq \beta.
    \end{equation*}
    In particular, if we choose a failure probability of $\beta_t = \frac{6\beta}{\pi^2t^2}$ for the output at time $t$, then
    \begin{align*}
        \Pr\big[\exists t \in \N\,&:\, \lvert z_t \rvert \geq \sigma_t \sqrt{2\ln(2/\beta_t)}\big]
        \leq \sum_{t=1}^{\infty} \Pr\left[\lvert z_t \rvert \geq \sigma_t \sqrt{2\ln(\pi^2 t^2/(3\beta))}\right] \leq \sum_{t=1}^{\infty} \beta_t = \beta,
    \end{align*}
    where the first inequality is using a union bound over all time steps, and the last inequality that $\sum_{t=1}^\infty \frac{1}{t^2} = \frac{\pi^2}{6}$.
    Rearranging, we have shown that
    \begin{equation*}
        \Pr[\forall t\in\N\,:\, \lvert z_t \rvert \leq \sigma_t \sqrt{2\ln(\pi^2 t^2/(3\beta))} ] \geq 1-\beta
    \end{equation*}
    and so we have proved $(\alpha, \beta)$-accuracy for $\cc$ where
    \begin{equation*}
        \alpha(t) = \sigma_t \sqrt{2\ln(\pi^2 t^2/(3\beta))}
        = O\!\left (C_{\eps/2, \delta/2}\log(t)\sqrt{\log(t/\beta)}\right)\,.
    \end{equation*}
    The final statement follows from using that $C_{\eps/2, \delta/2} = O(C_{\eps, \delta})$.
    This concludes the proof.
\end{proof}

Given access to an unbounded continual counter, we are able to construct the \emph{biased} continual counter.
We provide the constructive proof of \Cref{lem:cc-biased} next.
\begin{proof}[Proof of \Cref{lem:cc-biased}]
    Consider $\cc$ instantiated to be $(\eps, \delta/2)$-DP.
    On receiving an input stream $\sigma = (x_1, x_2, \dots)$, let its output be $\cc$ be $y_1', y_2', \dots$.
    We define the new continual counter $\bcc$, with output sequence $y_1, y_2, \dots$ via
    \begin{equation*}
        \forall t\in\mathbb{N} \;:\; y_t = \left\lceil \max\left\{  y_t' + \alpha_{\cc}(\eps, \delta/2, \eta, t) \;,\; \sum_{i=1}^t x_i\right\}\right\rceil\,.
    \end{equation*} 
    By the definition of $\bcc$, its output is integral, and it cannot underestimate the true count.
    \paragraph{Privacy.} 
    Define another continual counter $\cM$ that produces outputs $\tilde{y}_1, \tilde{y}_2, \dots$ where $\tilde{y}_t = \lceil y_t' + \alpha_{\cc}(\eps, \delta/2, \eta, t) \rceil$.
    $\cM$ is $(\eps, \delta/2)$-DP, as it is a data-independent postprocessing of $\cc$.
    From the accuracy of $\cc$, we have that
    \begin{equation*}
        \Pr\left[\forall t\in\N \;:\; \left\lvert y_t' - \sum_{i=1}^t x_i\right\rvert \leq \alpha_{\cc}(\eps, \delta/2, \eta, t)\right] \geq 1-\eta\,,
    \end{equation*}
    and so, for any input stream $\sigma$,
    \begin{equation*}
        \Pr\left[\bcc(\sigma) \neq \cM(\sigma)\right] \leq \eta.
    \end{equation*}
    To prove privacy for $\bcc$, fix $T\in\N$ and let $\sigma, \sigma'$ be 1-step 1-neighboring inputs of length $T$ and consider any measurable subset $O$ of length-$T$ output sequences.
    We write
    \begin{align}
        \Pr\left[\bcc(\sigma) \in O \right]
        &\leq \Pr\left[\cM(\sigma) \in O \right] + \Pr\left[\bcc(\sigma)\neq \cM(\sigma)\right]\notag\\
        &\leq \Pr\left[\cM(\sigma) \in O \right] + \eta.\label{eq:cc-biased1}
    \end{align}
    Additionally, from $\cM$ being $(\eps, \delta/2)$-DP,
    \begin{align}
        \Pr\left[\cM(\sigma) \in O \right]
        &\leq e^{\eps} \Pr\left[\cM(\sigma') \in O \right] + \delta/2\notag\\
        &\leq e^{\eps} \Pr\left[\bcc(\sigma') \in O \right] + e^{\eps}\eta + \delta/2.\label{eq:cc-biased2}
    \end{align}
    Combining \eqref{eq:cc-biased1} and \eqref{eq:cc-biased2}, we arrive at
    \begin{align*}
        \Pr\left[\bcc(\sigma) \in O \right]
        &\leq e^{\eps} \Pr\left[\bcc(\sigma') \in O \right] + (1+e^{\eps})\eta + \delta/2\\
        &= e^{\eps} \Pr\left[\bcc(\sigma') \in O \right] + \delta\,,
    \end{align*}
    where the last step follows from $\eta = 0.5\delta/(1+e^{\eps})$.
    As the neighboring relation is symmetric, we are done.
    \paragraph{Accuracy.}
    For the accuracy, note that a mechanism that directly outputs $y_t' + \alpha_{\cc}(\eps, \delta/2, \eta, t)$ at time $t\in\mathbb{N}$, would be $(\alpha, \beta)$-accurate for any $\beta\in(0, 1)$ where
    \begin{equation*}
        \alpha(t) = \alpha_\cc(\eps, \delta/2, \beta, t) + \alpha_\cc(\eps, \delta/2, \eta, t).
    \end{equation*}
    The same guarantee extends if the algorithm were to output $\max\{y_t' + \alpha_{\eta}(t), \sum_{i=1}^t x_i\}$, as taking the maximum can only reduce the error.
    Finally, as $y_t$ additionally takes the ceiling of the maximum, it follows that $\bcc$ is $(\alpha_{\bcc}, \beta)$-accurate for any $\beta\in(0, 1)$ where
    \begin{equation*}
        \alpha_\bcc(t) = \alpha(t) + 1 = \alpha_\cc(\eps, \delta/2, \beta, t) + \alpha_\cc(\eps, \delta/2, \eta, t) + 1.\qedhere
    \end{equation*}
\end{proof}
\section{Missing Observations in Section~\ref{sec:ecc-privacy}}\label{app:missing-privacy}

\begin{observation*}[Restatement of Observation~\ref{obs:neighbor-mapping-property}]
    Let $T\in\N$. Define $\calC_T^+$ as in Definition~\ref{def:f}. Let $\sigma=(x_1,\dots,x_T)$ and $\sigma'=(x_1',\dots,x_T')$ be two edit-neighboring sequences in $\left([0,1]\cup \{\bot\}\right)^T$. Let $f_{\sigma\to\sigma'}$ be the corresponding mapping. Let $t_{1:k}\in\calC_T^+$, and define
    \[
    t_{1:k}' = f_{\sigma\to\sigma'}(t_{1:k}).
    \]
    Let $i,j\in[T]$ and $p,q\in[k]$ be the indices associated with the definition of $f_{\sigma\to\sigma'}(t_{1:k})$. Let the sequences $\intsum(\sigma, t_{1:k})$ and $d(t_{1:k})$ denote the input streams of $\cc$ and $\bcc$ when $\ecc$ executes on $\sigma$ and $\partt(\sigma)=t_{1:k}$, and let $\intsum(\sigma', t_{1:k}')$ and $d(t_{1:k}')$ denote the same sequences for $\sigma'$ and $t_{1:k}'$ (see Observation~\ref{obs:inputs-cc-bcc} and Notation~\ref{not:privacy}). Then, the following statements hold:
    \begin{itemize}
        \item The sequences $\intsum(\sigma, t_{1:k})$ and $\intsum(\sigma', t_{1:k}')$ are $4$-step $1$-neighbors (see Definition~\ref{def:k-step-delta-neighboring}).
        \item The sequences $d(t_{1:k})$ and $d(t_{1:k}')$ are $2$-step $1$-neighbors (see Definition~\ref{def:k-step-delta-neighboring}).
    \end{itemize}
\end{observation*}

\begin{proof}
    We prove the observation for the case where $\sigma$ is an insertion neighbor of $\sigma'$. The reverse case is symmetric. For $\ell\in [k]$, let $d_\ell=t_\ell-t_{\ell-1}$ and $d_\ell'= t_\ell'-t_{\ell-1}'$ denote the $\ell$-th entries of $d(t_{1:k})$, and let $\intsum_\ell=\sum_{t=t_{\ell-1}+1}^{t_\ell}x_t$ and $\intsum_\ell=\sum_{t=t_{\ell-1}'+1}^{t_\ell'}x_t'$ denote the $\ell$-th entries of $\intsum(\sigma, t_{1:k})$ and $\intsum(\sigma', t_{1:k}')$. We prove this observation in two cases: $p\geq q-2$ and $p<q-2$.

    \medskip\noindent\textbf{Case 1: $p \ge q-2$.}
    In this case, by Definition~\ref{def:f}, we have $t_{1:k}' = t_{1:k}$. Hence, $d(t_{1:k}) = d(t_{1:k}')$. Next, we compare the interval sums. For every $\ell\in \{1, \dots, p-1\}\cup \{q+1, \dots, k\}$, the corresponding interval lies entirely outside the affected region from index $i$ to $j$. Thus $x_t = x_t'$ for all indices in the interval, and therefore $\intsum_\ell=\intsum_\ell'$. It remains to consider $\ell \in \{p,\dots,q\}$. For such $\ell$, both sums are taken over the same index set $\{t_{\ell-1}+1, \dots, t_\ell\}$. By the edit-neighboring property, the sequences $(x_1, \dots, x_T)$ and $(x_1', \dots, x_T')$ differ only in the indices $\{i, \dots, j\}$, where the values are shifted by one position. A direct comparison between $\intsum_\ell$ and $\intsum_\ell'$ shows that all terms cancel except possibly at the two boundary indices, where the lower boundary is the maximum of $i$ and the starting index of the $\ell$-th interval and the upper boundary is the minimum of $j$ and the ending index of this interval. More precisely, we have $t_{\ell-1}=t_{\ell-1}'$; $t_{\ell}=t_\ell'$; $x_t=x_t'$ for every $t\in\{t_{\ell-1}+1, \dots, i-1\}\cup\{j+1, \dots, t_\ell\}$; $x_t=x_{t-1}'$ for every $t\in \left\{\max\{i, t_{\ell-1}+1\}+1, \dots, \min\{t_\ell, j\}\right\}$. Therefore,
    \begin{align*}
        |\intsum_\ell-\intsum_\ell'|&= \left|\sum_{t=t_{\ell-1}+1}^{t_\ell}x_t - \sum_{t=t_{\ell-1}'+1}^{t_\ell'}x_t'\right|\\
        &= \left|\sum_{t=t_{\ell-1}+1}^{t_\ell}(x_t-x_t')\right|\\
        &= |x_{\max\{i, t_{\ell-1}+1\}}-x_{\min\{t_\ell, j\}}'|\\
        &\leq 1.
    \end{align*}  
    In the last inequality we used the fact that all $x_t$ values lie in $[0,1]$.
    Since $p \ge q-2$, there are at most $3$ indices $\ell\in \{p, \dots, q\}$. Hence the sequences $\intsum(\sigma,t_{1:k})$ and $\intsum(\sigma',t_{1:k}')$ differ in at most a $3$ coordinates, each by at most $1$, and thus (since $3<4$) these sequences are $4$-step $1$-neighbors.

    \medskip\noindent\textbf{Case 2: $p < q-2$.}
    By Definition~\ref{def:f}, the checkpoints $t_\ell'=t_\ell$ for every $\ell\in \{1, \dots, p\}\cup\{q-1, \dots, k\}$ and $t_\ell'=t_\ell-1$ for every $\ell\in \{p+1, \dots, q-2\}$. Therefore, the checkpoint differences $d_\ell=t_\ell-t_{\ell-1}$ and $d_\ell'=t_\ell'-t_{\ell-1}'$ are identical for every $\ell\in [k]\setminus\{p+1, q-1\}$ and differ by $1$ at $p+1$ and $q-1$. Hence, the sequences $d(t_{1:k})$ and $d(t_{1:k}')$ are $2$-step $1$-neighbors.

    As in Case~1, we have $\intsum_\ell=\intsum_\ell'$ for every $\ell\in \{1, \dots, p-1\}\cup \{q+1, \dots, k\}$, and $|\intsum_\ell-\intsum_\ell'|\leq 1$ for every $\ell\in \{p, q\}$. For $\ell \in \{p+2,\dots,q-2\}$, the intervals shift by one position, and we have $t_{\ell-1}+1=t_{\ell-1}'+2$; $t_\ell=t_\ell'+1$; and $x_t=x_{t-1}'$ for every $t\in \left\{t_{\ell-1}+1, \dots, t_\ell\right\}$. Thus,
    \begin{align*}
        \intsum_\ell-\intsum_\ell'&= \sum_{t=t_{\ell-1}+1}^{t_\ell}x_t - \sum_{t=t_{\ell-1}'+1}^{t_\ell'}x_t' \\
        &= \sum_{t=t_{\ell-1}+1}^{t_\ell}x_t - \sum_{t=t_{\ell-1}}^{t_\ell-1}x_t' \\
        &= \sum_{t=t_{\ell-1}+1}^{t_\ell} (x_t-x_{t-1}')=0
    \end{align*}    
    Finally, we consider $\ell = q-1$ and $\ell = p+1$. Here, $t_{q-2}' = t_{q-2}-1$ while $t_{q-1}' = t_{q-1}$. Again, we have $x_t=x_{t-1}'$ for every $t\in \{t_{q-2}+1, \dots, t_{q-1}\}$, $x_t=x_{t-1}'$. Therefore,
    \begin{align*}
        |\intsum_{q-1}-\intsum_{q-1}'|&= \left|\sum_{t=t_{q-2}+1}^{t_{q-1}}x_t - \sum_{t=t_{q-2}'+1}^{t_{q-1}'}x_t'\right|\\
        &= \left|\sum_{t=t_{q-2}+1}^{t_{q-1}}x_t - \sum_{t=t_{q-2}}^{t_{q-1}}x_t' \right|\\
        &= |-x_{t_{q-2}}'|\leq 1,
    \end{align*}
    which shows the bound for $\ell = q-1$. For $\ell = p+1$ an analogous argument shows $|\intsum_{p+1} - \intsum_{p+1}'| \le 1$. As a result, the sequences $\intsum(\sigma,t_{1:k})$ and $\intsum(\sigma',t_{1:k}')$ are identical, except in the indices $p, p+1, q-1, q$, each by at most $1$, and hence are $4$-step $1$-neighbors.
\end{proof}

\begin{observation}\label{obs:svt-inputs-neighboring}
    Let $T\in\N$, and let $\sigma=(x_1,\dots,x_T)$ and $\sigma'=(x_1',\dots,x_T')$ be two edit-neighboring sequences in $([-1,1]\cup\{\bot\})^T$. Let $f_{\sigma\to\sigma'}:\calC^+\to\calC$ be the function defined in Definition~\ref{def:f}. Let $t_{1:k}\in \calC^+$ and $t_{k+1}=T+1$. Define 
    $$t_{1:k}'= f_{\sigma\to\sigma'}(t_{1:k}).$$ 
    Define indices $i, j\in [T]$ and $p,q\in [k+1]$ as in the definition of $f_{\sigma\to\sigma'}(t_{1:k})$. For every $\ell\in [k]$, define
    \begin{align*}
        S_\ell &:= \left(\sum_{w=t_{\ell-1}+1}^t x_w\right)_{t=t_{\ell-1}+1}^{t_{\ell}}\text{ and } \qquad S_\ell' := \left(\sum_{w=t_{\ell-1}'+1}^t x_w'\right)_{t=t_{\ell-1}'+1}^{t_{\ell}'},
    \end{align*}
    with $x_w=\bot$ and $x_w'=\bot$ treated as $0$. Moreover, define
    \begin{align*}
        S_{k+1} &:= \left(\sum_{w=t_k+1}^t x_w\right)_{t=t_k+1}^{T}\text{ and } \qquad S_{k+1}' := \left(\sum_{w=t_k'+1}^t x_w'\right)_{t=t_k'+1}^{T},
    \end{align*}
    with $\bot$ values treated as $0$.
    \begin{enumerate}[label=\arabic*., ref=(\arabic*)]
        \item \label{item:svt-inputs-equal}
        For every $\ell\in \{1, \dots, p-1\}\cup \{p+2, \dots, q-2\}\cup \{q+1, \dots, k\}$, we have $S_\ell=S_\ell'$.
        \item \label{item:svt-inputs-shift-neighbor} Assume $q>p+2$ and $\sigma$ is an insertion neighbor of $\sigma'$ at step $i$. Then $S_{p+1}$ is a $1$-shift $1$-neighbor of $S_{p+1}'$, and $S_{q-1}'$ is a $1$-shift $1$-neighbor of $S_{q-1}$.
        \item \label{item:svt-inputs-shift-neighbor-prime} Assume $q>p+2$ and $\sigma'$ is an insertion neighbor of $\sigma$ at step $i$. Then $S_{p+1}'$ is a $1$-shift $1$-neighbor of $S_{p+1}$, and $S_{q-1}$ is a $1$-shift $1$-neighbor of $S_{q-1}'$.
        \item \label{item:svt-inputs-all-step-neighbor-k-plus-1} The sequences $S_{k+1}$ and $S_{k+1}'$ are all-step $1$-neighbors.
        \item \label{item:svt-inputs-all-step-neighbor-p-and-q} For $\ell\in \{p, q\}$, the sequences $S_\ell$ and $S_\ell'$ are all-step $1$-neighbors. 
        \item \label{item:svt-inputs-all-step-neighbor-q-equal-p-plus-2} Assume $q=p+2$. Then, the sequences $S_{p+1}$ and $S_{p+1}'$ are all-step $1$-neighbors.       
    \end{enumerate}
\end{observation}

\begin{proof}
    {Throughout this proof, whenever we invoke the assumption that every entry of $\sigma$ and $\sigma'$ has magnitude at most $1$, we mean that each entry is either a real value in this range or the symbol $\bot$, which is treated as the value $0\in [-1, 1]$. 
    For $\ell\in[k]$, the sequence $S_\ell$ has length $t_\ell-t_{\ell-1}$. For $h\in[t_\ell-t_{\ell-1}]$, we denote the $h$-th entry of $S_\ell$ by $S_\ell[h]$, which satisfies
    \[
        S_\ell[h]=\sum_{w=t_{\ell-1}+1}^{t_{\ell-1}+h} x_w.
    \]
    Similarly, for $\ell\in[k]$, the sequence $S_\ell'$ has length $t_\ell'-t_{\ell-1}'$. For $h\in[t_\ell'-t_{\ell-1}']$, we denote the $h$-th entry of $S_\ell'$ by $S_\ell'[h]$, which satisfies
    \[
        S_\ell'[h]=\sum_{w=t_{\ell-1}'+1}^{t_{\ell-1}'+h} x_w'.
    \]
    The sequences $\sigma$ and $\sigma'$ are edit-neighboring sequences, meaning that one of them is an insertion neighbor of the other at step $i$. By Definition~\ref{def:queue-neighboring} and Definition~\ref{def:f}, we have:
    \begin{itemize}[left=0pt]
        \item Assume $\sigma$ is an insertion neighbor of $\sigma'$ at step $i$. Then $j$ is the smallest index in $\{i,\dots,T\}$ such that $x_j'=\bot$ (and $j=T+1$ if no such index exists). By Definition~\ref{def:queue-neighboring}, we have:
        \begin{itemize}
            \item For every $t\in \{1, \dots, i-1\}$, $x_t=x_t'$.
            \item For every $t\in \{i+1, \dots, \min\{j, T\}\}$, $x_t=x_{t-1}'$.
            \item For every $t\in \{j+1, \dots, T\}$, $x_t=x_t'$.
        \end{itemize}
        Moreover, by the definition of $f_{\sigma\to\sigma'}$, $p, q\in [k+1]$ are the indices satisfying $t_{p-1}<i\leq t_p$ and $t_{q-1}<j\leq t_q$, and we have
        \begin{align*}
            t_{1:k}' \!=\! 
            \begin{cases}
                (t_1,\dots,t_p, t_{p+1}\!-\!1,\dots,t_{q-2}\!-\!1, t_{q-1},\dots,t_k), & \text{if } p<q-2,\\[2mm]
                (t_1,\dots,t_k), & \text{o.w.}.
            \end{cases}
        \end{align*}
        \item Assume $\sigma'$ is an insertion neighbor of $\sigma$ at step $i$. Then $j$ is the smallest index in $\{i,\dots,T\}$ such that $x_j=\bot$ (and $j=T+1$ if no such index exists). By Definition~\ref{def:queue-neighboring}, we have:
        \begin{itemize}
            \item For every $t\in \{1, \dots, i-1\}$, $x_t'=x_t$.
            \item For every $t\in \{i+1, \dots, \min\{j, T\}\}$, $x_t'=x_{t-1}$.
            \item For every $t\in \{j+1, \dots, T\}$, $x_t'=x_t$.
        \end{itemize}
        Moreover, by the definition of $f_{\sigma\to\sigma'}$, $p, q\in [k+1]$ are the indices satisfying $t_{p-1}<i\leq t_p$ and $t_{q-1}<j\leq t_q$, and we have
        \begin{align*}
            t_{1:k}' \!=\! 
            \begin{cases}
                (t_1,\dots,t_p, t_{p+1}\!+\!1,\dots,t_{q-2}\!+\!1, t_{q-1},\dots,t_k), & \text{if } q>p,\\[2mm]
                (t_1,\dots,t_k), & \text{o.w.}
            \end{cases}
        \end{align*}
    \end{itemize}
    In the second case, the indices $p$ and $q$ indeed satisfy $t_{p-1}'<i\le t_p'$ and $t_{q-1}'<j\le t_q'$. Moreover, in this case we have
    \begin{align*}
        t_{1:k} \!=\! 
        \begin{cases}
            (t_1',\dots,t_{p}', t_{p+1}'\!-\!1,\dots,t_{q-2}'\!-\!1, t_{q-1}',\dots,t_k'), & \text{if } q>p,\\[2mm]
            (t_1',\dots,t_k'), & \text{o.w.}.
        \end{cases}
    \end{align*}
    With these observations, the relation between the sequences $\sigma$ and $\sigma'$ and the relation between the checkpoint sequences $t_{1:k}$ and $t_{1:k}'$ in the second case are identical to those in the first case, with the roles of $\sigma=(x_1,\dots,x_T)$ and $\sigma'=(x_1',\dots,x_T')$ as well as $t_{1:k}$ and $t_{1:k}'$ swapped.
 
    Throughout this proof, we only rely on these relations. Therefore, it suffices to prove the desired properties of $S_\ell$ and $S_\ell'$ under the assumption that $\sigma$ is an insertion neighbor of $\sigma'$. By symmetry, an identical argument applies when $\sigma'$ is an insertion neighbor of $\sigma$, with the roles of $S_\ell$ and $S_\ell'$ exchanged. Using this, we prove items~\ref{item:svt-inputs-equal},~\ref{item:svt-inputs-all-step-neighbor-k-plus-1},~\ref{item:svt-inputs-all-step-neighbor-p-and-q}, and~\ref{item:svt-inputs-all-step-neighbor-q-equal-p-plus-2} under the assumption that $\sigma$ is an insertion neighbor of $\sigma'$. We also conclude item~\ref{item:svt-inputs-shift-neighbor-prime} by symmetry from  item~\ref{item:svt-inputs-shift-neighbor}.
    }
    
    \medskip\noindent\textbf{\ref{item:svt-inputs-equal}}: {
    By the above discussion, we assume without loss of generality that $\sigma$ is an insertion neighbor of $\sigma'$ at step $i$. We must show that for every $\ell\in \{1, \dots, p-1\}\cup \{p+2, \dots, q-2\}\cup \{q+1, \dots, k\}$, the sequences $S_\ell$ and $S_\ell'$ have the same length and identical entries.
    \begin{itemize}[left=0pt]
    \item \emph{Case $\ell\in\{1,\dots,p-1\}$.}
    In this case, $t_{\ell-1}=t_{\ell-1}'$ and $t_\ell=t_\ell'$, so both $S_\ell$ and $S_\ell'$ have length $t_\ell-t_{\ell-1}$. Since $i>t_{p-1}\ge t_\ell$, the additionally inserted element $x_i$ occurs strictly after all indices contributing to $S_\ell$. Hence,
    \[
        x_t=x_t' \qquad \text{for all } t\in\{t_{\ell-1}+1,\dots,t_\ell\}.
    \]
    It follows directly that $S_\ell=S_\ell'$.
    \item \emph{Case $\ell\in \{p+2, \dots, q-2\}$.}  
    Here, $t_{\ell-1}'=t_{\ell-1}-1$ and $t_\ell'=t_\ell-1$, so both sequences again have length $t_\ell-t_{\ell-1}$. Moreover, since $i\le t_p\le t_{\ell-1}$ and $j>t_{q-1}\ge t_\ell$, the shift affects all indices in the range $\{t_{\ell-1}+1,\dots,t_\ell\}$. In particular,
    \[
        x_t=x_{t-1}' \qquad \text{for all } t\in\{t_{\ell-1}+1,\dots,t_\ell\}.
    \]
    Therefore, for every $h\in[t_\ell-t_{\ell-1}]$,
    \begin{align*}
        S_\ell[h]
        = \sum_{w=t_{\ell-1}+1}^{t_{\ell-1}+h} x_w
        = \sum_{w=t_{\ell-1}+1}^{t_{\ell-1}+h} x_{w-1}'
        = \sum_{w=t_{\ell-1}}^{t_{\ell-1}+1+h} x_{w}'
        = \sum_{w=t_{\ell-1}'+1}^{t_{\ell-1}'+h} x_w'
        = S_\ell'[h].
    \end{align*}
    Hence, $S_\ell=S_\ell'$.
    \item \emph{Case $\ell\in\{q+1,\dots,k\}$.} In this case, $t_{\ell-1}=t_{\ell-1}'$ and $t_\ell=t_\ell'$, so the two sequences have equal length. Since $j\le t_q\le t_{\ell-1}$, the shifting finishes strictly before the indices contributing to $S_\ell$, and thus
    \[
        x_t=x_t' \qquad \text{for all } \quad t\in\{t_{\ell-1}+1,\dots,t_\ell\}.
    \]
    Consequently, $S_\ell=S_\ell'$.
    \end{itemize}
    }
    
    \medskip\noindent\textbf{\ref{item:svt-inputs-shift-neighbor}}: {
    We have $t_{p}=t_{p}$ and $t_{p+1}=t_{p+1}'+1$. Moreover, since $p+2<q\leq k+1$, we have $p+1\leq k$. Therefore, the sequence $S_{p+1}$ has length $t_{p+1}-t_{p}$, while $S_{p+1}'$ has length $t_{p+1}-t_{p}-1$. We must show that $S_{p+1}$ is a $1$-shift $1$-neighbor of $S_{p+1}'$, i.e., $S_{p+1}[1]\leq 1$ and $|S_{p+1}[h]-S_{p+1}[h-1]|\leq 1$ for every $h\in [t_{p+1}-t_{p}]/\{1\}$. The first condition holds because
    $$S_{p+1}[1] = |x_{t_{p}+1}|\leq 1.$$
    Fix $h\in\{2,\dots,t_{p+1}-t_{p}\}$. By definition, $i\leq t_p$ and $t_{p+1}<t_{q-1}< j$. Thus we have $x_t=x_{t-1}'$ for every $t\in \{t_p+1, \dots, t_{p+1}\}$. Therefore, for every $h\in[t_{p+1}-t_{p}]$, 
    \begin{align*}
        |S_{p+1}[h]-S_{p+1}'[h-1]|
        &= \left|\sum_{t=t_{p}+1}^{t_{p}+h}x_t-\sum_{t=t_{p}'+1}^{t_{p}'+h-1}x_t'\right| = \left|\sum_{t=t_{p}+1}^{t_{p}+h}x_t-\sum_{t=t_{p}+1}^{t_{p}+h-1}x_{t+1}\right|\\
        &= |x_{t_p+1}|\leq 1.
    \end{align*}
    Thus, $S_{p+1}$ is a $1$-shift $1$-neighbor of $S_{p+1}'$. The proof that $S_{q-1}'$ is a $1$-shift $1$-neighbor of $S_{q-1}$ follows by an identical argument.
    }
         
    \medskip\noindent\textbf{\ref{item:svt-inputs-shift-neighbor-prime}:} {
    This item holds immediately by symmetry.}

    \medskip\noindent\textbf{\ref{item:svt-inputs-all-step-neighbor-k-plus-1}:} {
    As before, we assume without loss of generality that $\sigma$ is an insertion neighbor of $\sigma'$ at step $i$. We have $t_\ell=t_\ell'$ for every $\ell\geq q-1$. Thus, since $q\leq k+1$, we have $t_k=t_k'$, and both sequences $S_{k+1}$ and $S_{k+1}'$ have length $T-t_k$. We must show that $|S_{k+1}[h]-S_{k+1}'[h]|\leq 1$ for every $h\in[T-t_k]$. 
    
    Fix $h\in [T-t_k]$. Both $p$ and $q$ can be equal to $k+1$, thus indices $i$ and $j$ could potentially lie in the set $\{t_k+1, \dots, T\}$. We consider the following three cases.
    \begin{itemize}[left=0pt]
        \item \emph{Case 1: $t_k+h < i$.} In this case, $x_t=x_t'$ for all $t\in\{t_k+1, \dots, t_k+h\}$, and thus $S_{k+1}[h]=S_{k+1}'[h]$.
        \item \emph{Case 2: $i \le t_k+h < j$.} In this case, $x_t=x_t'$ for all $t\in\{t_k+1, \dots, i-1\}$ and $x_t=x_{t-1}'$ for all $t\in\{i+1, \dots, t_k+h\}$. Hence,
        \begin{align*}
            |S_{k+1}[h]-S_{k+1}'[h]|&= \left|\sum_{t=t_k+1}^{t_k+h}x_t-\sum_{t=t_k'+1}^{t_k'+h-1}x_t'\right|\\ 
            &= \left|\sum_{t=t_k+1}^{t_k+h}x_t-\sum_{t=t_k+1}^{i-1}x_t - \sum_{t=i}^{t_k+h-1}x_{t+1} - x_{t_k+h}'\right|\\
            &= |x_i-x_{t_k+h}'| \leq 1,
        \end{align*}
        where we use the fact that $x_i, x_{t_k+h}'\in [0,1]$.
        \item \emph{Case 3: $t_k+h \geq j$.} In this case, $x_t=x_t'$ for all $t\in\{t_k+1, \dots, i-1, j+1, \dots, t_k+h\}$ and $x_t=x_{t-1}'$ for all $t\in\{i+1, \dots, t_k+j\}$. Therefore,
        \begin{align*}
            |S_{k+1}[h]-S_{k+1}'[h]|
            &= \left|\sum_{t=t_k+1}^{t_k+h}x_t-\sum_{t=t_k'+1}^{t_k'+h-1}x_t'\right|\\
            &= \left|\sum_{t=t_k+1}^{t_k+h}x_t-\sum_{t=t_k+1}^{i-1}x_t - \sum_{t=i}^{j-1}x_{t+1} - x_j' - \sum_{t=j+1}^{t_k+h}x_{t}\right|\\
            &= |x_i|\leq 1,
        \end{align*}
        where we use the facts that $x_i\in [0,1]$ and $x_j'$ is considered as $0$.
    \end{itemize}
    }
    
    \medskip\noindent\textbf{\ref{item:svt-inputs-all-step-neighbor-p-and-q}:} {
    As before, we assume without loss of generality that $\sigma$ is an insertion neighbor of $\sigma'$ at step $i$. If $p=k+1$, then the proof follows directly from item~\ref{item:svt-inputs-all-step-neighbor-k-plus-1}. Assume $p\leq k$. Then both sequences $S_p$ and $S_p'$ have length $t_p - t_{p-1}$. We must show that $S_p$ and $S_p'$ are all-step $1$-neighbors, i.e., $|S_p[h]-S_p'[h]|\leq 1$ for every $h\in[t_p-t_{p-1}]$. 
    
    Fix $h\in[t_p-t_{p-1}]$. By definition, $i$ lies in the range $\{t_{p-1}+1, \dots, t_p\}$. As $p$ and $q$ could be equal, $j$ could also lie in this range. We consider the following three cases. (Case three does not happen if $q>p$.)
    \begin{itemize}
        \item \emph{Case 1: $t_{p-1}+h < i$.} In this case, $x_t=x_t'$ for all $t\in\{t_{p-1}+1, \dots, t_{p-1}+h\}$, and thus $S_p[h]=S_p'[h]$.
        \item \emph{Case 2: $i \le t_{p-1}+h < j$.} In this case, $x_t=x_t'$ for all $t\in\{t_{p-1}+1, \dots, i-1\}$ and $x_t=x_{t-1}'$ for all $t\in\{i+1, \dots, t_{p-1}+h\}$. Hence,
        \begin{align*}
            |S_p[h]-S_p'[h]|&= \left|\sum_{t=t_{p-1}+1}^{t_{p-1}+h}x_t-\sum_{t=t_{p-1}'+1}^{t_{p-1}'+h-1}x_t'\right|\\ 
            &= \left|\sum_{t=t_{p-1}+1}^{t_{p-1}+h}x_t-\sum_{t=t_{p-1}+1}^{i-1}x_t - \sum_{t=i}^{t_{p-1}+h-1}x_{t+1} - x_{t_{p-1}+h}'\right|\\
            &= |x_i-x_{t_{p-1}+h}'| \leq 1,
        \end{align*}
        where we use the fact that $x_i, x_{t_{p-1}+h}'\in [0,1]$.
        \item \emph{Case 3: $t_{p-1}+h \geq j$.} In this case, $x_t=x_t'$ for all $t\in\{t_{p-1}+1, \dots, i-1, j+1, \dots, t_{p-1}+h\}$ and $x_t=x_{t-1}'$ for all $t\in\{i+1, \dots, t_{p-1}+j\}$. Therefore,
        \begin{align*}
            |S_p[h]-S_p'[h]|&= \left|\sum_{t=t_{p-1}+1}^{t_{p-1}+h}x_t-\sum_{t=t_{p-1}'+1}^{t_{p-1}'+h-1}x_t'\right|\\
            &= \left|\sum_{t=t_{p-1}+1}^{t_{p-1}+h}x_t-\sum_{t=t_{p-1}+1}^{i-1}x_t - \sum_{t=i}^{j-1}x_{t+1} - x_j' - \sum_{t=j+1}^{t_{p-1}+h}x_{t}\right|\\
            &\leq |x_i|\leq 1,
        \end{align*}
        where we use the facts that $x_i\in [0,1]$ and $x_j'$ is considered as $0$.
    \end{itemize}
    Thus, $S_p$ and $S_p'$ are all-step $1$-neighbors. 
    
    It remains to show $S_q$ and $S_{q}'$ are all-step $1$-neighbors. If $p=q$, then the statement holds as shown for $p$. Assume $q>p$. Both sequences $S_q$ and $S_q'$ have length $t_q - t_{q-1}$. We must show that $|S_q[h]-S_q'[h]|\leq 1$ for every $h\in[t_q-t_{q-1}]$. 
    
    Fix $h\in [t_q-t_{q-1}]$. Since $q>p$, we have $i\leq t_{q-1}$, and also by definition, $t_{q-1}< j\leq t_q$. There are two cases:
    \begin{itemize}[left=0pt]
        \item \emph{Case 1: $t_{q-1}+h \leq j$.} In this case, $x_t=x_{t-1}'$ for all $t\in\{t_{q-1}+1 \dots, t_{q-1}+h\}$. Hence,
        \begin{align*}
            |S_q[h]-S_q'[h]|
            &= \left|\sum_{t=t_{q-1}+1}^{t_{q-1}+h}x_t-\sum_{t=t_{q-1}'+1}^{t_{q-1}'+h-1}x_t'\right| = \left|\sum_{t=t_{q-1}+1}^{t_{q-1}+h}x_t-\sum_{t=t_{q-1}+1}^{t_{q-1}+h-1}x_{t+1} - x_{t_{q}+h}'\right|\\
            &\leq |x_{t_{q-1}+1}-x_{t_{q}+h}'| \leq 1,
        \end{align*}
        where we use the fact that $x_{t_{q-1}+1},x_{t_{q}+h}'\in [0,1]$.
        \item \emph{Case 2: $t_{p-1}+h > j$.} In this case, $x_t=x_{t-1}'$ for all $t\in\{t_{q-1}+1 \dots, j\}$ and $x_t=x_t'$ for all $t\in\{i+1, \dots, t_{q-1}+h\}$. Therefore,
        \begin{align*}
            |S_q[h]-S_q'[h]|
            &= \left|\sum_{t=t_{q-1}+1}^{t_{q-1}+h}x_t-\sum_{t=t_{q-1}'+1}^{t_{q-1}'+h-1}x_t'\right|\\
            &= \left|\sum_{t=t_{q-1}+1}^{t_{q-1}+h}x_t-\sum_{t=t_{q-1}+1}^{j-1}x_{t+1} - x_j' - \sum_{t=j+1}^{t_{q-1}+h}x_{t}\right|\\
            &\leq |x_{t_{q-1}+1}|\leq 1,
        \end{align*}
        where we use the facts that $x_{t_{q-1}+1}\in [0,1]$ and $x_j'$ is considered as $0$.
    \end{itemize}
    }

    \medskip\noindent\textbf{\ref{item:svt-inputs-all-step-neighbor-q-equal-p-plus-2}:} {
    As before, we assume without loss of generality that $\sigma$ is an insertion neighbor of $\sigma'$ at step $i$. Since $p=q-2$, by definition, we have $t_{p}=t_{p}'$ and $t_{p+1}=t_{p+1}'$. Moreover, since $p+1= q-1$ and $q\leq k+1$, we have $p+1\leq k$. Consequently, both sequences $S_{p+1}$ and $S_{p+1}'$ have length $t_{p+1}-t_{p}$. We must show that $S_{p+1}$ and $S_{p+1}'$ are all-step $1$-neighbors, i.e., for every $h\in[t_{p+1}-t_{p}]$, $|S_{p+1}[h]-S_{p+1}'[h]|\leq 1$. 
    
    Fix $h\in[t_{p+1}-t_p]$. By definition, $i\leq t_p$ and $j> t_{q-1}=t_{p+1}$. Thus, $x_t=x_{t-1}'$ for all $t\in\{t_p+1, \dots, t_{p}+h\}$. Hence,
        \begin{align*}
            |S_{p+1}[h]-S_{p+1}'[h]|&= \left|\sum_{t=t_p+1}^{t_p+h}x_t-\sum_{t=t_p'+1}^{t_p'+h-1}x_t'\right|\\
            &= \left|x_{t_p+1}+\sum_{t=t_p+2}^{t_p+h}x_t-\sum_{t=t_p+1}^{t_p+h-1}x_{t+1} - x_{t_p+h}'\right|\\
            &\leq |x_{t_p+1}-x_{t_p+h}'| \leq 1,
        \end{align*}
        where we use the fact that $x_{t_p+1},x_{t_p+h}'\in [0,1]$.
    }

\end{proof}

\section{Related Work on Continual Counting}\label{se:related}
Throughout the related work, discussions of error pertain to the $\ell_\infty$ error with constant probability on streams of known length, unless specified otherwise.
This corresponds to $\alpha$, for e.g., for an $(\alpha, 1/3)$-accurate mechanism.

\paragraph{Standard Continual Counting.}
The concurrent works \cite{dwork2010differentially} and \cite{chan2011private} initiated the study of continual counting for pure DP.
Their celebrated binary tree mechanism achieves an error of $O_{\eps}(\log^2 T)$, and \cite{dwork2010differentially} proved a lower bound of $\Omega_{\eps}(\log T)$.
\cite{jain2012differentially} adapted the binary tree mechanism to approximate DP, achieving an error of $O_{\eps, \delta}(\log^{3/2} T)$.
While the lower bound in \cite{dwork2010differentially} can be extended to $(\eps, \delta)$-DP for polynomially small $\delta$, the $\ell_2$ lower bound by \cite{henzinger2023almost} implies an $\ell_\infty$ error of $\Omega_{\eps, \delta}(\log T)$ for constant $\delta$.
Besides asymptotics, there is a rich literature on optimizing the constants in these upper bounds, see e.g., \cite{honaker2015efficient,kairouz2021practical,fichtenberger2021differentially,denisov2022improved,andersson2023smooth,henzinger2023almost,dvijotham2024efficient,henzinger2024unifying,andersson25binning,henzingerU2025,henzinger25normalized}.
Common across all these constructions is that they are based on adding correlated noise, drawn independently of the input, to the counts, and are all variants of the \emph{(matrix) factorization mechanism} introduced by~\cite{li2015matrix}.

If the error of the mechanism is not measured over worst-case inputs, but instead is allowed to adapt to the \emph{sparsity} of the input, then there are additional results.
Defining $n := \lvert\{ x_t \neq  0 : t\in[n] \}\rvert$ as the sparsity of an input $\sigma$, \cite{dwork2015pure} gave an accuracy bound of $O_\eps(\log T + \log^2 n)$ for pure DP.
Their key idea was the \emph{partitioning mechanism}, which decomposes the stream into $O_\eps(n)$ blocks, each containing no more than $O_\eps(\log T)$ updates, and then running a continual counter on the blocks.
The same idea can be applied to $(\eps, \delta)$-DP, yielding an error of $O_{\eps, \delta}(\log T + \log^{3/2} n)$. The work by \cite{DBLP:conf/aistats/HenzingerSS25} extended this to the histogram setting, i.e.~to higher-dimensional inputs and outputs and~\cite{DBLP:journals/corr/abs-2209-01387} generalizes the technique to all linear queries over a set of elements that is modified at each time step by the insertion or deletion of one item.
Finally, recent work by \cite{cohen24lower} showed a matching lower bound of $\Omega_{\eps, \delta}(n)$ for very sparse streams where $n = O(\log T)$.

\paragraph{Continual Counting under Non-Standard Neighborings}
To the best of our knowledge, all past work on continual counting that uses a different neighboring relation fall into one of two categories.
The first category derives from implementing user-level DP under the standard neighboring relation, allowing a user to contribute multiple entries in $\sigma$.
If any user is allowed to contribute up to $k\in [T]$ entries in a stream $\sigma$, then $\sigma\sim_k\sigma'$ if at most $k$ of the inputs differ.
This is for example a motivating example in DP machine learning \citep{kairouz2021practical}, where continual counting is used as a primitive (privately summing gradients), and we would like to support users contributing more than one example during training.\footnote{See the recent monograph by \cite{pillutla2025correlated} for an in-depth treatment of this application.}
To achieve better utility than what a direct application of group privacy would yield, it has become common to enforce more structure in user participation patterns.
In particular, \cite{choquette2022multi,choquette2023amplified} introduced the notion of \emph{$b$-min-separated $k$-repeated participation}, where $\sigma\sim_{(k, b)}\sigma'$ if the two streams differ in at most $k$ positions, and these positions are separated by at least $b$ steps.
There has been a considerable amount of work focusing on optimizing factorization-based mechanisms in this setting~\cite{kalinin24banded,mcmahan2024hassle,mcKenna25scaling,kalinin2026back,kalinin26learning}.
We also highlight the work of \cite{dong23continual} on user-level DP. Instead of restricting participation patterns, they dynamically track user participation to give a down-neighborhood optimal algorithm (up to polylogarithmic factors) for $\eps$-DP.

The second class of works is based on reductions from \emph{continual cardinality estimation} problems to continual counting.
Here the input stream $\sigma$ contains updates (additions and/or deletions of items), and we maintain a set of all present items, $S_t$, over time.
At each time step $t$, we are to output a statistic $f(S_t)$, e.g., the size of $S_t$.
Rather than designing algorithms for directly releasing $f$, \cite{song2018differentially} proposed solving the continual counting problem on $d=(f(S_1), f(S_2) - f(S_1), \dots, f(S_T) - f(S_{T-1}))$.
Depending on the problem, and the neighboring relation $\sigma\sim_{c}\sigma$, the induced neighboring relation $d\sim_r d'$ for the corresponding continual counting problem may either overlap with the standard notion, or deviate significantly.
E.g., releasing the number of edges or connected components under edge-insertions with edge-level DP cleanly reduces to regular continual counting \cite{song2018differentially,fichtenberger2021differentially}.
On the other hand, counting distinct elements in the turnstile model, even under event-level DP, requires error $\Omega_{\eps, \delta}(T^{1/4})$ for worst-case inputs~\cite{jain23distinct}.
This last result implies that continual counting under the corresponding induced neighboring $\sim_r$ suffers the same polynomial lower bound.
For additional works that leverage this reduction technique, see e.g., \cite{epasto2023differentially,fichtenberger2022constant,jain24time,henzinger24distinct,raskhodnikova25dynamic,cummings25space,andersson26improved} and references therein.

Our work on the edit-neighboring relation, $\sim_e$, is distinct from past works on continual counting.
Firstly, we are not aware of any past work that defines neighboring inputs in terms of \emph{shifts}, and, to the best of our knowledge, no neighboring notion has (semantically) treated \enquote{$\bot$} differently from \enquote{0}.
Additionally, we are not aware of any existing algorithm for continual counting, under any neighboring relation, that both (1) is not based on adding input-independent correlated noise to the true counts, \emph{and} (2) targets worst-case $\ell_\infty$ error.
For example, while the algorithm by Dwork, Naor, Reingold and Rothblum~\cite{dwork2015pure} based on private partitions for swap-neighboring streams is outside the class of input-independent additive noise algorithms, and allows for improved error on sparse streams, algorithms within that class still achieve a lower error on dense streams.
By contrast, we give an algorithm for continual counting under $\sim_e$-neighboring achieving polylogarithmic error, and additionally show that any algorithm based on input-independent additive noise has to incur an exponentially larger error.

\section{Continual Counting under Extended Edit Neighboring}
\label{sec:modified-insertion-neighboring} 

In this section, we study continual counting under an extended version of edit neighboring that allows streams of different lengths. In Definition~\ref{def:queue-neighboring}, an insertion into a stream without $\bot$ was defined by adding an element at position $i$ and discarding the final element to preserve the length. In contrast, in the extended definition, the final element is not removed, and consequently neighboring streams may have different lengths. This introduces an additional challenge for differential privacy: a mechanism must avoid revealing the exact length of the input stream. The goal of this extension is to process all inputs without discarding any element. To achieve this, the mechanism cannot simply run for a fixed number of steps; instead, it must continue processing until it receives a special symbol $\$$ denoting the end of the input stream, at which point it halts.

\begin{definition}[Extended Edit Neighbors]\label{def:modified-insertion-neighboring}
    Let $T, T'\in\N\cup\{0\}$, and let $\sigma=(x_1,\dots,x_T,\$)$ and $\sigma'=(x_1',\dots,x_{T'}',\$)$ be two sequences in $([0,1]\cup\{\bot\})^*\times \{\$\}$. We say that $\sigma$ is an \emph{extended insertion neighbor of $\sigma'$ at step} $i\in\{1, \dots, T\}$ if the following conditions hold:
    \begin{itemize}
        \item For all $t< i$, we have $x_t=x_t'$.
        \item Let $j$ be the smallest index in $\{i,\dots,T\}$ such that $x_{j}'=\bot$. If no such index exists, set $j=T+1$.
        \begin{itemize}
            \item If $j\leq T$, then the sequences $\sigma$ and $\sigma'$ are of the same length, i.e., $T=T'$; for every $t\in \{i+1, \dots, j\}$, we have $x_t=x_{t-1}'$; and for every $t\in\{j+1,\dots,T\}$, we have $x_t=x_t'$.
            \item If $j=T+1$, then $T=T'+1$, and $x_t = x_{t-1}'$ for all $t \in \{i+1, \dots, T+1\}$.
        \end{itemize}
    \end{itemize}
    We say that $\sigma$ and $\sigma'$ are \emph{length-changing edit neighbors} if either $\sigma$ is a length-changing insertion neighbor of $\sigma'$ at some step $i\in[T]$, or $\sigma'$ is a length-changing insertion neighbor of $\sigma$ at some step $i\in[T']$.
\end{definition}

We now define a wrapper mechanism $\wrp$ that transforms any continual counter that is private with respect to the edit neighboring in Definition~\ref{def:queue-neighboring} (e.g., $\ecc$ from Section~\ref{sec:ecc}) into one that is private with respect to the extended edit neighboring in Definition~\ref{def:modified-insertion-neighboring}.

\medskip
\noindent\textbf{Mechanism $\wrp$.}
At initialization, the mechanism $\wrp$ receives privacy parameters $\eps$ and $\delta$. It first samples $U \sim \Lap(2/\eps)$ and sets
$$C = \max\left\{0, 1 + \frac{2\ln(1/\delta)}{\eps} + \lceil U \rceil\right\}.$$
It then initializes an instance of $\ecc$ with parameters $\eps/2$ and $\delta/2$.

At each time step $t$, $\wrp$ takes one of the following actions:
\begin{itemize}
    \item If an input $x_t \in [0,1] \cup \{\bot\}$ arrives, it feeds $x_t$ to $\ecc$ and returns its output.
    \item If the end-of-stream symbol $\$$ arrives, $\wrp$ generates $C$ additional inputs over the next $C$ steps as follows: while $C > 0$, it decrements $C$, feeds $\bot$ to $\ecc$, and outputs the result. Once $C = 0$, $\wrp$ halts permanently.
\end{itemize}
By construction, given an input stream $\sigma=(x_1,\dots,x_T,\$)$, the mechanism $\wrp$ produces at least $T$ outputs.

\begin{lemma}\label{lem:modified-insertion-neighboring-privacy}
    Let $\wrp$ be the wrapper mechanism described above with privacy parameters $\eps>0$ and $0<\delta\leq 1$. Then, $\wrp$ is $(\eps, \delta)$-DP with respect to the notion of extended neighbors defined in Definition~\ref{def:modified-insertion-neighboring}.
\end{lemma}

\begin{proof}
    Let $T,T'\in\N\cup\{0\}$, and let $\sigma=(x_1,\dots,x_T,\$)$ and $\sigma'=(x_1',\dots,x_{T'}',\$)$ be two extended edit-neighboring sequences in $([0,1]\cup\{\bot\})^*\times \{\$\}$. We must show that the output distributions of $\wrp$ on $\sigma$ and $\sigma'$, denoted by $\wrp(\sigma)$ and $\wrp(\sigma')$, are $(\eps, \delta)$-indistinguishable. 
    
    Without loss of generality, we assume that $\sigma$ is an extended insertion neighbor of $\sigma'$ at step $i\in [T]$. We consider two cases: $T=T'$ and $T=T'+1$.
    
    \medskip\noindent\textbf{Case $T=T'$.} Recall the random variable $U \sim \Lap(2/\eps)$ and the padding length $C=\max\left\{0, 1 + 2\ln(1/\delta)/\eps + \lceil U \rceil\right\}$ in the definition of $\wrp$. The distribution of the padding length $c$ is identical whether $\wrp$ is executed on $\sigma$ or $\sigma'$. Therefore, to show $\wrp(\sigma)$ and $\wrp(\sigma')$ are $(\eps, \delta)$-indistinguishable, it suffices to show that for every fixed padding length $c$, the conditional output distributions are $(\eps, \delta)$-indistinguishable.

    Fix $c\in \N\cup\{0\}$. Conditioned on the padding length being $c$, the outputs $\wrp(\sigma)$ and $\wrp(\sigma')$ equal the output streams of $\ecc$ on the padded inputs $\sigma\cdot \bot^c$ and $\sigma'\cdot \bot^c$, respectively. We now show that these padded streams are edit neighbors under Definition~\ref{def:queue-neighboring}.

    Since $T=T'$, by Definition~\ref{def:modified-insertion-neighboring}, there exists $j\in [T]$ such that: (1) $x_j'=\bot$; (2) $x_t = x_{t-1}'$ for every $t\in \{i+1, \dots, j\}$; and (3) $x_t = x_t'$ for every $t\in \{j+1, \dots, T\}$. Comparing this with Definition~\ref{def:queue-neighboring} implies that the padded sequences $\sigma\cdot \bot^c$ and $\sigma'\cdot \bot^c$ are edit neighbors. Therefore, by Theorem~\ref{thm:ecc-privacy} and the choice of privacy parameters in $\wrp$, the output distributions of $\ecc$ on $\sigma\cdot \bot^c$ and $\sigma'\cdot \bot^c$ are $(\eps/2, \delta/2)$-DP (and consequently, $(\eps, \delta)$-DP) with respect to edit neighboring, completing the proof for this case.
    
    \medskip\noindent\textbf{Case $T=T'+1$.}
    We must show that for every measurable set $\calY\subseteq\R^*$,
    \begin{align*}
        \Pr[\wrp(\sigma)\in \calY] \leq e^{\eps}\Pr[\wrp (\sigma')\in \calY] + \delta,
    \end{align*}
    and 
    \begin{align*}
        \Pr[\wrp(\sigma')\in \calY] \leq e^{\eps}\Pr[\wrp(\sigma)\in \calY] + \delta.
    \end{align*}
    By Fact~\ref{fact:sum-conditions}, the above inequalities are equivalent to
    \begin{equation}\label{eq:privacy-insertion-neighboring-t+1-to-t}
    \begin{aligned}
        &\sum_{z\in \Z}\Pr\!\big[\lceil U\rceil = z\big]\cdot \Pr[\wrp(\sigma)\in \calY \mid \lceil U\rceil = z]\\
        &\quad\leq e^{\eps}\sum_{z\in \Z}\Pr\!\big[\lceil U\rceil = z\big]\cdot\Pr[\wrp(\sigma')\in \calY \mid \lceil U\rceil = z] + \delta,
    \end{aligned}
    \end{equation}
    and 
    \begin{equation}\label{eq:privacy-insertion-neighboring-t-to-t+1}
    \begin{aligned}
        &\sum_{z\in \Z}\Pr\!\big[\lceil U\rceil = z\big]\cdot \Pr[\wrp(\sigma')\in \calY \mid \lceil U\rceil = z] \\
        &\quad\leq e^{\eps}\sum_{z\in \Z}\Pr\!\big[\lceil U\rceil = z\big]\cdot\Pr[\wrp(\sigma)\in \calY \mid \lceil U\rceil = z] + \delta,
    \end{aligned}
    \end{equation}
    By Lemma~\ref{lem:laplace-concentration-bound},
    \begin{align*}
        \Pr[\lceil U \rceil \geq - \frac{2\ln(1/\delta)}{\eps}]\geq 1-\frac{1}{2}e^{-\frac{\eps}{2}\cdot \frac{2\ln(1/\delta)}{\eps}} = 1-\frac{\delta}{2},
    \end{align*}
    which combined with the fact that probabilities are non-negative implies
    \begin{align*}
        &\sum_{z\geq - 2\ln(1/\delta)/\eps - 1}\Pr\!\big[\lceil U\rceil = z\big]\cdot\Pr[\wrp(\sigma)\in \calY \mid \lceil U\rceil = z]\\
        &\leq \sum_{z\in \Z}\Pr\!\big[\lceil U\rceil = z\big]\cdot \Pr[\wrp(\sigma)\in \calY \mid \lceil U\rceil = z]\\
        &\leq \sum_{z\geq - 2\ln(1/\delta)/\eps}\Pr\!\big[\lceil U\rceil = z\big]\cdot\Pr[\wrp(\sigma)\in \calY \mid \lceil U\rceil = z] +\frac{\delta}{2},
    \end{align*}
    and
    \begin{align*}
        &\sum_{z\geq - 2\ln(1/\delta)/\eps + 1}\Pr\!\big[\lceil U\rceil = z\big]\cdot\Pr[\wrp(\sigma')\in \calY \mid \lceil U\rceil = z]\\
        &\leq \sum_{z\in \Z}\Pr\!\big[\lceil U\rceil = z\big]\cdot \Pr[\wrp(\sigma')\in \calY \mid \lceil U\rceil = z]\\
        &\leq \sum_{z\geq - 2\ln(1/\delta)/\eps}\Pr\!\big[\lceil U\rceil = z\big]\cdot\Pr[\wrp(\sigma')\in \calY \mid \lceil U\rceil = z] +\frac{\delta}{2}.
    \end{align*}
    Hence, to show Inequalities~\ref{eq:privacy-insertion-neighboring-t+1-to-t} and \ref{eq:privacy-insertion-neighboring-t-to-t+1} and complete the proof, it suffices to prove
    \begin{equation}\label{eq:ins-neighbor-t+1-to-t}
    \begin{aligned}
        &\sum_{z\geq - 2\ln(1/\delta)/\eps}\Pr\!\big[\lceil U\rceil = z\big]\cdot \Pr[\wrp(\sigma)\in \calY \mid \lceil U\rceil = z] + \frac{\delta}{2}\\
        &\leq e^{\eps}\SPACESUM\sum_{z\geq - 2\ln(1/\delta)/\eps+1}\!\!\!\Pr\!\big[\lceil U\rceil = z\big]\cdot \Pr[\wrp(\sigma')\in \calY \mid \lceil U\rceil = z] + \delta,
    \end{aligned}
    \end{equation}
    and 
    \begin{equation}\label{eq:ins-neighbor-t-to-t+1}
    \begin{aligned}
        &\sum_{z\geq - 2\ln(1/\delta)/\eps}\Pr\!\big[\lceil U\rceil = z\big]\cdot \Pr[\wrp(\sigma')\in \calY \mid \lceil U\rceil = z] + \frac{\delta}{2}\\
        &\leq e^{\eps}\SPACESUM\sum_{z\geq - 2\ln(1/\delta)/\eps-1}\!\!\!\Pr\!\big[\lceil U\rceil = z\big]\cdot \Pr[\wrp(\sigma)\in \calY \mid \lceil U\rceil = z] + \delta.
    \end{aligned}
    \end{equation}
    Since $U\sim\Lap(2/\eps)$, the Laplace distribution satisfies
    \[
    \Pr\!\big[\lceil U\rceil = z\big]\ge e^{-\frac{\eps}{2}}\Pr\!\big[\lceil U\rceil = z+1\big],
    \]
    and
    \[
    \Pr\!\big[\lceil U\rceil = z+1\big]\ge e^{-\frac{\eps}{2}}\Pr\!\big[\lceil U\rceil = z\big].
    \]
    Therefore, to prove Inequalities \ref{eq:ins-neighbor-t+1-to-t} and \ref{eq:ins-neighbor-t-to-t+1}, it suffices to show
    \begin{align*}
        &\sum_{z\geq - 2\ln(1/\delta)/\eps}\Pr\!\big[\lceil U\rceil = z\big]\cdot \Pr[\wrp(\sigma)\in \calY \mid \lceil U\rceil = z]\\ 
        &\leq e^{\eps/2}\SPACESUM\sum_{z\geq - 2\ln(1/\delta)/\eps+1}\Pr\!\big[\lceil U\rceil = z-1\big]\cdot \Pr[\wrp(\sigma')\in \calY \mid \lceil U\rceil = z] + \frac{\delta}{2},
    \end{align*}
    and 
    \begin{align*}
        &\sum_{z\geq - 2\ln(1/\delta)/\eps}\Pr\!\big[\lceil U\rceil = z\big]\cdot \Pr[\wrp(\sigma')\in \calY \mid \lceil U\rceil = z]\\
        &\leq e^{\eps/2}\SPACESUM\sum_{z\geq - 2\ln(1/\delta)/\eps-1}\Pr\!\big[\lceil U\rceil = z+1\big]\cdot \Pr[\wrp(\sigma)\in \calY \mid \lceil U\rceil = z] + \frac{\delta}{2}.
    \end{align*}
    By shifting the variable $z$ in the right-hand sides, we must equivalently show 
    \begin{equation}\label{eq:semi-final-inst-neighbor-1}
    \begin{aligned}
        &\sum_{z\geq - 2\ln(1/\delta)/\eps}\Pr\!\big[\lceil U\rceil = z\big]\cdot \Pr[\wrp(\sigma)\in \calY \mid \lceil U\rceil = z]\\ 
        &\leq e^{\eps/2}\SPACESUM\sum_{z\geq - 2\ln(1/\delta)/\eps}\!\!\!\!\!\Pr\!\big[\lceil U\rceil \!=\! z\big]\!\cdot\! \Pr[\wrp(\sigma')\!\in\! \calY \!\mid\! \lceil U\rceil \!=\! z\!+\!1] \!+\! \frac{\delta}{2},
        \end{aligned}
    \end{equation}
    and 
    \begin{equation}\label{eq:semi-final-inst-neighbor-2}
    \begin{aligned}
        &\sum_{z\geq - 2\ln(1/\delta)/\eps}\Pr\!\big[\lceil U\rceil = z\big]\cdot \Pr[\wrp(\sigma')\in \calY \mid \lceil U\rceil = z]\\
        &\leq e^{\eps/2}\SPACESUM\sum_{z\geq - 2\ln(1/\delta)/\eps}\!\!\!\!\Pr\!\big[\lceil U\rceil \!=\! z\big]\!\cdot\! \Pr[\wrp(\sigma)\!\in\! \calY \!\mid\! \lceil U\rceil \!=\! z\!-\!1] \!+\! \frac{\delta}{2}.
        \end{aligned}
    \end{equation}
    We will show that for every $z\geq - 2\ln(1/\delta)/\eps$, the following inequalities hold:
    \begin{equation}\label{eq:final-inst-neighbor-1}
    \begin{aligned}
        \Pr[\wrp(\sigma)\in \calY \mid \lceil U\rceil = z]
        \leq e^{\eps/2} \Pr[\wrp(\sigma')\in \calY \mid \lceil U\rceil = z+1] + \frac{\delta}{2},
    \end{aligned}
    \end{equation}
    and 
    \begin{equation}\label{eq:final-inst-neighbor-2}
    \begin{aligned}
        \Pr[\wrp(\sigma')\in \calY \mid \lceil U\rceil = z]
        \leq e^{\eps/2} \Pr[\wrp(\sigma)\in \calY \mid \lceil U\rceil = z-1] + \frac{\delta}{2}.
    \end{aligned}
    \end{equation}
    Multiplying Inequalities~\ref{eq:final-inst-neighbor-1} and \ref{eq:final-inst-neighbor-2} by $\Pr\!\big[\lceil U\rceil = z\big]$, summing over all integers $z\geq - 2\ln(1/\delta)/\eps$, and applying the fact that 
    $$\sum_{z\geq - 2\ln(1/\delta)/\eps} \Pr\!\big[\lceil U\rceil = z\big]\cdot \frac{\delta}{2} = \Pr[\lceil U\rceil \geq - 2\ln(1/\delta)/\eps]\cdot \frac{\delta}{2}< \frac{\delta}{2}$$
    yield Inequalities \ref{eq:semi-final-inst-neighbor-1} and \ref{eq:semi-final-inst-neighbor-2}.
 
    It remains to show Inequalities~\ref{eq:final-inst-neighbor-1} and \ref{eq:final-inst-neighbor-2} hold. Fix any $z\geq - 2\ln(1/\delta)/\eps$. For every $\lceil U\rceil\in \{z-1, z, z+1\}$, we have $\lceil U\rceil\geq - 2\ln(1/\delta)/\eps - 1$, and thus the padding size $C$ satisfies
    $$C=\max\left\{0, 1 + 2\ln(1/\delta)/\eps + \lceil U \rceil\right\} = 1 +  2\ln(1/\delta)/\eps + \lceil U \rceil.$$
    Consider Inequality~\ref{eq:final-inst-neighbor-1}. Conditioned on $\lceil U\rceil = z$ and given $\sigma=(x_1,\dots,x_{T+1})$ as input, $\wrp$ feeds the following stream to $\ecc$: 
    \[
    \tau = (x_1,\dots,x_{T+1}) \cdot \bot^c,
    \qquad
    c = 1 + \frac{2\ln(1/\delta)}{\eps} + z,
    \]
    Moreover, conditioned on $\lceil U\rceil = z+1$ and given $\sigma'=(x_1', \dots, x_{T})$ as input, $\wrp$ feeds the following stream to $\ecc$: 
    \[
    \tau' = (x'_1,\dots,x'_T)\cdot \bot^{c'},
    \qquad
    c' = 1 + \frac{2\ln(1/\delta)}{\eps} + z + 1.
    \]
    Since $T=T'+1$, by Definition~\ref{def:modified-insertion-neighboring}, we have $x_t = x_{t-1}'$ for all $t \in \{i+1, \dots, T+1\}$. Since $c'\geq 1$, we also know that $x_{T+1}'=\bot$. Comparing this with Definition~\ref{def:queue-neighboring} implies that the padded sequences $\tau$ and $\tau'$ are edit neighbors. Therefore, by Theorem~\ref{thm:ecc-privacy} and the choice of privacy parameters for $\ecc$, the distributions of the output sequences of $\ecc$ on $\tau$ and $\tau'$ are $(\eps/2, \delta/2)$-indistinguishable, which immediately implies Inequality~\ref{eq:final-inst-neighbor-1}. The proof of Inequality~\ref{eq:final-inst-neighbor-2} follows by an identical argument.
\end{proof}

\begin{lemma}\label{lem:modified-insertion-neighboring-accuracy}
    Let $\wrp$ be the wrapping mechanism described above with privacy parameters $\eps>0$ and $0<\delta\leq 1$, constructed using a continual counter $\ecc$ that for any given input stream $x'$, with probability $1-\beta$ for all $t\geq1$, incurs additive error $E_\ecc(x',t)$ at time-step $t$. Let $T\in\N$ and $\sigma=(x_1,\dots,x_{T})\in [0,1]^{T}$. Then, with probability at least $1-\beta-\delta/2$, the following conditions hold:
    \begin{enumerate}
        \item The mechanism $\wrp$ generates $T^*$ outputs $(y_1, \dots, y_{T^*})$, where $T^*\in \{T, T+1, \dots, T+1+ 4\lceil\ln(1/\delta)/\eps\rceil\}$.
        \item For $t\in [T^*]$, define $s_t=\sum_{i=1}^t x_i$, where $x_i=0$ for $i\in \{T+1, \dots, T^*\}$. For all $t\in [T^*]$, $|s_t-y_t|\leq E_\ecc(\sigma\cdot \bot^C, t)$.
    \end{enumerate}
\end{lemma}

\begin{proof}
    \begin{enumerate}
        \item From the Laplace tail bound on the random variable $U$, we have that with probability $1-\delta/2$, $U \leq \lceil 2 \ln (1/\delta)/\eps\rceil$. It follows that when $\wrp$ encounters $\$$, it sets $C \leq 1 + 4\lceil \ln (1/\delta)/\eps\rceil$ with probability $1-\delta/2$. The claim follows.
        \item We see that on the first $T$ steps of the stream $\sigma \cdot \bot^C$, the output of $\wrp$ is identical to the output of $\ecc$ on $\sigma$. For $t\in [T+1,T^*]$, we see that $s_t = s_T$, and from the accuracy guarantee of $\ecc$, with probability $1-\beta$, for all $t\geq1$, $|y_t - s_t| \leq E_\ecc(\sigma\cdot \bot^C, t)$.
    \end{enumerate}
\end{proof}

\end{document}